\documentclass[11pt]{article}

\usepackage[margin=1in]{geometry}
\usepackage{amsthm}

\usepackage{amsmath,amssymb,mathrsfs,enumerate,bbm,array}
\usepackage{mathtools}
\usepackage{graphicx,xcolor,url}
\usepackage{bm,algorithmic,algorithm2e}

\newcommand{\rev}[1]{\textcolor{black}{#1}}

\renewcommand{\P}{\mathbb{P}}
\newcommand{\Q}{\mathbb{Q}}
\newcommand{\E}{\mathbb{E}}
\newcommand{\F}{\mathcal{F}}
\newcommand{\HH}{\mathcal{H}}
\newcommand{\OO}{\mathcal{O}}
\newcommand{\s}{\sigma}
\newcommand{\1}{\mathbbm{1}}

\newcommand{\eps}{\epsilon}

\newcommand{\R}{\mathbb{R}}
\renewcommand{\a}{\alpha}
\newcommand{\cQ}{\mathfrak{q}_\alpha}

\newcommand{\cD}{\mathcal{D}}
\newcommand{\bB}{\mathbf{B}}
\newcommand{\bC}{\mathbf{C}}
\newcommand{\cZ}{\mathcal{Z}}
\newcommand{\cH}{\mathcal{H}}
\newcommand{\cR}{\mathfrak{R}}

\newcommand{\blue}{} 
\newcommand{\cN}{\mathcal{N}}

\DeclareMathOperator\logit{logit}

\DeclareMathOperator\VaR{VaR}
\DeclareMathOperator\TVaR{TVaR}
\DeclareMathOperator\var{\mathbb{V}ar}

\newtheorem{lemma}{Lemma}

\theoremstyle{remark}
\newtheorem*{remark}{Remark}

\title{Sequential Design and Spatial Modeling for Portfolio Tail Risk Measurement}

\author{Jimmy Risk\footnotemark[2]
\and Michael Ludkovski\footnotemark[1]}

\begin{document}
\maketitle

\footnotetext[1]{Department of Statistics and Applied Probability, University of California, Santa Barbara, 93106-3110 USA. Partial support by NSF DMS-1521743 is gratefully acknowledged. {ludkovski@pstat.ucsb.edu}}

\footnotetext[2]{Department of Mathematics and Statistics, Cal Poly Pomona {jrisk@cpp.edu} }

\renewcommand{\thefootnote}{\arabic{footnote}}
\begin{abstract}
We consider calculation of capital requirements when the underlying economic scenarios are determined by simulatable risk factors. In the respective nested simulation framework, the goal is to estimate portfolio tail risk, quantified via $\VaR$ or $\TVaR$ of a given collection of future economic scenarios representing factor levels at the risk horizon. Traditionally, evaluating portfolio losses of an outer scenario is done by computing a conditional expectation via inner-level Monte Carlo and is  computationally expensive. We introduce several inter-related machine learning techniques to speed up this computation, in particular by properly accounting for the simulation noise. Our main workhorse is an advanced Gaussian Process (GP) regression approach which uses nonparametric spatial modeling to efficiently learn the relationship between the stochastic factors defining scenarios and corresponding portfolio value. Leveraging this emulator, we develop sequential algorithms that adaptively allocate inner simulation budgets to target the quantile region. The GP framework also yields better uncertainty quantification for the resulting $\VaR/\TVaR$ estimators that reduces bias and variance compared to existing methods. We illustrate the proposed strategies with two case-studies in two and six dimensions.
\end{abstract}




%
%
%
%
%
%

\section{Introduction}

The latest insurance and financial regulations mandate estimation of quantile-based portfolio risk measures. The
Solvency II \cite{christiansen2014fundamental} framework calls for computing the $99.5\%$-level \emph{value-at-risk} ($\VaR$) for a 1-year horizon, while the Basel III regulations in banking \cite{basel2013fundamental} require to report the related \emph{tail value-at-risk} ($\TVaR$). In practice, these quantities are calculated by first generating a representative set $\cZ$ of future economic scenarios and then evaluating the empirical loss quantile based on $\cZ$. However, due to the underlying cashflow and valuation complexity, directly computing future portfolio value is usually not feasible and instead approximations are used. This is done for each scenario by a Monte Carlo evaluation of the conditional expectation that probabilistically defines portfolio value, leading to a \emph{nested} simulation problem.

In traditional nested simulation, each outer scenario is treated independently, so the respective portfolio loss is approximated through the sample average across the inner simulations that define the resulting cashflows.
The scenarios $\cZ$ are typically realizations of underlying stochastic factors  or risk drivers $(Z_t)$, such as interest rate factors, equity price factors, mortality factors, macroeconomic factors, et cetera. Thus, we can identify each $z^n \in \cZ$ as the value of $Z_T$ in scenario $\omega^n$ at the risk horizon $T$; the portfolio value $f(z)$ is identified with the expected cashflow $Y(\cdot)$ (which depends on the future path $(Z_t)_{t \geq T}$ beyond $T$) conditional on $Z_T = z$:
\begin{equation}\label{eq:conditionalexp}
f(z) \doteq \mathbb{E}\left[Y\!\left( (Z_t)_{t \geq T}\right) \mid Z_T=z\right].
\end{equation}
Note that we assume that $Z$ is Markovian which is essentially always the case in the practical context; if necessary $Z$ is augmented to make it Markov. The more important assumption is that $f(z)$ is not available in closed form, so for a given $z$ it must be approximated via the \emph{inner} step of the nested procedure.

Our goal is to efficiently estimate $\VaR_\a$ and/or $\TVaR_\a$ of $f(Z_T)$ when the risk drivers follow a simulatable Markov process. Towards this end we build upon two fundamental strategies that have been developed for nested simulation. The first idea is to
improve  the {inner} simulations by adaptively allocating the \emph{simulation budget} to scenarios with larger portfolio losses. Indeed with the mentioned quantile levels, $\approx 99\%$ of the outer scenarios are irrelevant from the point of view of $\VaR/\TVaR$ computation: only the tail scenarios matter. The second idea is to exploit the \emph{spatial} structure of $\cZ$ coming from $(Z_t)$: nearby inputs $z, z'$ should produce similar values $f(z), f(z')$. Therefore, we can improve estimation of $f(z)$ by borrowing information from inner simulations at nearby outer scenarios. This brings a regression perspective that converts the local strategy of computing a sample average into a global goal of approximating the function $z \mapsto f(z)$.

In this article we combine and greatly extend these approaches by embedding them in the framework of \emph{statistical emulation}. Emulation treats $f$ as an unknown function and seeks to produce a functional estimate $\hat{f}$ that optimizes a given objective criterion.  For our emulation paradigm we propose
%
\emph{Gaussian process (GP) regression} \cite{rasmussen2006gaussian,santner2013design}, or \emph{kriging}, which has become a leading choice in the machine learning community. GP's bring a probabilistic (i.e.~Bayesian for a statistician) perspective, so that inference of $f$ is viewed as conditioning on the observed simulation output, and estimating a risk measure as computing a posterior expectation of the corresponding random variable. This is arguably the natural framework in the Monte Carlo context where all outputs are intrinsically stochastic. GP's also bring a non-parametric regression approach, alleviating the question of finding good functional approximation spaces for the spatial regression component. \rev{Below we show that the (GP) emulation approach, in combination with tailored active learning techniques, can provide accurate estimates even with a small simulation budget, while maintaining a relatively low numerical overhead  (fitting, prediction, etc.). Thus, our main message is that significant efficiency can be squeezed from the proposed {statistical tools}, dramatically improving upon nested simulation.}

In the language of emulation, Equation \eqref{eq:conditionalexp} defines a black-box function which we wish to learn. The origins of emulations/GP regression are in analysis of \emph{deterministic} computer experiments where $f(z)$ can be (expensively) evaluated without any noise, see e.g.~\cite{forrester2008engineering,gramacy2008gaussian,marrel2008efficient} and the monographs by Santner et al.~\cite{santner2013design} and Rasmussen and Williams~\cite{rasmussen2006gaussian}. More recently, emulation was adapted to the Monte Carlo setting where $f(z)$ is only noisily observed: see e.g.~\cite{ankenman2010stochastic,binois2016practical,chen2012stochastic}. Also related is our earlier work~\cite{risk2016statistical} that studied approximate pricing of deferred longevity-linked contracts, requiring an average of Equation \eqref{eq:conditionalexp} over the whole distribution of $Z_T$.

For tail risk estimation, the learning objective is inherently different from the standard $L^2$ criterion, since the fitted $\hat{f}$ only needs to be accurate in the tail of $f(Z_T)$. Indeed, as mentioned 99\%+ of outer scenarios will turn out to be irrelevant. To our knowledge, there is little literature that directly addresses quantile/tail estimation in the Monte Carlo setting of~\eqref{eq:conditionalexp}. We mention the work of
Oakley \cite{oakley2004estimating} who introduced the use of GP's for quantile estimation of deterministic computer experiments and Liu and Staum \cite{liu2010stochastic} who pioneered stochastic kriging for estimating $\TVaR_\a$ in nested Monte Carlo of financial portfolios. Two further important papers are by Bauer et al.~\cite{bauer2012calculation} who elucidated the use of (parametric linear) regression for tail risk estimation, and Broadie et al.~\cite{broadie2011efficient} who proved some properties of such Least Squares Monte Carlo methods for learning $\VaR_\a$. More broadly, our localized objective resembles three inter-related formulations in the engineering reliability literature that target $f(z)$ in relation to a given level or threshold $L$:
\begin{itemize}
  \item Contour-finding: determining the set $\{z : f(z) \in (L-\varepsilon, L+\varepsilon)\}$, $\varepsilon$ small, cf.~\cite{labopin2016sequential,picheny2010adaptive};

  \item Inference of the excursion set $\{z: f(z) \leq L\}$ \cite{chevalier2014fast};

   \item Estimating the probability of failure, aka ``excursion volume'', $\P(\{z: f(z) \leq L\})$ \cite{bect2012sequential,chevalier2014kriginv}.
\end{itemize}
Note that the contour $\{ z: f(z) = L\}$ is the boundary of the level set $\{z : f(z) \leq L\}$, and the excursion volume is a weighted integral over the level set.
However, in all of the above (except \cite{labopin2016sequential}), the threshold $L$ is exogenously given, while in the Value-at-Risk context, the desired quantile $\VaR_\a$ is implicitly defined via $f$ itself, see \eqref{eq:VaRf} below. Moreover, all of \cite{bect2012sequential,chevalier2014fast,labopin2016sequential,picheny2010adaptive} dealt with deterministic noise-free experiments.

At the heart of emulation is design of experiments, i.e.~determining which simulations to run so as to \emph{learn} the input-output pairing $z \mapsto f(z)$ as quickly as possible. This naturally connects to adaptive allocation of inner simulations. Adaptivity is achieved by dividing the overall simulation into several \emph{stages} that gradually learn the shape of $f$. The initial stage uses a fraction of the simulation budget  to learn about $f(z)$ on the whole domain. The resulting search region targeting the objective is then refined through further stages. For example, the approach in \cite{oakley2004estimating} has two stages: in the second stage inner simulations are non-uniformly allocated to scenarios in the quantile region based on the initial inference about $f$.  Liu and Staum~\cite{liu2010stochastic} proposed a total of three stages;  during the second stage another fraction of the budget is allocated uniformly to scenarios in the estimated tail, and finally the remaining budget is dispersed (non-uniformly) among the tail scenarios to minimize the posterior variance of the $\TVaR_\a$ estimator.

Further, inspired by optimization objectives (finding the global maximum of $f$), there are \emph{sequential design} approaches which use a large number of stages (in principle adding just 1 simulation per stage). In particular,  \emph{stepwise uncertainty reduction} (SUR) methods define an acquisition function $\cH(\cdot)$ and greedily pick scenarios $z$ that maximize $\cH(z)$ conditional on current data $\mathcal{D}$. For example,  Picheny et al.~\cite{picheny2010adaptive} sequentially sample $f$ at scenarios according to a {targeted expected improvement} criterion. Another type of adaptive strategies does not rely on emulation/spatial modeling, rather utilizing the tools of what is known as ranking-and-selection (R\&S) \cite{ChenFu08,Frazier14}. However, R\&S procedures are generally based on doing a hypothesis test for $f(z^n) > L$  in parallel for each scenario $z^n$. Theoretically this allows to decouple the estimation problems, but such schemes are impractical if $L$ is itself unknown. 

In the $\VaR$ context, several aspects of the setting make the emulation problem statistically challenging. First, the typical number of outer scenarios $N$ is quite large, on the order of $N \in [10^4, 10^5]$ for practitioners. This renders standard emulation and R\&S strategies computationally heavy. Furthermore, the typical simulation budget $\cN$ is often quite small, $\cN \sim N$, which means that a brute-force uniform allocation approach can only sample each scenario a handful of times, leading to disastrous estimates. It further implies that practical allocation schemes must be quite \emph{aggressive} and the asymptotic guarantees available in the R\&S literature are not applicable. At the same time $\cN$ is large by emulation standards which treat evaluations of $f$ as extremely expensive (typical budget is $\ll 500$ simulations) and hence the associated algorithms often carry unacceptable computational overhead.  Second, the outer scenarios are usually treated by practitioners as a fixed object (namely coming from an exogenous Economic Scenario Generator aka ESG). In other words, the scenario space  $\cZ$ is taken to be discrete; one cannot add (or subtract) further scenarios like in the strategy of Broadie et al.~\cite{broadie2011efficient}. Neither can one employ continuous search methods that are popular in the simulation optimization literature. Third, the nature of the underlying cashflows makes the simulation noise in the cashflows $Y$ highly non-Gaussian. It is typically skewed (because many cashflows have embedded optionality and hence the corresponding distribution has point masses) and with low signal-to-noise ratio. The latter implies that cross-scenario information borrowing is crucial to maximize accuracy. Fourth, the portfolio losses are highly inhomogeneous, i.e.~$\eps$ is \emph{heteroskedastic} which strongly affects the inner simulation allocations. Fifth, because the objective function is implicitly defined in terms of $f$, constructing an appropriate estimator, and especially quantifying its accuracy (aka standard error) is nontrivial in its own right.

The framework that we propose overcomes all of the above challenges and consists of three key pieces. First, we connect to the burgeoning literature on level-set estimation for deterministic computer experiments, tailoring the recent successes of the SUR techniques  \cite{chevalier2014kriginv} to the $\VaR/\TVaR$ problem. \rev{As far as we know this is the first link between these disparate literatures}, which we see as an invitation towards closer marriage of machine learning and simulation-based risk measurement. Second, we work with a state-of-the-art GP emulator. By their nature, GPs offer rich uncertainty quantification properties, in particular many analytic formulas for active learning criteria that are used in guiding the simulation allocation.  We take this a step further, employing an advanced GP methodology \cite{hetGP,binois2016practical} that is customized for the Monte Carlo simulation context. Third, we {take advantage} of the \emph{discrete scenario set} which intrinsically calls for a replicated design. In turn, replication yields (i) improved noise properties that minimize non-Gaussianity; (ii) ability to simultaneously learn the mean response $f(\cdot)$ and the conditional simulation variance $\tau^2(\cdot)$ to handle heteroskedasticity; (iii) reduced model overhead; (iv) convenient GP implementation. The symbiotic relationship of replication and kriging was already observed in \cite{{ankenman2010stochastic}}; the \texttt{hetGP} \cite{hetGP} library that we use gracefully handles {all} aspects of (i)-(iv).

Rather than proposing a single specialized algorithm, we present a  general framework comprising several modules. By adjusting these ``moving parts'', the algorithm can (a) evaluate different risk-measures (we illustrate with both VaR and TVaR); (b) use different emulator codes; (c) switch between different ways of running the sequential budget allocation such as the initialization phase; number of sequential stages; termination criterion; batch-size and so on; (d) rely on different VaR estimators. To illustrate the above choices we present extensive numerical illustrations that compare the impact of different modules. In two case studies, we compare it with benchmarks, such as the algorithm in Liu and Staum~\cite{liu2010stochastic}.

Our statistical learning perspective implies a preference for \emph{sequential} algorithms that employ the feedback loop of simulate--estimate--assess--simulate--... Compared to classical one-/two-stage designs, the fully sequential strategies offer several attractive features. First, they internalize the uncertainty quantification that is offered by the emulator and therefore lead to a more automated ``artificial intelligence'' implementation. Second, they can be run online, i.e.~on a server or in the cloud with interim results always available for inspection, while the algorithm proceeds to refine its accuracy. Thus, the user no longer has to specify the simulation budget (i.e.~number of inner simulations one is willing to run) a priori, but instead interacts with the Monte Carlo on-the-fly. Third, sequential methods  are a starting point for more general re-use of simulations, for example for periodic re-computation of the whole problem as business time progresses (usually VaR calculations are done on a weekly or monthly cycle).

\rev{Our overall contribution is to remedy limitations of current $\VaR$/$\TVaR$ estimation algorithms through ideas from GP emulators and related sequential design problems.  In the context of $\TVaR$, arguably the closest approach, and the one that we are to a large extent indebted to, is the algorithm in Liu and Staum~\cite{liu2010stochastic}. Relative to~\cite{liu2010stochastic} we may list the following 4 important improvements. First, we generalize their 3-stage strategy to a fully sequential $k$-stage procedure, dubbed SV-GP below. In particular, we eliminate the need for their intermediate second stage which was performing a conservative ``exploration'' (and required finetuning), and also extend their algorithm to handle non-constant replication amounts. Second, we employ the latest \texttt{hetGP} emulators which significantly improve the emulation in low-budget environments, in particular by lowering the bias in learning $f(x)$. Third, we provide an alternative class of allocation strategies relying on Active Learning heuristics. Finally, while ~\cite{liu2010stochastic} only treated $\TVaR$, we work (and compare the implementation nuances) with both $\VaR$ and $\TVaR$.}

The paper is organized as follows: in Section \ref{sec:objective} we discuss the emulation objective and portfolio tail risk, as well as $\VaR_\a$ estimators.  Section \ref{sec:StochasticKriging} formally introduces the GP model and its mathematical details.  Next, Section \ref{sec:sequential} develops the sequential design criteria for our problem.  Section \ref{sec:algorithm} provides the full implementation.  Finally, we apply the algorithm to two case studies:  Section \ref{sec:casestudy-BS} involves a two-factor model, where we compare sequential $\VaR/\TVaR$ strategies to several benchmarks, along with a look at various GP versions; Section \ref{sec:casestudy2} investigates performance in higher dimensions where the state process is six-dimensional and the simulator is much more expensive.

\section{Objective}\label{sec:objective}
Given a probability space $(\Omega, \F, \P)$, we consider a stochastic system with Markov state process $Z=(Z_t)$. Both continuous- and discrete-time models can be treated. Typically, $Z$ is a multivariate stochastic process based on either a stochastic differential equation ($t \in \mathbb{R}_+$) or time-series ARIMA framework ($t \in \mathbb{N}$). Based on the realizations of $Z,$ the modeler needs to assess the resulting capital $f(Z_T)$ at an intermediate horizon $T$.  The portfolio value $f(Z_T)$ is computed as the conditional expected value of discounted cashflows  $Y_t$ given $\mathcal{F}_T,$ the information at time $T$. The cashflow $Y_t(z_\cdot)$ at date $t$ could depend on the whole path $Z_{[T,t]}$ of the stochastic factor from $T$ to $t$ (but not on $z_{[0,T]}$). To compute their net present value $Y$, we discount at a constant risk-free rate $\beta$
\begin{equation}\label{eq:Y-trajectories}
Y \doteq \sum_{t=T}^\infty e^{-\beta(t-T)}Y_t(Z_{[T,t]}).
\end{equation}

Since $Z$ is Markov, we can write the conditional expectation of $Y$ given $\mathcal{F}_T$ as a function $f(z)$ defined in \eqref{eq:conditionalexp}, and similarly will henceforth use the notation $Y(z)$. We assume that $f(\cdot)$ does not possess any simple/closed-form functional form. However, since $f(z)$ is a conditional expectation, it can be sampled using a simulator, i.e.~the modeler has access to an engine that can generate independent, identically distributed trajectories of $(Z_t)$, $t \ge  T$.

Fixing a quantile $\alpha$ (most commonly $\a=0.005$ corresponding to 99.5\%-$\VaR$ level mentioned in the Introduction),
the \emph{capital requirements} based on $\VaR_\a$ and $\TVaR_\a$ are respectively the threshold such that the probability of loss exceeding $\VaR_\a$ is\footnote{Throughout we view $f$ as portfolio value, i.e.~more is better and the tail risk concerns the left tail where $f$ is most negative.} $\a$, i.e.
\begin{equation}\label{eq:VaRf}
\VaR_{\a}(f(Z_T)) \doteq \sup\{x : \P(f(Z_T) \leq x) \leq \a\},
\end{equation}
and the $\a-$tail average,
\begin{equation}\label{eq:TVaRf}
\TVaR_{\a}(f(Z_T)) \doteq \E\bigl[f(Z_T) \mid f(Z_T) \leq \VaR_{\a}(f(Z_T))\bigr].
\end{equation}

As mentioned, rather than working with the theoretical distribution of $Z_T$, we consider a given scenario set $\cZ$ of size $N$ which is assumed to be fixed for the remainder of the presentation. Typically the outer scenarios are provided by an Economic Scenario Generator (ESG) that calibrates the dynamics of $(Z_t)$ to the historical data, i.e.~it is under the physical measure $\mathbb{P}$. The inner simulations are based on a mark-to-market law, i.e.~it is under the risk-neutral measure $\mathbb{Q}$. Therefore, the evolution of $(Z_t)$ is different on $[0,T]$ and on $[T,\infty)$. This is one reason for treating $\cZ$ as fixed, allowing us to eschew discussion of the precise mechanism of the ESG/outer simulator.

Given the loss scenarios $\{ z^1, z^2, \ldots, z^N\} \doteq \cZ$, computing $\VaR_a$ now translates into finding the $\a N$ (assumed to be integer) order statistic  $f^{(\a N)}$ of $f^{1:N}$, where we use the shorthand $f^n \doteq f(z^n)$, and the superscripts $(n)$ refer to the $n$th ordered value (in increasing order) of $f^{1:N}$.  More generally, we consider risk measures $R$ for $\cZ$ defined through weights: given a collection of losses $f^{1:N}$ let
\begin{equation}
R \doteq \sum_{n=1}^N w^{n} f^n, \label{eq:riskmeasure}
\end{equation}
where $\sum_{n=1}^N w^{n} =1.$  For $\VaR_{\a}$ and $\TVaR_{\a}$  the weights are respectively
\begin{align}
 w^{n, \VaR} &= \1_{\{n : f^{n} = f^{(\a N)}\} };\label{eq:VaRweight}\\
w^{n, \TVaR} &= \rev{\frac{1}{\a N} \cdot} \1_{ \{n : f^{n} \leq f^{(\a N)}\} }. \label{eq:TVaRweight}
\end{align}

Another way to think about the tail risk is through the corresponding \emph{tail regions}
\begin{align}
\cR^{\VaR} \doteq \{z^n : f^n = f^{(\alpha N)}\},\qquad \text{and}\quad
\cR^{\TVaR} \doteq \{z^n : f^n \leq f^{(\alpha N)}\}.
\end{align}
Finally, we also introduce the notation $\cQ$ to denote the exact $\a$-quantile scenario (i.e.~the one in $\cR^{\VaR}$): $f(\cQ) = f^{(\a N)}$.

\subsection{Tail Risk Estimation with Nested Simulation}

The standard Monte Carlo approach is based on nested simulation. Starting from $\cZ,$ for each $n$, $f(z^{n})$ is approximated by an inner empirical average
\begin{equation}\label{eq:N-in2}
f(z^{n}) \simeq \bar{y}^{n} \doteq \frac{1}{r^n} \sum_{i=1}^{r^n} y^{n,i}, \qquad n = 1, \ldots, N,
\end{equation}
where $y^{n,i}$ is the present value of loss from \eqref{eq:Y-trajectories} based on independent replications $Y_\cdot^{n,i}( z^{n}_\cdot),$ \rev{$i \in \{1,\ldots, r^n\}$}  evaluated through the trajectories $z^{n,i}_\cdot$ of $(Z_t)_{t \geq T} | Z_T = z^n$.
Equations \eqref{eq:VaRf} and \eqref{eq:TVaRf} are  then estimated by a sample quantile and tail average respectively. Specifically, we sort the estimated scenario losses $(\bar{y}^n)_{n=1}^N$ in increasing order $\bar{y}^{(1)} \le \bar{y}^{(2)} \le \ldots \le \bar{y}^{(N)}$ and
 take $\hat{R}^{SA,\VaR} \doteq \bar{y}^{(\a N)}$ and $\hat{R}^{SA,\TVaR} \doteq \frac{1}{\a N} \sum_{i=1}^{\a N} \bar{y}^{(i)}$.

 In the simplest setting, the inner budget is constant across scenarios, $r^n = \bar{r}$ leading to total simulation budget of $\OO(N\cdot \bar{r})$. This quickly becomes expensive: for example, a budget of $\bar{r}=10^3$ and $N=10^3$ scenarios requires $10^6$ total simulations. Moreover, the empirical estimator $\hat{R}^{SA, \VaR}$ is biased, see  \rev{Kim and Hardy}~\cite{kim2007quantifying} (who point out that one cannot quantify which direction the bias lies without knowing some details about $f$ and the distribution of $Z_T$) and Gordy and Juneja~\cite{gordy2010nested}. These papers discuss ways of reducing bias, through bootstrap and jackknife, respectively. Gordy and Juneja~\cite{gordy2010nested} also provide consequent optimal strategies for choosing $N$ and $\bar{r}$ for a fixed budget $N \bar{r} = \cN$.

As another way to  alleviate the bias in $\hat{R}^{SA,\VaR}$, one may construct
several modified versions of $\VaR_{\a}$ estimators that stem from weighted averages of nearby order statistics.  Called $L-$estimators \cite{sheather1990kernel}, they offer robustness compared to a single sample order statistic. Effectively, such estimators modify the weights $w^n$ defining \eqref{eq:riskmeasure} to account for the uncertainty in the marginal estimators $\hat{f}(z^n)$. A well known construction is by Harrell and Davis~\cite{harrell1982new}, where the weights for a $\VaR_\a$ estimator are chosen as $R^{HD,\VaR} = \sum_n \tilde{w}^{(n)} \hat{f}^{(n)}$ with
\begin{equation}\label{eq:HDweights}
\tilde{w}^{(n)} = \int_{(n-1)/N}^{n/N} \mathrm{Beta}(t; (N+1)\a N, (N+1)(1-\a N) ) dt,
\end{equation}
where $\mathrm{Beta}(x; a, b)$ is the Beta-distribution density function with shape parameters $(a,b)$ and
%
the index $(n)$ refers to the order statistics of $\hat{f}$ (e.g. the $\bar{y}^n, n=1, \ldots, N$
if one uses the empirical plug-in approach). The weights $\tilde{w}^{(n)}$ are largest near $n=\alpha N$ and taper quickly on either side, smoothing out the binary 0/1 nature of \eqref{eq:VaRweight}. 
We refer to \cite{sfakianakis2008new,sheather1990kernel} for further details about the Harrell-Davis weights and their performance relative to other $L-$estimators. While the general conclusion is that no one estimator performs best, in our experiments we find a clear preference for \eqref{eq:HDweights} over the original \eqref{eq:VaRweight}  both in terms of smaller bias and lower standard error. These concerns do not extend to the $\TVaR$ case, since in our experience the empirical estimator $\hat{R}^{\TVaR}_\a = \frac{1}{\a N} \sum_{n=1}^{\a N} \hat{f}^{(n)}$ is already robust enough against mis-specification of a few borderline scenarios. Theoretically, additional robustness could be achieved via smoothing the cut-off in $\hat{w}^{n,\TVaR}$ near $\hat{f}^{(\a N)}$ to create a structure similar to \eqref{eq:HDweights}.

\section{Gaussian Process Emulation}\label{sec:StochasticKriging}
To construct thrifty schemes for approximating \eqref{eq:VaRf} and \eqref{eq:TVaRf}, we replace the inner step of repeatedly evaluating $f(z)$ by a \emph{surrogate model} $\hat{f}$ for $f$. This emulation framework generates a fitted $\hat{f}$ by solving regression equations over a training dataset. In contrast to classical regression, the modeler is in charge of the experimental design (building the training samples) which are viewed as arriving sequentially.  Particularly for tail risk, we seek a procedure that identifies the tail scenarios that matter for \eqref{eq:VaRf}-\eqref{eq:TVaRf}. We note that emulation offers a statistical perspective on approximating $f$. An alternative, commonly employed by practitioners, is to construct analytic approximations, for instance by integrating out some of the randomness in $Z_{[T,\infty)}$ to obtain a closed-form expression for $f(z)$. However, the quality of such formulas is hard to judge and they often require customized derivations. In contradistinction, the surrogate $\hat{f}$ directly smoothes out the Monte Carlo noise from nearby scenarios and is constructed to be asymptotically consistent with the true $f$ as the budget increases.

Formally, the statistical problem of emulation deals with a sampler
\begin{align}\label{eq:oracle3}
  Y(z) = f(z) + \eps(z),
\end{align}
where we identify $f(\cdot)$ with the unknown \emph{response surface} and $\eps(\cdot)$ is the sampling noise, with variance $\tau^2(z)$, assumed to be independent and identically distributed across different calls to the sampler. Specifically, the noise is taken to be Gaussian $\eps(z) \sim \mathcal{N}(0, \tau^2(z)),$ but scenario-dependent. Emulation now involves the (i) experimental design step of proposing a design $\mathcal{D}_k$ at each stage $k$ of the procedure that forms the training dataset, and (ii) a learning procedure that uses the collected outputs $\mathcal{D}_k = (z^{n},\mathbf{y}^{n}_k)_{n=1}^{N}$, with $\mathbf{y}^{n}_k = \{y^{n,1}_k, \ldots, y^{n,{r^{n}_k}}_k\}$ being a collection of $r^n_k$ realizations of \eqref{eq:oracle3} at $z^{n}$, to construct a fitted response surface $\hat{f}(\cdot)$ which is based on the conditional law of $f$ over the appropriate functional space, $\mathcal{L}(f |\mathcal{D}_k)$.  Here, we consider the case where $\hat{f}$ is fitted sequentially based on a multi-step procedure, and the subscript $k$ denotes the step counter, so that $\mathcal{D}_k \subseteq \mathcal{D}_{k+1}$.

\subsection{GP Posterior} \label{sec:SimpleKriging2}
A GP  surrogate assumes that $f$ in \eqref{eq:oracle3} has the form
\begin{equation}\label{eq:kriging3}
f(z) = \mu(z)+X(z),
\end{equation}
where $\mu : \R^d \rightarrow \R$ is a trend function, and $X$ is a mean-zero square-integrable Gaussian process with covariance kernel $C(x, x')$. The role of $C$ is to generate the reproducing kernel Hilbert space  $\HH_C$ which is the functional space that $X$ is assumed to belong to.

Since the noise $\eps(z)$ is also Gaussian, we have  that $\hat{f}_k(z) \equiv f(z) | \cD_k\sim \mathcal{N}( \mu(z) + m(z), s^2(z))$ has a Gaussian posterior, which reduces to computing the kriging mean $m(z)$ and kriging variance $s^2(z)$. One can then take $\mu(z)+m(z)$ as the point estimate of $f$ at scenario $z$. In turn, the kriging variance $s^2(z)$ offers a principled empirical estimate of model accuracy, quantifying the approximation quality. In particular, one can use $s^2(z)$ as the proxy for the mean-squared error (MSE) of $\hat{f}$ at $z$. 
The Gaussian process property implies that not only is the marginal posterior at $z$  Gaussian, but also the joint distribution of $(f(z^{'1}), \ldots, f(z^{'M}) ) | \cD_k$ is multivariate normal (MVN) for any $M$-tuple $z^{'1:M}$.

By considering the process $f(z)-\mu(z),$ we may assume without loss of generality that $f$ is statistically centered at zero. Denoting the sample average at each scenario $z^n$ by $\overline{y}^n_k = \frac{1}{r^n_k}\sum_{i=1}^{r^n_k} y^{n,i}_k$ as in Equation \eqref{eq:N-in2} and the overall collection as $\overline{\mathbf{y}}_k \doteq (\overline{y}^{1}_k, \ldots, \overline{y}^{N}_k)$ (\rev{note that here we assume that all outer scenarios have been already sampled, see further discussion on p.~14}), the resulting posterior mean and variance of $\hat{f}_k(z)$ follow from the MVN conditional equations \cite{roustant2012dicekriging}
\begin{equation}
\left\{ \begin{aligned}
m_k(z) &\doteq \mathbf{c}(z) (\bC+\bm{\Delta}_k)^{-1}\overline{\mathbf{y}}^T_k;\\
s^2_k(z) &\doteq C(z,z)-\mathbf{c}(z) (\bC+\bm{\Delta}_k)^{-1}\mathbf{c}(z)^T,
\end{aligned} \right. \label{eq:sk-meanvar}
\end{equation}
where $\mathbf{c}(z) = \left(C(z,z^{j})\right)_{1 \leq j \leq N}$, $\bC \doteq \left[C(z^{i},z^{j})\right]_{1 \leq i, j \leq N},$ $\bm{\Delta}_k$ is the $N \times N$ diagonal matrix with entries $\tau^2(z^{1})/r_k^1, \ldots, \tau^2(z^{N})/r_k^{N}$ and
$^T$ denotes transpose.  Note that the conditional variance $\tau^2(\cdot)$ is typically unknown, so applying \eqref{eq:sk-meanvar} requires replacing it with a further approximation $\hat{\tau}^2(\cdot)$, as discussed in Section \ref{sec:IntrinsicVariance}.

\begin{remark}
\rev{  For readers unfamiliar with emulators and surrogates, the idea is equivalent to regression: projecting the unknown $f$ onto a specified approximation space $\mathcal{H}$. The latter is ``non-parametric'' in the sense of being infinite-dimensional and dense in the class of continuous functions, so that for an input dataset of size $N$, there are $\mathcal{O}(N)$ degrees of freedom. In contrast to classical regression, in emulation there is no ``data'': the controller is also in charge of proposing scenarios $x$ where $f$ should be probed; because the outputs are stochastic it makes perfect sense to batch, i.e.~probe the same scenario multiple times. GP emulation treats $f$ as a realization of a random field, recasting regression as probabilistically conditioning this random field on the given input/output pairs. This is equivalent to taking $\mathcal{H}$ to be the RKHS generated by kernel $C$ and then applying the standard $L^2$-projection equations.}
\end{remark}

\subsection{Covariance kernels and hyperparameter estimation}
Fitting a GP emulator is equivalent to specifying the covariance function $C(\cdot, \cdot)$ and applying \eqref{eq:sk-meanvar}. The standard choice is a spatially stationary kernel,  $$C(z,z') \equiv c(z-z') =  \s^2 \prod_{j=1}^d  g( z_j-z'_j; {\theta}_j),$$
 reducing to the one-dimensional base kernel $g$. Below we use the Mat\'ern-5/2 kernel
\begin{equation}\label{eq:maternkernel2}
g(h; \theta) = \left(1+\frac{\sqrt{5}h}{\theta} + \frac{5 h^2}{3 \theta^2}\right) \exp\left(-\frac{\sqrt{5}h}{\theta}\right).
\end{equation}    The hyper-parameters ${\theta}_j$ are called characteristic length-scales and informally correspond to the distance between inputs over which the response $f$ can change significantly \cite[Ch 2]{rasmussen2006gaussian}. %
We utilize the \texttt{R} packages \texttt{DiceKriging} \cite{roustant2012dicekriging} and \texttt{hetGP} \cite{hetGP} that fit the hyper-parameters $\s, \theta_j$ by Maximum Likelihood. The latter is a nonlinear optimization problem and the packages differ in
how the search for the global maximum (initialization, gradient estimation, etc.) is conducted.

\subsection{Intrinsic Variance}\label{sec:IntrinsicVariance}
The simulation variance $\tau^2(z)$ is scenario-dependent and a crucial piece of fitting a surrogate. In a basic GP emulator, it is approximated by a local empirical variance, leveraging the replications offered by multiple inner simulations. This approach, known as Stochastic Kriging (SK) and introduced in Ankenman et al.~\cite{ankenman2010stochastic}, sets
\begin{equation}\label{eq:tauEstimate}
\hat{\tau}^2_k(z^n) \doteq \frac{1}{r^n_k -1}\sum_{i=1}^{r_k^n}(y_k^{n,i} - \bar{y}_k^n)^2.
\end{equation}
However, SK requires $r^n_k \ge 2$ and more realistically $r_n^k \gg 10$, as an insufficient number of inner simulations will give increasingly unreliable estimates of the conditional variance. This places a major restriction on the sequential algorithms which are then required to sample all relevant scenarios a minimal number of times. Moreover, $\hat{\tau}$ in \eqref{eq:tauEstimate} is only available a posteriori, so there is no natural way to predict simulation variance at a yet-to-be-sampled scenario (Ankenman et al.~\cite{ankenman2010stochastic} suggested to do a GP-based interpolation of $\hat{\tau}^2$ for that purpose).

To overcome this limitation, we use an extended GP model~\cite{binois2016practical} that jointly models the response mean $z\mapsto f(z)$ and (log)-variance $z \mapsto \log \tau^2(z)$. Specifically, an additional latent variable $\Lambda_n$ is introduced for each scenario $z^n \in \cD$. The estimated ${\tau}^2$'s are then viewed as \emph{spatially smoothed} versions of $\Lambda_n$, combining the above interpolation idea of SK with additional smoothing to account for the uncertainty about the true $\tau^2(z^n)$. The model jointly learns the hyperparameters defining the covariance kernel ${C}(\cdot, \cdot)$ of $f(\cdot)$, and all the $\Lambda_n$'s. While this increases the complexity of the underlying optimization problem that is solved during the likelihood maximization step, tractability is maintained thanks to analytic expressions for the respective likelihood function gradients. We refer to Binois et al.~\cite{binois2016practical} for the full details and utilize the resulting \texttt{hetGP} \cite{hetGP} package where
this procedure is efficiently implemented. Relative to SK, the \texttt{hetGP} approach has no lower-bound constraint on $r^n_k$'s and is especially beneficial at the early stages of the algorithm where  many of the $\hat{\tau}$'s are unstable. According to \cite{binois2016practical}, using \texttt{hetGP} improves performance even for homoskedastic models, simply through better handling of the replication, avoiding the pitfalls of \eqref{eq:tauEstimate}. \rev{In our experiments, \texttt{hetGP} performed much better, in particular lowering the bias in $m_{k}(\cdot)$ for small $k$, and hence improving the ``zooming in'' feature of our sequential design strategies.}

\subsection{Estimating Portfolio Risk}\label{sec:tailrisk-inGP}
After a surrogate for $f(\cdot)$ is constructed, there remains the problem of estimating
 the risk measure $R$ defined in \eqref{eq:riskmeasure}. This \emph{identification} problem is itself non-trivial, not least because we seek more than a point estimate $\hat{R}$.

 Taking a probabilistic point of view, $R$ as a random variable (defined through the random variable $f^{1:N}$) that carries its own  posterior distribution  given $\cD_k$.
 Thus, a GP-based estimate of $R$ is \begin{align}\label{eq:risk-posterior}
\hat{R}^{GP}_k \doteq \E[ R | \cD_k] = \sum_n \E[ w^n f(z^n)  | \cD_k].
\end{align}
Because the weights $w^n$ are defined implicitly in terms of order statistics of $f(z^{1:N})$, $\hat{R}^{GP}$ requires integrating against the joint distribution of $\hat{f}^{1:N}$. While the latter is multivariate Gaussian, there are no closed-form formulas for the probability that one coordinate $\hat{f}^n$ of a MVN distribution is a particular order statistic $\hat{f}^{(\a N)}$ (though see \cite{labopin2016sequential} who develop a related approximation). Nevertheless, the conditional expectation in \eqref{eq:risk-posterior} can be in principle numerically approximated by making draws from the posterior MVN law of $\hat{f}^{1:N}$, evaluating the resulting quantile/tail and averaging.

A much cheaper solution is to work with the  kriging means $m_k(z)$ in analogue to the standard plug-in estimators $\hat{R}^{SA}$ based on $\bar{y}$'s in nested simulation. Specifically, we use the Harrell-Davis $L-$estimator \eqref{eq:HDweights} based on the sorted posterior means:
\begin{align}\label{eq:meanofR}
  \hat{R}^{HD, \VaR}_k \doteq \sum_n \tilde{w}^{(n)} m^{(n)}_k.
\end{align}
For estimating $\TVaR_\a,$ the empirical weights $\hat{w}^{n,\TVaR}_k = \frac{1}{\alpha N} \1_{\{ m_k(z^n) \le m_k^{(\a N)}\}}$ based on the posterior means are used analogously to \eqref{eq:meanofR}.

Given an estimator $\sum_n \hat{w}^{(n)}_k \hat{f}^{(n)}_k$ and denoting $\hat{\mathbf{w}}_k \doteq (\hat{w}_k^1, \ldots, \hat{w}_k^{N})$, its mean is obtained by plugging in the posterior means $m^{n}_k$ for $\hat{f}^n$ like in \eqref{eq:meanofR} and we can similarly obtain the estimator variance:
\begin{equation}\label{eq:varofR}
s^2(\hat{R}_k) \doteq \var \left( \sum_n \hat{w}^{(n)} f^{(n)}_k \big| \cD_k \right) = \hat{\mathbf{w}}_k\left[\mathbf{C}-\mathbf{c}(z) (\bC+\bm{\Delta}_k)^{-1}\mathbf{c}(z)^T\right]\hat{\mathbf{w}}_k^T.
\end{equation}
Note that the inside term in Equation \eqref{eq:varofR} is the posterior covariance matrix of $\hat{f}^{1:N}_k$, taking advantage of the covariance structure of the GP for additional smoothing. For later use we also define the estimated quantile scenario $\widehat{\cQ}$ which solves $m_k(\widehat{\cQ}) = m^{(\a N)}_k$ and may be compared to the true quantile $\cQ$ based on $f$. 

\begin{remark}
To motivate the use of $\hat{R}^{HD,\VaR}$, consider the ``in-between'' approach of starting with \eqref{eq:risk-posterior} but treating the two terms as independent:
\begin{align*}
\E[ R | \cD_k]  & \simeq \sum_n \E[ f(z^n) | \cD_k] \cdot \E[  w^n | \cD_k]
 = \sum_n \omega^n_k m_k(z^n),
\end{align*}
where the weights are (up to a normalizing factor) $\omega^n_k = \P( z^n \in \cR | \cD_k)$. This means that we should replace the 0/1 weights in \eqref{eq:VaRweight}-\eqref{eq:TVaRweight} with their smoothed posterior probabilities. Simulating $\omega^n$ in our case studies and plotting them against the ordered means  $m^{1:N}$ gives a bell-shape that reasonably matches the Beta shape of the Harrell-Davis weights $\tilde{w}^{(n)}$ in \eqref{eq:HDweights}. (However note that as $k \to \infty$, $\omega^n_k \to w^{n,Var}$ while the Harrell-Davis weights $\tilde{w}$ are independent of $k$.) %
\end{remark}

\section{Sequential Design for Tail Approximation}\label{sec:sequential}

To estimate $R$ we assume being given a total budget of $\cN$ inner simulations to allocate across $N$ fixed scenarios. We then apply the following general procedure \rev{to split the computation} into $K$ sequential rounds $k=1,\ldots, K$, where $\cN_k$ is the \emph{remaining} simulation budget before round $k$ (with $\cN_0 \equiv \cN$):
\begin{enumerate}[I.]
\item Initialize $\hat{f}_0$ by generating simulations over a subset of \emph{pilot} scenarios.

\item Sequentially over $k=0,1,\ldots,$ until $\cN_k=0$, predict $\hat{f}_k$ on $\cZ$ to determine which scenarios are close to $\cR$. Allocate more inner simulations to \rev{this/these} scenario(s), i.e.~increase the corresponding $r^n$'s. This is achieved via an \emph{acquisition function} $\cH(z^n)$ that  takes into account the ``closeness'' of $z^n$ to $\cR$ and the uncertainty $s_k(z^n)$. Potentially a new outer scenario that previously had $r^n_k = 0$ might be sampled. Then update to produce $\hat{f}_{k+1}$ based on the new MC output.
\item The final estimate  $\hat{R}_K$ is obtained from Equation \eqref{eq:meanofR}, with uncertainty expressed via \eqref{eq:varofR}.
\end{enumerate}

Making the above mathematically precise boils down to two objectives: (i)~discover the region of $\cZ$ corresponding to $\cR$, and (ii)~reduce $s_k^2(\cdot)$ in this region. These objectives match the exploration-exploitation tradeoff:  allocating too many replications to solve (i) produces a surrogate that lacks precision even though it recognizes the location of $\cR$, while focusing only on (ii) without sufficient searching may zero in on the wrong region.  To guide this tradeoff during sequential allocation, the acquisition function is based on an uncertainty measure. The strategy is then to (myopically) carry out \emph{stepwise uncertainty reduction} (SUR, see \cite{bect2012sequential,chen2012stochastic,chevalier2014fast,oakley2004estimating,picheny2010adaptive}) by determining what new simulations would most reduce expected uncertainty for the \emph{next} round. \rev{Despite their proliferation in the machine learning literature, to our knowledge, we are the first to apply these concepts for portfolio risk measurement.}

The computational overhead associated with the loop in Step II is non-negligible. It concerns the computation of the acquisition function $\cH_k(z^{1:N})$, as well as the updating of the GP surrogates $\hat{f}_k \to \hat{f}_{k+1}$. For computational efficiency, rather  than adding a single inner simulation, we therefore work with batches of size $\Delta r_k$, meaning that $\cN_{k+1} = \cN_k - \Delta r_k$.
 The sequential design procedure is therefore to determine the best allocation $r_k'^n$ satisfying $\sum_n r_k'^n = \Delta r_k$.  The approaches in Sections \ref{sec:sur}-\ref{sec:appendixEI} allocate all $\Delta r_k$ new replications to a single $z^{k+1}$ scenario.
Alternatively the approach in Section \ref{sec:dad} solves an additional optimization problem on how to distribute the new inner simulations across multiple outer scenarios. For ease of presentation we focus on constant budget per step $\Delta r$; for example, $\Delta r = 0.01 \cN_1$ (where $\cN_1$ is the remaining budget after initialization) yields a procedure with $K=100$ rounds. Further speed-ups to reduce the overhead are discussed in Section~\ref{sec:implementation}.

\begin{remark}
  \rev{Not all outer scenarios will typically be sampled by our algorithms, especially in the early rounds. Therefore, we distinguish between the total $N$ vs.~$N'(k) \le N$ which is the number of scenarios sampled by round $k$. Thus, in expressions such as \eqref{eq:sk-meanvar} or \eqref{eq:meanofR}, the effective computation is with matrices/vectors of size $N'(k)$. This is important for computational overhead since a major bottleneck is handling the $N'(k) \times N'(k)$ matrix $\bC$.}
\end{remark}

\subsection{Targeted MSE Criterion}\label{sec:sur}
For the quantile objective which corresponds to Value-at-Risk computation\rev{,} the region of interest $\cR^{\VaR}$ is the contour $\{ z : f(z) = L \}$, where the level $L = f^{(\a N)}$ is implicit. To learn the contour, we wish to reduce the kriging variance for scenarios close to $L$. To this end we consider the \emph{targeted mean square error} for a scenario $z$ originally introduced in Picheny et al.~\cite{picheny2010adaptive}:
\begin{align}\notag
\text{tmse}^{\VaR}_k(z) &\doteq s_k^2(z)W^{\VaR}_k(z; L) \qquad\text{where}\\ \label{eq:tmse}
  W^{\VaR}_k(z; L) & \doteq  \frac{1}{\sqrt{2 \pi (s^2_k(z) + \varepsilon^2)}} \exp\left(-\frac{1}{2} \left(\frac{m_k(z)-L}{\sqrt{s^2_k(z)+\varepsilon^2}}\right)^2\right) = \phi( m_k(z) - L, s^2_k(z) + \varepsilon^2)
\end{align}
 is based on the Gaussian pdf $\phi$. Thus, $W^{\VaR}_k(z; L)$ is largest for scenarios where predicted portfolio value is close to $L$  and scenarios that have higher posterior variance. The parameter $\varepsilon$ in \eqref{eq:tmse} controls how localized is the criterion around the level $L$.  Picheny et al.~\cite{picheny2010adaptive} recommended $\varepsilon$ to be five percent of the response range, however, they considered noiseless samplers with $\tau^2 \equiv 0$.  In our case, it is desirable to have $\varepsilon$ decrease as $k$ increases, to reflect improving knowledge about $R$. We propose to take $\varepsilon_k \doteq s(\hat{R}^{HD}_k)$ from \eqref{eq:varofR} which captures the uncertainty about the quantile.  This is large when $k$ is small (uncertain in earlier stages), and decreases rapidly as $k$ increases.

The overall acquisition function to be greedily minimized is then $\cH^{timse}_{k+1}$, the total (integrated) $\text{tmse}$ over all of $\cZ$ conditional on adding simulations at $z^{k+1}$:
\begin{align}
\cH^{\text{timse}}_{k+1} &\doteq \frac{1}{N}\sum_{n=1}^N \text{tmse}^{\VaR}_{k+1}(z^n) 
   = \frac{1}{N} \sum_{n=1}^N s_{k+1}^2(z^n)W^{\VaR}_{k+1}(z^n; R). \label{eq:timse-1}
\end{align}
The right hand side of \eqref{eq:timse-1} still contains terms that will only be known after sampling at $z^{k+1}$. Let us define
\begin{equation}
V_k(z^n ; z^{m}) \doteq \left.  C(z^n,z^n)-\mathbf{c}(z^n) (\bC+\bm{\Delta}^{cand}_{k+1})^{-1}\mathbf{c}(z^n)^T
\right|_{\bm{\Delta^{cand}_{k+1}} = \text{diag}\left(\frac{\hat{\tau}^{2}_k(z^{1})}{r^1_k}, \ldots, \frac{\hat{\tau}^{2}_k(z^{m})}{r^m_k+\Delta r_k},\ldots, \frac{\hat{\tau}^{2}_k(z^{N'})}{r^{N'}_k}\right)} \label{eq:varznew}
\end{equation}
which approximates the next-step kriging variance at $z^n$ under the assumption that $\Delta r_k$ additional replications were added to $z^m$ (keeping all other GP pieces frozen from the $k$-th round). Note that the noise matrix $\bm{\Delta}^{cand}_{k+1}$ is affected by the $\Delta r_k$'s (and might have one more row/column relative to $\bm{\Delta}_k$ if a hitherto unsampled scenario $z^{m}$ is considered), but the covariance matrix  $\bC$ is not.

We furthermore approximate $W^{\VaR}_{k+1}(z^n; R)$ with the current tmse weight $W^{\VaR}_k(z^n)$ evaluated at the level $\hat{R}^{HD}_k$, so the final criterion is
\begin{equation}
\widehat{\cH}^{\VaR,\text{timse}}_k(z) \doteq \frac{1}{N} \sum_{n=1}^N V_k(z^n; z) W^{\VaR}_k(z^n; \hat{R}^{HD}_k), \label{eq:timse}
\end{equation}
which is numerically minimized over the next sampling scenario $z$. Selecting $z^{k+1}$ as the minimizer of \eqref{eq:timse} is termed the ST-GP (for ``Sequential TIMSE based on a GP'') procedure.

\begin{remark}\label{rk:EI}
The original \cite{picheny2010adaptive} considered the  case where $\cZ$ is continuous, so the designs were augmented with new  sites $z^{k+1}$ and $\cH^{\text{timse}}$ was defined via an integral. In our case $\cZ$ is finite ($\cH^{\text{timse}}$ is a sum) and fixed, so we add \emph{replications} to existing $z \in \cZ$.
\end{remark}

For $\TVaR_\a$, all scenarios in the left tail  need to be considered so we modify the criterion in Equation \eqref{eq:timse} to instead use weights (keeping everything else the same, including the use of $\varepsilon=s(\hat{R}^{HD}_k) )$
\begin{align}\label{eq:tmse-TVaR}
W_k^{\TVaR}(z; \hat{R}^{HD}_k) & \doteq \frac{1}{\sqrt{2 \pi (s^2_k(z) + s^2(\hat{R}^{HD}_k) )}} \Phi\left(\frac{\hat{R}^{HD}_k - m_k(z)}{\sqrt{s^2_k(z)+s^2(\hat{R}^{HD}_k)}}\right),
\end{align}
where $\Phi(\cdot)$ denotes the standard Gaussian cdf.  The weights $W^{\TVaR}$ increase as $z$ becomes deeper in the tail (i.e.~$\hat{R}^{HD}_k-m_k(z)$ is more positive), while still compensating for uncertainty.

\rev{Note that a further option is to consider a convex combination $W_k = \alpha W_k^{\VaR} + (1-\alpha) W_k^{\TVaR}$, where $0 < \alpha < 1$, combining the objectives of correctly identifying the cutoff for tail scenarios, with accurately predicting all tail losses. In particular, auxiliary experiments suggest that this choice with $\alpha \in [0.1,0.3]$ improves ST-GP for learning $\TVaR$. For brevity, we only present the results for $\alpha=1$ ($\VaR$) and $\alpha=0$ ($\TVaR$).}

\subsection{Expected Improvement Criterion}\label{sec:appendixEI}
An alternate route based on Bect et al.~\cite{bect2012sequential} is to analyze the random variables $\1_{\{\hat{f}_k(z^n) \leq L\}}$ which have variance $\var(\1_{\{\hat{f}_k(z^n) \leq L \} } | \cD_k) = p_k(z^n)(1-p_k(z^n))$, where $$p_k(z^n) \doteq \P(\hat{f}_k(z^n) \leq L | \mathcal{D}_k) = \Phi \bigl( (L -m_k(z^n))/ s_k(z^n)\bigr).
  $$
 These variances are low when $p_k(z^n)$ is close to 0 or 1, i.e.~when $\hat{f}_k$ has strong understanding of which side $\hat{f}_k(z)$ is with regard to $L$. This motivates to minimize the expected level set uncertainty across all $\cZ$, $\sum_n p_{k+1}(z^n)(1-p_{k+1}(z^n))$, conditional on adding inner simulations at a chosen $z^{m}$.  As in the previous section, we substitute $\hat{R}^{HD}_k$ for $L$ and consider the look-ahead version $P_{k}(z^n;z^m)$ of $p_{k+1}(z^n)$ assuming replications are added to scenario $z^m$, cf.~\eqref{eq:varznew},
\begin{equation}
P_{k}(z^n;z^m) = \Phi \left( \frac{\hat{R}^{HD}_k - m_k(z^n)}{ V_k( z^n; z^m) } \right).
\label{eq:Pk}
\end{equation}
The resulting Expected Contour Improvement criterion to minimize is
\begin{align}
\widehat{\cH}^{\VaR,ECI}_k(z) \doteq \frac{1}{N} \sum_{n=1}^N P_{k}(z^n;z)(1-P_{k}(z^n;z)). \label{eq:sur}
\end{align}
\rev{In the sequel this allocation rule in combination with a GP emulator is labelled SE-GP.}

For $\TVaR$, the analogue of \eqref{eq:sur} would be to minimize the predictive uncertainty of the level set $\{z : \hat{f}_{k+1} \leq L\}$. From \cite{labopin2016sequential}, this can be achieved through the acquisition function
\begin{align}
\widehat{\cH}^{\TVaR,ECI}_k(z) \doteq \frac{1}{N} \sum_{n=1}^N P_{k}(z^n;z). \label{eq:surTVaR}
\end{align}
However, note that $\widehat{\cH}^{\TVaR,ECI}$ in fact tends to neglect scenarios satisfying $m_k(z) \ll \hat{R}^{HD}_k$ because in such cases $P_k(z^n; z) \approx p_k(z) \approx 1$ so minimal reduction in uncertainty is achieved by sampling the extreme loss scenarios. Thus, \eqref{eq:surTVaR} ends up behaving similar to \eqref{eq:sur} and does not explore enough of the left tail as needed for proper $\TVaR$ estimation.

\subsection{Dynamic Allocation Designs}\label{sec:dad}

Adding  inner simulations in batches allows the possibility of sampling in parallel several different outer scenarios. This can be advantageous in anticipating the \emph{global} updating effect of running new inner simulations. Such parallel updating also offers a bridge between the sequential approaches and the fixed-stage methods, such as the 3-stage approach of \cite{liu2010stochastic}.

Let $\Delta r_{k}$ be the budget for stage $k$. We wish to choose $\{r_{k}'^n\}$ to minimize the posterior estimator variance $s^2(\hat{R}_{k+1})$ in Equation \eqref{eq:varofR}  subject to the constraints $\sum_{n} r_k'^n = \Delta r_{k}$, and $r_k'^n \geq 0.$  To do so, we extend the solution of Liu and Staum~\cite{liu2010stochastic}, who considered the specific case when $r^n_k$ is constant in both $k$ and $n$, to the sequential design context where $r^n_k$ varies among rounds $k$ and locations $z^n$.
Conditioning on a set of $r_k^{'n}$'s and writing out the look-ahead variance,
\begin{equation}
s^2_{k+1}(\hat{R}_{k+1}) = \hat{\mathbf{w}}_{k+1} (\bC - \bC (\bC + \bm{\Delta}^{cand}_{k+1})^{-1} \bC) \hat{\mathbf{w}}_{k+1}^T \label{eq:minimizing-var2}
\end{equation}
where $\hat{\mathbf{w}}_{k+1}$ are the estimates of the weights in \eqref{eq:riskmeasure}.  Note that $s^2_{k+1}(\hat{R}_{k+1})$ is driven by each $r_{k+1}^n$ through the look-ahead noise matrix $\bm{\Delta}^{cand}_{k+1}$ that has diagonal entries $\frac{\tau^2(z^n)}{r^n_k + r^{'n}_k}$. To proceed,
we freeze the weights $\hat{w}^n_{k+1} = \hat{w}^n_k$, whereby the key calculation involves expressing the matrix inverse  $(\bC + \bm{\Delta}^{cand}_{k+1})^{-1}$ in terms of $r^{'n}$'s. This is done in
Lemma~\ref{lemma:inverseApprox} in Appendix~\ref{sec:varianceMinimization}. The respective proof further shows  that minimizing \eqref{eq:minimizing-var2} is then equivalent to minimizing
\begin{equation}\label{eq:sv-gp}
\mathbf{u}_{k} \bm{\Delta}^{cand}_{k+1} \mathbf{u}_{k}^T \qquad\quad\text{where}\qquad\mathbf{u}_{k}^T \doteq  (\bC + \bm{\Delta}_k)^{-1}\bC \hat{\mathbf{w}}_k^T.
\end{equation}
Solving this minimization problem can be done using a \emph{pegging algorithm}, see e.g.~\cite{bretthauer1999nonlinear}. As before, for the $\tau^2(z^n)$ terms in $\bm{\Delta}_k$ and $\bm{\Delta}^{cand}_{k+1}$ we use the \texttt{hetGP} estimates.

Note that \eqref{eq:sv-gp} is predicated on \emph{adding} replicates to scenarios with existing simulations since otherwise $\bm{\Delta}^{cand}_{k+1}$ would be undefined for $r^{'n}_k = 0$. Thus, to implement \eqref{eq:sv-gp} we first create the allocation subset $\cZ^{SV}_k$ of scenarios to optimize over so that in \eqref{eq:sv-gp} the index $n$ runs through $\{ n : z^n \in \cZ^{SV}_k \}$. We found that constructing $\cZ^{SV}$ based on the ST-GP weights  works well (see Section~\ref{sec:implementation} below which uses the same construction for screening scenarios for ST-GP and SE-GP). If necessary, prior to minimizing \eqref{eq:sv-gp} we then add single replicates to any scenarios in $\cZ^{SV}_k$ that have $r^n_k = 0$.

\rev{The above algorithm, which we label SV-GP} (for Variance minimization) in what follows,  based on solving \eqref{eq:sv-gp} makes three approximations. First, it modifies the minimization problem by using a matrix approximation in Lemma~\ref{lemma:inverseApprox}. Second, it treats noise variances as known (specifically it uses same estimate of ${\tau}^2(z)$ across stage $k$ and $k+1$). Third, it treats the GP emulator as fixed, i.e.~it does not take into account that the hyperparameters are themselves evolving in $k$. Due to all of the above, the resulting allocation $r_k'^n$ is sub-optimal. This means that while the method is providing a recipe to allocate arbitrary simulation budget, it cannot entirely replace the sequential rounds. Indeed, this is the basic shortcoming of a fixed-stage approach like that of \cite{liu2010stochastic}: by taking only a few stages,  extra onus is placed  on the initialization and proper fine-tuning of the GP surrogate. As we show below, the ultimate performance of SV-GP is not materially better than the more ``naive'' one-scenario-at-a-time strategies of the previous section. Since the approximations in
\eqref{eq:sv-gp} diminish  as each $r_{k}^n$, $n=1, \ldots, N$ increases, SV-GP  tends to perform best at the later stages. It also implies that it may be beneficial to increase batch amounts $\Delta r_k$ as $k$ grows.

\begin{figure}[!ht]
  \centering
  \begin{tabular}{cc}
  VaR & TVaR\\
  \includegraphics[scale=0.28]{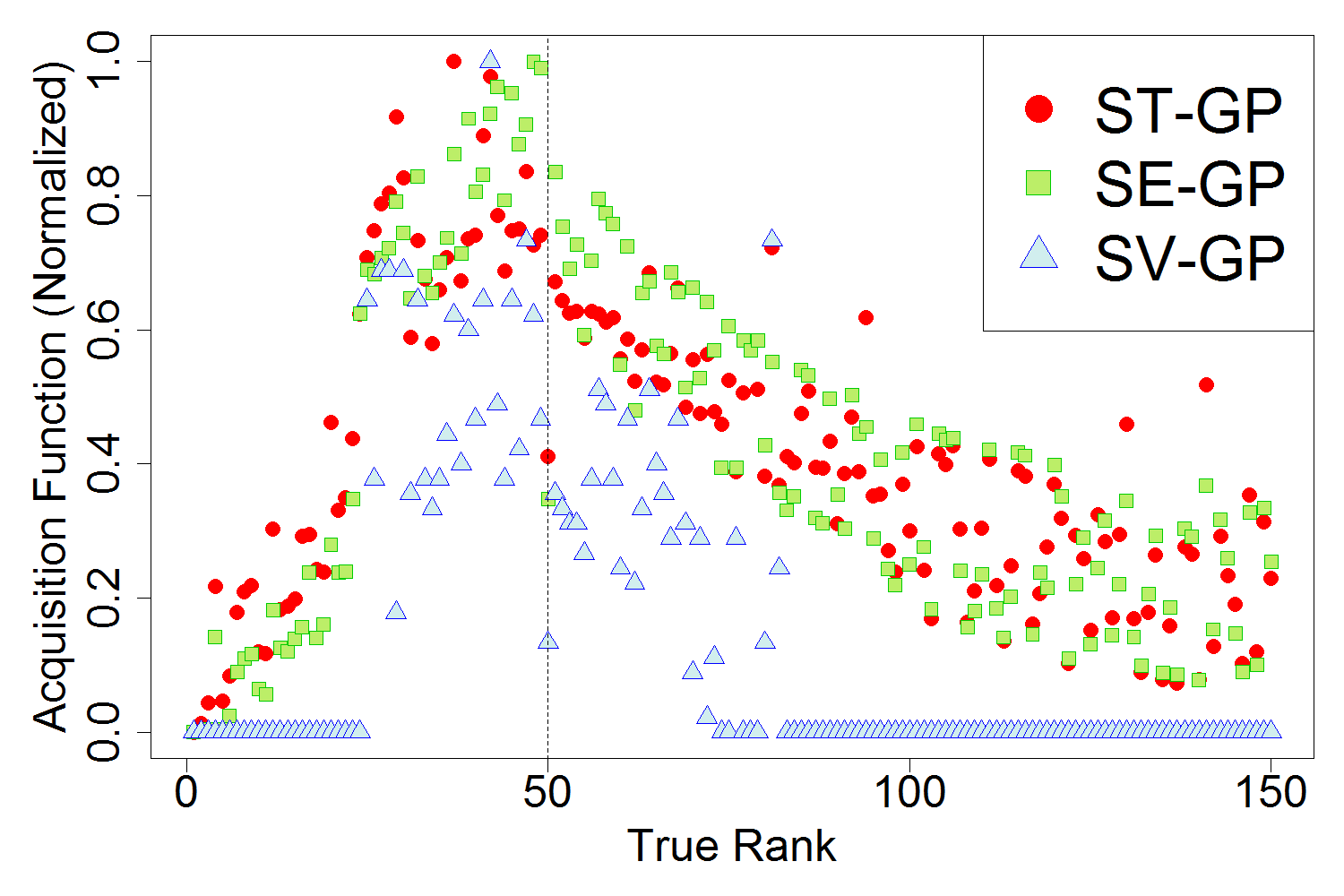}  &
  \includegraphics[scale=0.28]{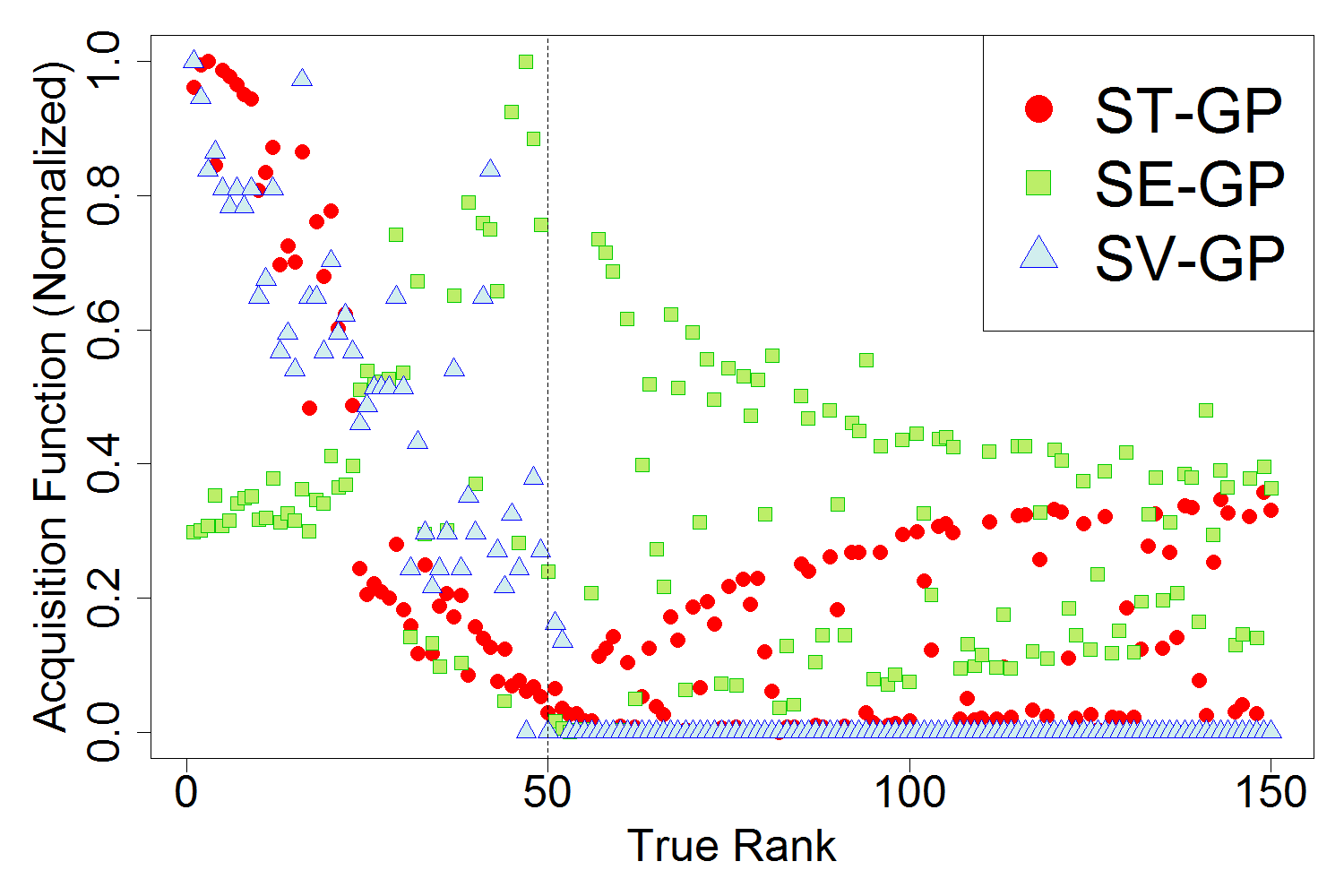}
  \end{tabular}

  \caption{Normalized acquisition functions for ST-GP, SE-GP and SV-GP on the same simulation output, namely after round $k=20$, using SR-GP-2 and the 2-D Black-Scholes case study.  Here, ``acquisition function'' for SV-GP refers to the allocation results from the minimization problem in Equation \eqref{eq:sv-gp}.  For visualization, ST-GP and SE-GP criteria are changed in sign so that they allocate $\Delta r_k$ to the point with the \emph{highest} value; all $\cH$'s are normalized to be in the range $[0,1]$. The vertical dashed lines indicate the true quantile scenario $\cQ$. }\label{fig:acquisitionplot}
\end{figure}

To illustrate the various acquisition functions,  Figure~\ref{fig:acquisitionplot} shows normalized values of $\widehat{\cH}_k(z^n)$ for ST-GP, SE-GP and SV-GP. For the latter, we take $\cH^{SV}_k(z^n) \equiv r^{'n}$, i.e.~the next-stage allocation amounts. The underlying model comes from the case study in Section \ref{sec:casestudy-BS} below.
We observe that for $\VaR$, all $\widehat{\cH}$ target the quantile $\cQ$,  with SV-GP more strongly focused in this region. Recall that ST-GP and SE-GP will pick $z^{k+1}$ as the maximizer of the shown acquisition functions, while SV-GP divides $\Delta r_k$ according to the indicated $r^{'n}_k$. For $\TVaR$, ST-GP and SV-GP target the deep tail (scenarios with worst portfolio losses), while SE-GP
still focuses on the quantile region.

\subsection{Rank-Based Allocation}
A simpler approach is to allocate  inner simulations based on the posterior means $m_k(z^n)$. Specifically, for given parameters $L$ and $U$ satisfying $L \le \a N  \le U$, the SR-GP (``Sequential Rank'') algorithm allocates uniformly to all scenarios with rank between $L$ and $U$: $\cR^{SR-GP}_k(L,U) \doteq \{z^n \in \cZ: m_k(z^n) \in [m_k^{(L)}, m_k^{(U)}] \}$. Thus, $r^{'n}_k = \Delta r_k / (U-L)$ for $z^n \in \cR^{SR-GP}_k(L,U)$ and zero otherwise.

If the values of $L$ and $U$ are chosen conservatively, SR-GP ``blankets'' all scenarios in the predicted neighborhood of interest. This ensures exploration of the potential tail, but the adaptive targeting might not be localized enough. Conversely, taking $L = U = \a N$ yields a very aggressive SR-GP scheme for estimating $\VaR_\a$: it greedily adds all $\Delta r_k$ scenarios to the empirical quantile, i.e.~scenario $\widehat{\cQ}_k$. Employing this extreme as a comparator addresses the value of exploration during the sequential design. In general, SR-GP should be suboptimal compared to ST-GP, SE-GP and SV-GP since it does not take uncertainty of $\hat{f}_k$ into account. It also suffers from the a priori difficulty to reasonably set $L$ and $U$.  To obtain well-performing estimators in the case studies, our  choices for $L$ and $U$ came through several iterations of tweaking to see what worked well, all depending on the case study and risk measure. 

\section{Algorithm}\label{sec:algorithm}

\begin{figure}[!ht]
  \centering
  \begin{tabular}{ccc}
  $k=1$ & $k=30$ & $k=100$ \\
  \includegraphics[height=2.1in,width=0.31\textwidth,trim=0.1in 0.35in 0in 0.1in]{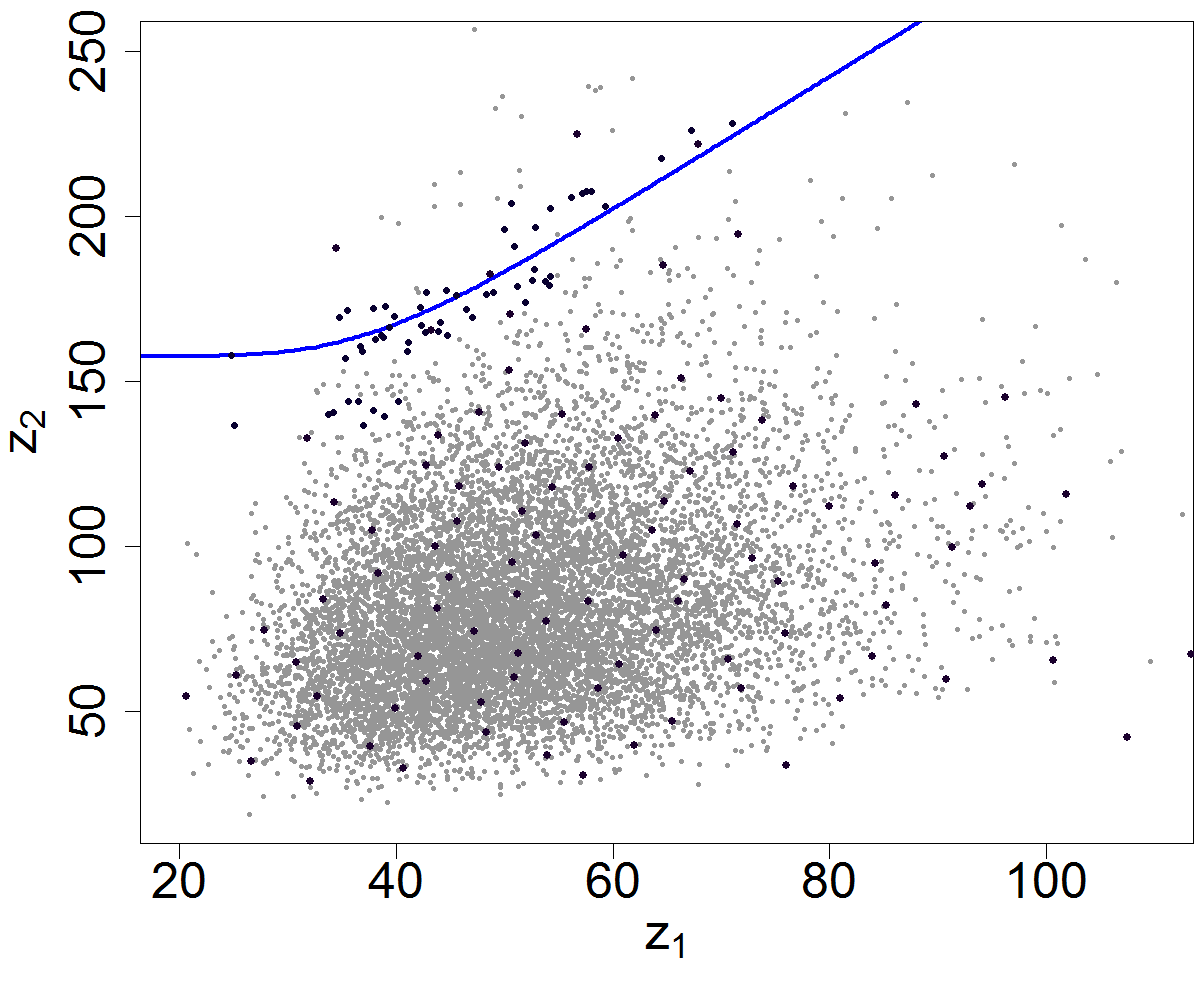} &
  \includegraphics[height=2.1in,width=0.31\textwidth,trim=0in 0.35in 0in 0.1in]{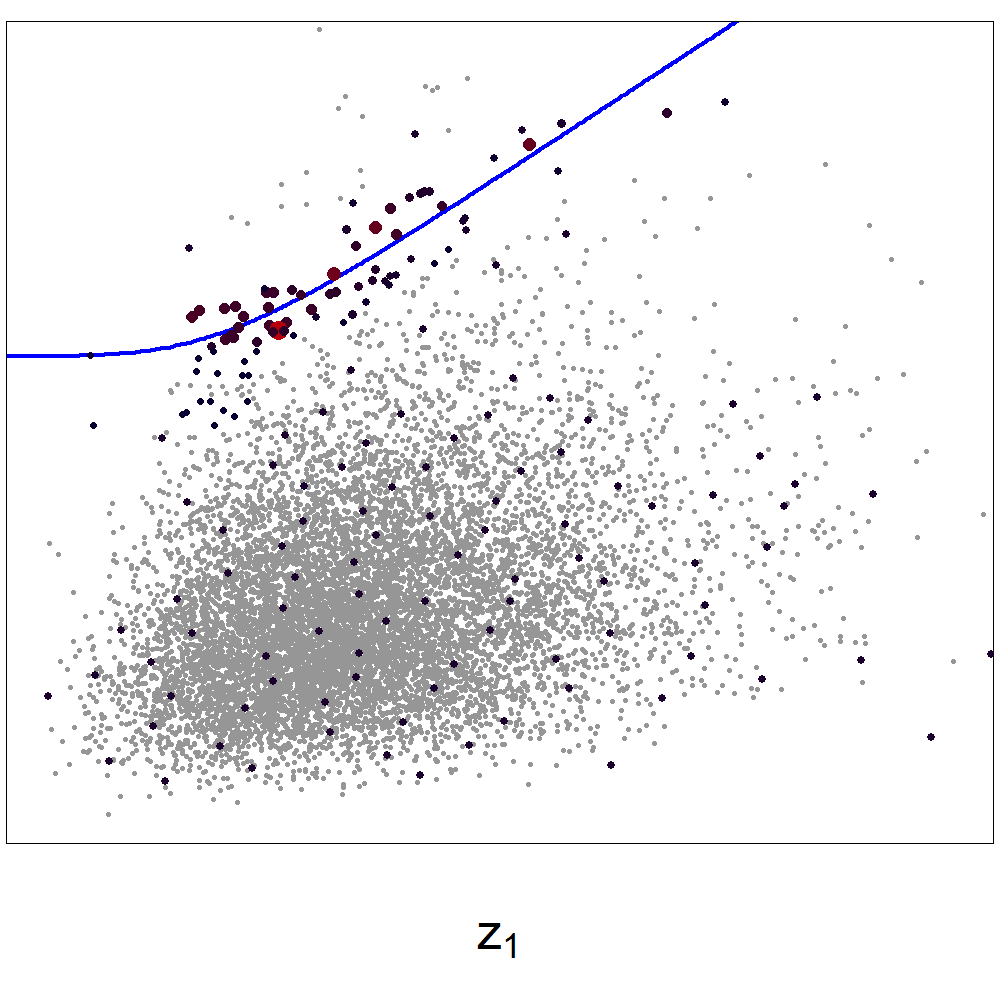} &
  \includegraphics[height=2.1in,width=0.31\textwidth,trim=0in 0.35in 0in 0.1in]{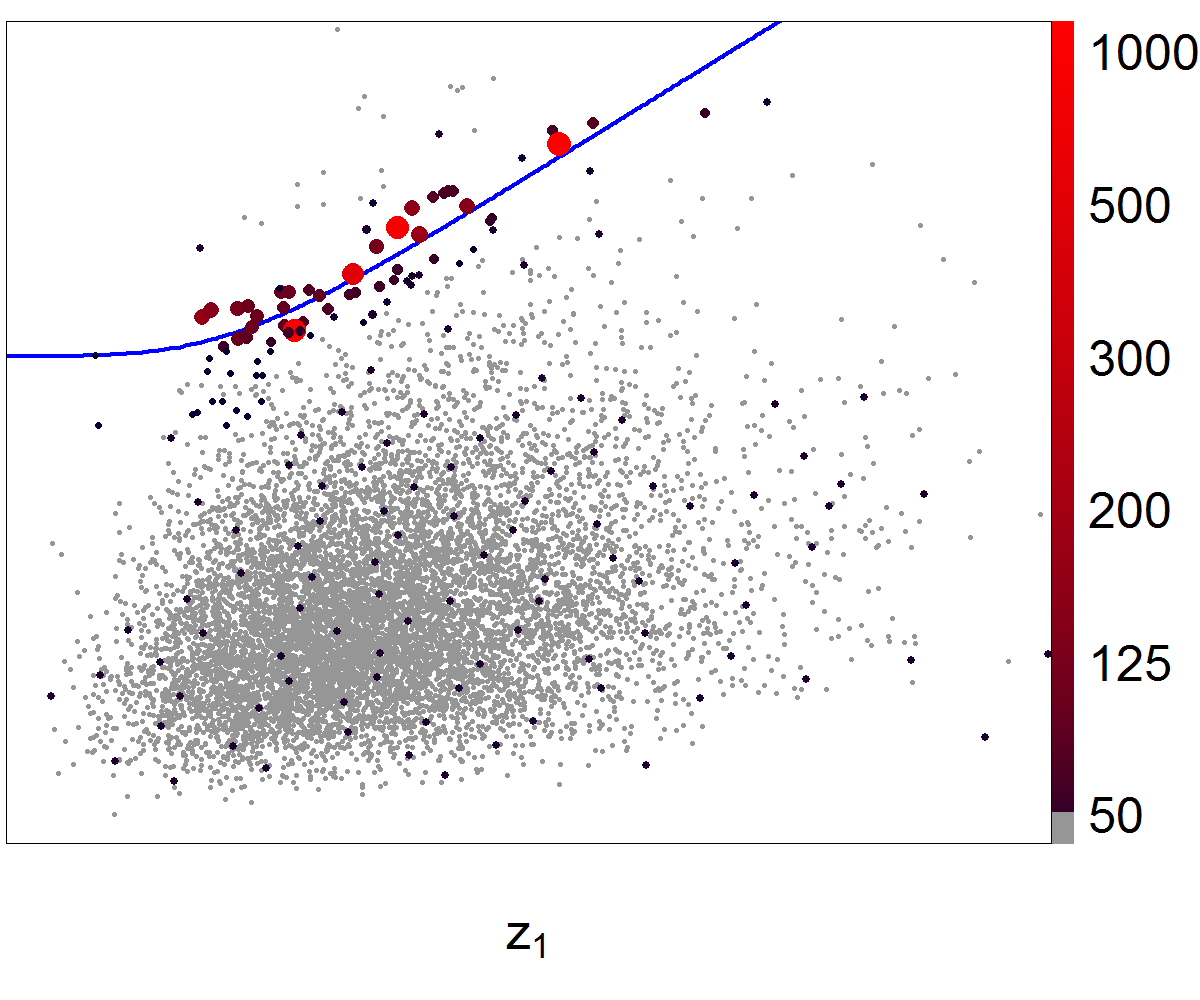}
  \end{tabular}
  \caption{Sequential budget allocation by SV-GP at stages $k=1, 30, 100$ to learn $\VaR_{0.005}$ in the 2-D Black-Scholes case study of Section~\ref{sec:casestudy-BS}. The blue line indicates the true  quantile contour $f^{(50)}$.  Each dot represents an outer scenario $z^n$; the respective size and color are scaled non-linearly in $r^n_k$. Some scenarios receive as many as $r^n_k \approx 1200$ scenarios (total budget of $\cN = 10^4$). \label{fig:movie}}
\end{figure}

To illustrate our sequential schemes,  Figure~\ref{fig:movie} visualizes the evolution of scenario allocations $r^n_k$ over one run of the SV-GP method for $\VaR$ estimation. As $k$ increases, \rev{the} vast majority of the budget goes to the tail scenarios, highlighting the adaptive inner simulations. Thus,  in the later stages very few additional scenarios are augmented to $\cD_k$; rather the values of $r_k^n$ are increasing for locations near the quantile $\cQ$. The learning of $f$ can be observed by comparing the $k=1$ panel where the algorithm investigated scenarios with $z_2>225$, and the panels with $k=30$ and $k=100$ where the corresponding values of $r^n_k$ remain low after the emulator realizes that $|m_k(z^n)-R|$ is in fact large in that region. The final plot for $k=100$ also shows how the GP's spatial covariances are taken into account: scenarios are sampled in ``clusters''---centered around a single scenario with a very high $r_k^n$, with nearby locations benefiting from this information.

\subsection{Updating the Emulator}

At each stage of the algorithm, the GP emulator is updated from $\hat{f}_k$ to $\hat{f}_{k+1}$. After recording the number of replications $(r'^n_{k})_{n=1}^N$ to add to each scenario according to the procedures discussed  in Section \ref{sec:sequential} and  generating the new data $\{ y^{n,i} : i = r^n_{k}+1, \ldots r^n_{k} + r'^n_k\}$, we have the recursive update
\begin{align}
r^n_{k+1} & = r^n_k + r^{'n}_k, \label{eq:updating-r}
\end{align}
which is used for $\bm{\Delta}_{k+1}$ in \eqref{eq:sk-meanvar}. Next, directly re-computing $m_{k+1}, s^2_{k+1}$ requires inverting the $| \cD_k| \times | \cD_k|$ GP covariance matrix $\mathbf{C}$ which is computationally expensive. Instead, we \emph{update} the GP model by re-using the expressions on the RHS of \eqref{eq:sk-meanvar} from stage $k$. In our context updates come in two flavors, both of them internally incorporated in \texttt{hetGP}. First, if a new outer scenario $z^{k+1}$ (that previously had $r^n_k = 0$) is added to $\cD_{k}$ then the GP updating equations \cite{chevalier2014fast} can be applied, utilizing the Woodbury formula for finding matrix inverse when augmenting by 1 row. Second, if only new replications are added, i.e.~$|\cD_{k+1}| = |\cD_k|$ then we adjust $\bm{\Delta}_{k+1}$ without touching $\mathbf{C}$ \cite{hetGP}. Intuitively, it is sufficient to just keep track of $\bar{y}^n$, $\hat{\tau}^2(z^n)$ and $r^n_k$'s, and the latter 3 vectors all satisfy simple updating arithmetic
\cite{chan1982updating,kaminski2015method}
\begin{align}
\label{eq:updatingYbar} \bar{y}^n_{k+1} &= \frac{r^n_k \bar{y}^n_k + r'^n_k \bar{y}'^n_k}{r^n_k+r'^n_k}, \\
\label{eq:updatingTau} \hat{\tau}_{k+1}^2(z^n) &= \frac{1}{(r^n_k+r'^n_k-1)}\left((r^n_k-1)\hat{\tau}^2(z^n)+ (r'^n_k-1)\hat{\tau}_k'^2(z^n) + \frac{r'^n_k r^n_k}{r^n_k+r'^n_k}\left(\bar{y}^n_k -\bar{y}'^n_k \right)^2  \right),
\end{align}
where $\bar{y}^{'n}_k$ is the average of the new $y^{n,i}$'s and $\hat{\tau}^{'2}_k$ is the corresponding sample variance. Note that updating freezes the GP hyperparameters, see Section~\ref{sec:implementation} below.

\subsection{Initialization of $\hat{f}$}\label{sec:initialization}
Initially, we have $r^n = 0$ for all $n$, i.e.~no data to inform us where the tail region may lie. To initialize, the typical solution is to use a small number $N_{init}$ of \emph{pilot} scenarios that are representative of the entire domain $\cZ$, determined by some space-filling algorithm (e.g.~Latin Hypercube sampling, LHS).  Chauvigny et al.~\cite{chauvigny2011fast} provides a more elaborate method using statistical \emph{depth functions} to identify pilot scenarios based on the geometry of $\cZ$.

A small challenge is that LHS and other space-filling approaches are not  directly applicable to a discrete scenario set which is non-uniform in space (see Figure~\ref{fig:contour}).
Instead we develop a minimax-style initialization procedure based on Euclidean distance $\|z-z'\|_2$ and described in Appendix~\ref{App:initial}. Note that the resulting design size $N_{init} = | \cD_0|$ is not fully pre-determined. After determining $\cD_0$, we allocate an equal number of inner simulations to each scenario, i.e.~$r'^n_0 = \Delta r_0/N_{init}$ for all $\{ n : z^n \in \cD_0\}$ to build $\hat{f}_0$. In our case studies below we take $\Delta r_0 = 0.1 \cN$ and $N_{init} = 0.01 N$ so that we spend 10\% of the total budget on the pilots which cover 1\% of the scenarios. The $k=1$ panel of Figure~\ref{fig:movie} illustrates this space-filling procedure by plotting the pilot scenarios (heavy black dots).  It also shows the first round allocation, where locations having $r^n_1>0$ are starting to line up along the initial contour estimate.

\subsection{Implementation Details} \label{sec:implementation}
The two most popular kernels for GP emulation are  the Mat\'ern$-5/2$ family in Equation \eqref{eq:maternkernel2} and the Gaussian family,
\begin{equation}\label{eq:Gsn-kernel}
 g^{Gsn}(h, \theta) \doteq \exp\left(\frac{-h^2}{2 \theta^2}\right).
\end{equation}
Note that the choice of the kernel family modifies the meaning of the hyperparameters, so direct comparison is fraught. In Section~\ref{sec:het-vs-sk} we present numerical evidence that the
 choice of covariance kernel generally has only a minor impact on the resulting emulators $\hat{f}_k$.

The classical method for inferring the hyperparameters $\bm{\theta}$ and process variance $\sigma^2$ is by optimizing the marginal likelihood, either through MLE or penalized MLE, using the likelihood function based on the distributions described in Section \ref{sec:SimpleKriging2}. Either case leads to a nonlinear optimization problem. 
In a sequential setting, the hyperparameters are ideally refitted at each stage to improve the Likelihood function as more data comes in. However, this is computationally impractical since evaluation of the likelihood requires $\mathcal{O}(|\mathcal{D}_k|^3)$.  Alternatives to this are to refit the hyperparameters at certain points in the algorithm, e.g.~after 10\%, 20\%, \ldots, 90\% (or with a nonlinear schedule, such as 2\%, 4\%, 8\%, \ldots) of the budget has been depleted.  The latter scheme has the advantage of refitting more frequently earlier when there is less certainty about the hyperparameters, and when the refits are cheaper.

The above logic also implies preference for algorithms that keep $| \cD_k|$ small. In that sense a more targeted algorithm like ST-GP is preferred over SV-GP that by its nature spreads inner simulations across many scenarios. A related overhead concerns predicting $\hat{f}_k$ on $\cZ$, i.e.~evaluating $m_k(z^{1:N})$ and $s_k(z^{1:N})$, which is required by all the acquisition functions. This carries $\mathcal{O}( | \cD_k|^2 \cdot | \cZ |)$ effort, once again showing the cost of a large design $\cD_k$. It also leads to the trick of reducing $\cZ$ to a candidate set $\cZ^{cand}_k$ by screening scenarios that are far from the tail.  One way to screen is to re-use the weights $W_k(z^n)$ from Equations~\eqref{eq:tmse} and~\eqref{eq:tmse-TVaR} to assess relevance of scenarios. Specifically, in our case studies below we determine the candidate set $\cZ^{cand}_k$ at stage $k$ via
\[ \cZ^{cand}_k \doteq \{z^m \in \cZ : W^{\VaR}_{k}(z^m)/(\sum_{n=1}^N W^{\VaR}_k(z^n)) > 10^{-3}\},\]
for $\VaR$ and
$\cZ^{cand}_k \doteq \{z^n \in \cZ : W_k^{\TVaR}(z^n) > 10^{-3}\}$ for $\TVaR$. Thus only above outer scenarios are considered when evaluating the acquisition functions (methods ST-GP, SE-GP, SV-GP).  As an illustration, on a typical run of the first case study, this screening shrunk the prediction set from $N=10000$ scenarios to $| \cZ^{cand}_1| = 528$ (657 for TVaR) after the first round, reducing the overhead of predicting the GP surrogate on $\cZ^{cand}_1$ by a factor of $\approx 20$.  Furthermore, the candidate sets $|\cZ^{cand}_k|$ shrink as $k$ grows, as the GP better learns the tail scenarios.

In the examples below we used a constant stage-wise budget $\Delta r_k \equiv \Delta r$; an embellished version could easily include changing batch sizes over rounds, for example smaller batches in earlier rounds where the GP has high global uncertainty. Choosing batch size also determines the total number of rounds $K$, influencing GP regression and prediction overhead, as well as the computation cost of SUR acquisition functions. In general we find (cf.~Section \ref{sec:cs1results}) that the algorithms focus on only a small number of scenarios which means that the same $z^n$'s get repeatedly picked by our sequential criteria. As a result, a large batch size, say $\Delta r = 0.1 \cN_1$, offers little performance loss compared to a smaller $\Delta r =0.01 \cN_1$, and reduces the overhead of selecting $z^{k+1}$.

Returning to the GP emulator itself, we found that a simple trend function $\mu(\cdot)$ in \eqref{eq:kriging3} is sufficient for de-trending. While the choice of $\mu(\cdot)$ has little impact on the prediction $\hat{f}_k$, it helps the nonlinear optimization routine in fitting the hyperparameters, and is also beneficial in the very beginning of the sequential stages. Note that the trend modifies somewhat the spatial dependence of the response and hence the $\theta_j$ hyperparameters. A parametric mean function can also be specified via basis functions with respective coefficients to be fitted, so that $\mu(z) = \beta_0  + \sum_{j} \beta_j h_j(z)$.  In this case, the $\beta$'s are estimated simultaneously with the other hyperparameters, an approach known as Universal Kriging~\cite{roustant2012dicekriging}. \rev{In our experience, in most financial settings $\mu(\cdot)$ can be either easily identified as the intrinsic value of the portfolio (like in our first case study), or taken to be a constant $\beta_0$ (fitted as part of the GP, see the second case study).}

\subsection{Comparison with other Approaches}\label{sec:benchmarks}
To benchmark the proposed algorithms  defined in Section \ref{sec:sequential}, we compare to several alternatives. We primarily focus on other regression-based approaches and concentrate on quantifying the role of (i) adaptive budget allocation; and (ii) sequential approaches as compared to simpler 1-, 2- or 3-stage methods. A summary of these benchmarks, as well as the procedures discussed previously in the section are given in Table \ref{table:procedures1}.

\begin{enumerate}
\item LB: a perfect information method used as a Lower Bound. We assume that the tail/quantile scenarios are known a priori, and look to minimize MSE of $\hat{R}$. For $\VaR$ this corresponds to minimizing the posterior variance at $\cQ$ which is trivially achieved by allocating the entire budget $\cN$ to the true quantile scenario. In that case there is no GP surrogate and the estimator variance is the Monte Carlo averaging error, i.e.~$\tau^2( \cQ)/\cN$. For $\TVaR_\a$, we allocate the budget $\cN$ uniformly among the true tail $\{z^n : f(z^n) \leq f^{(\a N)}\}$ and fit a GP to the results.  

\item A3-GP: The Adaptive 3-stage algorithm of Liu and Staum~\cite{liu2010stochastic}. The first stage simulates from space-filling pilot scenarios similar to our approach. The second stage allocates uniformly across a screened candidate set $\cZ^{cand}_1$.  The third stage then solves for $r'^{n}$ to minimize variance of $\hat{R}_2$ like in SV-GP and Lemma \ref{lemma:inverseApprox}. We follow their suggestions, with $0.02N$ stage-2 design locations, and stage-3 budget of $\cN_2=0.7\cN$.

\item U2-GP: A 2-stage approach. After a space-filling first stage as in Section \ref{sec:initialization}, the remaining budget is allocated uniformly among the lowest $2 \a N$ scenarios. This is the simplest version of an adaptive allocation with $r_1^{'n} = \cN_1/(2 \a N) \1_{ \{ m_0(z^n) \leq m_0^{(2\a N)}\} }$. Note that U2-GP can be seen as a version of SR-GP with $K=2$ rounds and $L=1, U=2\a N$. Comparing to this approach quantifies the gain of multi-stage procedures.

\item  U1-GP: A 1-stage approach which uniformly allocates $\cN/|\cZ|$ to each $z \in \cZ$, fits a GP surrogate $\hat{f}_0$ to the resulting design and uses $ m_0^{1:N}$ to estimate $\hat{R}$. This is the crudest comparator that has no adaptive allocation but still employs spatial smoothing.  We find that fitting a GP to the entire collection of outputs is computationally unfeasible for large scenario sets, $N \gg 1000$, so instead we fit to $\{z^n : y^n \leq y^{(2000)}\}$, which yields similar overhead times to ST-GP and still offers significant improvement over no GP smoothing.

\end{enumerate}

\begin{table}[ht]\small
\centering
\begin{tabular}{c|l|l}\small
Name & Description & Parameters\\ \hline
LB & Monte Carlo under perfect information & $\TVaR$ allocates uniformly to true tail scenarios \\
ST-GP & Targeted MSE criterion \eqref{eq:timse} & $\varepsilon_k = s(\hat{R}^{HD}_k)$\\
SE-GP & Expected Level Improvement criterion \eqref{eq:sur} &  \\
SV-GP & Batch Variance minimization (Section \ref{sec:dad}) & \\
SR-GP-1 & Uniform on $\{z^n : m_k^{(L)} \leq m(z^n) \leq m_k^{(U)} \}$ & L=50, U=50 ($\VaR$), L=1,U=50 ($\TVaR$) \\
SR-GP-2 & & L=26, U=75 ($\VaR$), L=1, U=75 (TVaR) \\
A3-GP & 3-stage from \cite{liu2010stochastic}& 70\% of budget for stage 3; $r_2^n \equiv 10$\\
U2-GP & Two-stage & SR-GP with $L=1, U=100$ and $K=2$\\
U1-GP & Uniform allocation & GP fit to the bottom 20\% of $z^n$'s based on $\bar{y}^{1:N}_1$ \\
U1-SA & Uniform allocation & Uses sample averages $\bar{y}^{1:N}$ as output\\
BR-SA & R\&S algorithm from \cite{broadie2011efficient} &
\end{tabular}
\caption{Comparator approaches as described in Section \ref{sec:benchmarks}.  All approaches (except U1-GP) use the same initialization procedure described in Algorithm \ref{alg:initialFit} with $N_{init}=0.01 N$ stage-1 scenarios.  The budgeting parameters for the sequential methods are chosen to yield $K=100$.}\label{table:procedures1}
\end{table}

We also consider two non-GP methods to investigate the importance of spatial smoothing. The first is called U1-SA, which is the same as U1-GP but uses the sample means $\bar{y}^n$ and empirical variances $\hat{\tau}^2(z^n)$ in place of GP-based posterior means. This is the ``vanilla'' nested simulation method.  The second is a  ranking-and-selection (R\&S) approach, introduced in Broadie et al.~\cite{broadie2011efficient} for the purpose of comparing $f^{1:N}$ with a given level $L$. With our notation, at stage $k$ this BR-SA algorithm allocates $\Delta r_k$ simulations to the scenario $z^{k+1}$ which minimizes
$\cH^{BR}_k(z^n) \doteq  r_k^n \cdot | \bar{y}^n_k - L| / \tau(z^n)$. The above acquisition function compares the scenario sample average to $L$, normalizing by the respective sample uncertainty $\tau(z^n)/r_k^n$ which resembles $\cH^{ECI}$ in \eqref{eq:sur} but without having an emulator. For our purposes we replace $L$ with $\hat{R}_k^{HD}$. Also, we follow the suggestion of \cite{broadie2011efficient} to estimate local noise variance by a weighted $\tilde{\tau}(z^n)$ based on the aggregate variances gathered, corresponding to a very simple homoskedastic smoother:
\begin{equation*}
\tilde{\tau}^{2}_k(z^n) = \frac{r^n_k}{r^n_k+\tilde{r}} \hat{\tau}^2_k(z^n) + \frac{\tilde{r}}{r^n_k + \tilde{r}} \bar{\tau}^2_k,
\end{equation*}
where $\bar{\tau}^2_k = \frac{1}{N'} \sum_{n=1}^{N'} \hat{\tau}^2_k(z^n)$ is the average over all empirical variances.  Here, $\tilde{r}$ is the smoothing parameter, and we take $\tilde{r}=5$ following \cite{broadie2011efficient}.  A comparison with these non-GP benchmarks is discussed in Section \ref{sec:gains}.

\begin{remark}
  Note that following above ``recipe'' one may combine any of the introduced acquisition functions with a sample-averages based approach, yielding, say, ST-SA or SV-SA approaches. As explained, a limitation of not having an emulator is inability to properly forecast $\tau^2(z)$ at scenarios that do not have enough inner simulations.  From the other direction, we recall the Least Squares Monte Carlo methods (like in Bauer et al.~\cite{bauer2012calculation}) which use spatial smoothing but a non-sequential allocation. We do not compare to the latter since our focus is on experimental design rather than purely the regression/smoothing step.
\end{remark}

\section{Case Study: Black Scholes Option Portfolio}\label{sec:casestudy-BS}
We begin with a two--dimensional example where $f(z)$ can be computed exactly. Working in 2-D with a known $f$
allows for easy visualization of the algorithms, and provides exact error calculations. Consider a portfolio whose value is driven by two risky assets $Z_t \equiv (S^1_t, S^2_t)$ that have Geometric Brownian motion dynamics:
\begin{align*}
dS^1_t &= \beta S^1_t\, dt + \s_{1} S^1_t \, dW^{(1)}_t,\\
dS^2_t &= \beta S^2_t\, dt + \s_{2} S^2_t \, dW^{(2)}_t.
\end{align*}
Above the $W^{i}$ are correlated Brownian motions under the risk neutral measure $\Q$ with $d\langle W^{(1)}, W^{(2)}\rangle_t = \rho \, dt$ and $\beta$ is the constant interest rate. Our model parameters are summarized in Table \ref{table:VaR-portfolio}. The portfolio consists of Call options: long  100 $K_1=40-$strike Calls on $S^1$ and short 50 $K_2=85-$strike Calls on $S^2$.
We work with a risk horizon $T=1$ year, and target $\VaR_\a$ and $\TVaR_\a$ risk measures with $\a = 0.005$. By risk-neutral pricing, the value of the portfolio at $T$ and starting with $z \equiv (z_1, z_2) \in \mathbb{R}^2_+$ is
\begin{align}
f(z) \doteq \E^{\Q}\left[100 e^{-\beta(T_1-T)} \left(S^1_{T_1}- 40\right)_+ - 50 e^{-\beta (T_2-T)}\left(S^2_{T_2}- 85\right)_+ \big| (S^1_T, S^2_T) = z\right]. \label{eq:Pi}
\end{align}

\begin{table}[!ht]
\centering
\begin{tabular}{cccccc}
 Asset & Position & Initial Price $S^i_0$ & Strike $K_i$ & Maturity  $T_i$ & Volatility $\sigma_i$ \\
\hline $S^1$ & 100 & 50 & 40 & 2 & 25\%\\
$S^2$ & -50 & 80 & 85 & 3 & 35\%\\ \hline
 &
\multicolumn{2}{c}{ Correlation $\rho = 0.3$} & \multicolumn{3}{c}{ Interest Rate $\beta = 0.04$} \\ \hline
\end{tabular}
\caption{\label{table:VaR-portfolio} Parameters of the 2-D Case Study for a Black-Scholes portfolio on stocks $S^1$ and $S^2.$ }
\end{table}

 A contour plot of $f$, obtained from the Black Scholes formula, is shown in Figure~\ref{fig:contour}. Observe that $f(z)$ is most negative in the upper-left corner, while $\tau^2(z)$ is increasing in both $z_1$ and $z_2$, so that additional focus may be spent toward the upper-right corner where uncertainty is largest and the scenarios are more scarce. The figure also exhibits the scenario set $\cZ$ which we generated as a \emph{fixed} sample from the bivariate log-normal distribution of $(S^1_T,S^2_T)$. Note that here $\cZ$ was generated using $\mathbb{Q}$-distribution, so the dynamics of the factors on $[0,T]$ and $[T, T_i]$ are the same, i.e.~the physical and risk-neutral measures coincide. This is solely for a simpler presentation of the case-study; in our experience the role of $\cZ$ is secondary to the other considerations. The red line in Figure~\ref{fig:contour} shows the true (relative to the shown $\cZ$) quantile loss $f^{(\a N)} = f(\cQ) = -4052.02$, indicating the region that the methods are supposed to target.

\begin{figure}[!htb]
  \centering
  \includegraphics[scale=0.35,trim=0.1in 0.25in 0.1in 0.1in]{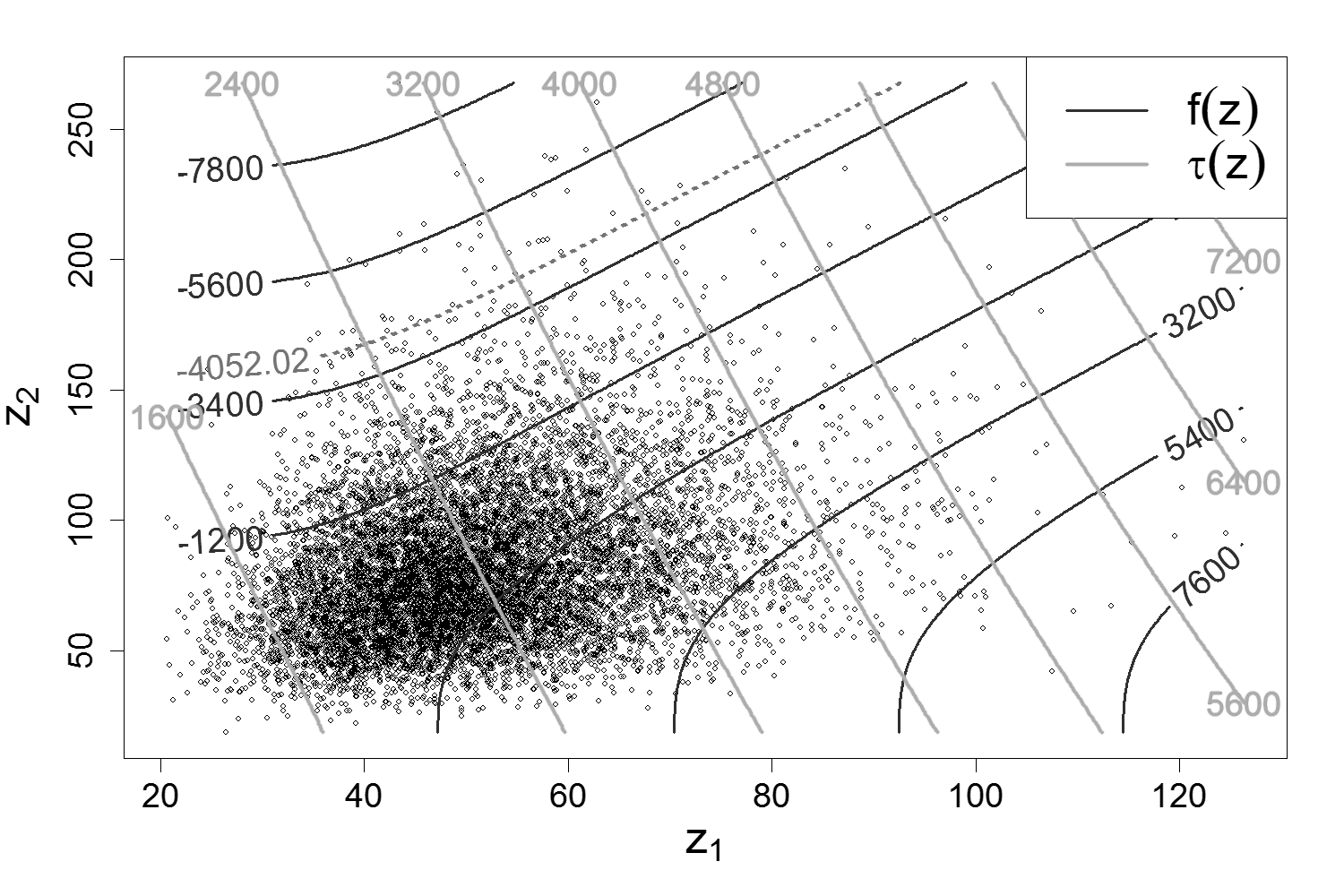}
  \caption{The true portfolio value $f(z)$ for the 2-D Black-Scholes option case study, along with the respective simulation standard deviation $\tau(z)$. In red is the true quantile loss $f^{(50)}=-4052.02$. The point cloud represents the scenario set $\cZ = \{z^n, n=1,\ldots, 10000\}$.}\label{fig:contour}
\end{figure}

In this case study, the outputs $y^{n,i}$ are obtained by simulating the log-normal values of $S^1_2, S^2_3$ conditional on $(S^1_1, S^2_1)=(z_1,z_2)$ and plugging them into the payoff
\begin{equation}\label{eq:cs1payoff}
Y = 100 e^{-\beta(T_1-T)} \left(S^1_{T_1}- 40\right)_+ - 50 e^{-\beta (T_2-T)}\left(S^2_{T_2}- 85\right)_+,
\end{equation}
where $T=1$, $T_1 = 2$ and $T_2 = 3$. We proceed to compare all the methods listed in Table~\ref{table:procedures1} using the global experiment parameters of  $N=10^4$, $\cN=10^4$, so that $\a N = 50$. \blue{Note that since $\cN = N$,  U1-GP allocates a single inner simulation per scenario. The fully sequential schemes use $N_{init} = 0.01 N = 100$ with $r_0^n = 10$ during initialization (i.e.~$\cN_1 = \cN - r_0 N_{init} = 9000$), and $K=100$ stages, so that $\Delta r_k = \cN_1/K = 90$ during all other rounds. The GP hyperparameters are refitted at stage $k=10, 20, \ldots, 100$ for each sequential method, and after each stage for A3-GP and U2-GP.}

Since we have the exact value-at-risk $R^{\VaR} \doteq f^{(50)}$ and Tail VaR $R^{\TVaR} \doteq \frac{1}{50} \sum_{n=1}^{50} f^{(n)}$, the bias and squared error (SE) of a given final estimate $\hat{R}_K$ for either $\VaR$ or $\TVaR$ are simply
\begin{align*}
\text{bias}(\hat{R}_K) \doteq \hat{R}_K - R, \qquad\text{and} \qquad
\text{SE}(\hat{R}_K) \doteq \left(\hat{R}_K - R\right)^2.
\end{align*}

For the mean function $\mu(\cdot)$ of the GP in \eqref{eq:kriging3} we took the intrinsic value of the portfolio at $T$, that is,
\begin{equation}\label{eq:meanfunction}
\mu(z_1,z_2) = 100e^{-0.04}(z_1-40)_+-50e^{-2\cdot0.04}(z_2-85)_+.
\end{equation}
The above choice was compared to a simpler constant mean function $\mu(z) = \beta_0$ with $\beta_0$ fitted along with other GP hyperparameters; the latter yielded no statistically significant differences in the resulting performance of the schemes although it did increase the GP uncertainties $s_k(z^n)$.

\subsection{Comparing Algorithms}\label{sec:cs1results}
Assessment of the algorithms is done through performing 100 macro-replications $m=1, \ldots, 100$ (i.e.~repeating the whole procedure with a fresh set of sampled $y^{n,i}$'s) for both VaR and TVaR with the fixed set $\cZ$, obtaining estimates $\hat{R}^{[1]}_K, \ldots, \hat{R}^{[100]}_K$. This yields a true sampling distribution of the estimators from various algorithms, controlling for the intrinsic variability of inner simulations. These outputs are used to compute bias, variance, and SE, and the results are illustrated through several tables and figures.  Detailed visualization of sequential methods is provided in Section~\ref{sec:surcomparison}.

Table \ref{table:CS1results} reports (i) 
$SD(\hat{R}_K^{[1:100]})$ which is the empirical standard deviation of $\hat{R}_K$ over the 100 macro replications, (ii) the average estimated GP posterior uncertainty associated to $\hat{R}_K$: $\overline{s} \doteq \frac{1}{100} \sum_{m=1}^{100} s( \hat{R}_K^{[m]})$,  (iii) the root mean squared error relative to the ground truth,  $RMSE \doteq \sqrt{\frac{1}{100} \sum_{m=1}^{100}  (\hat{R}^{[m]}_K - R)^2} $ along with (iv) average $|\cD_K|$, the number of distinct outer scenarios chosen by the algorithm.  Note that for methods other than LB and U1-GP, $|\cD_K|-100$ (where $100=N_{init}$ is the number of pilot stage-0 scenarios) is the number of locations chosen after initialization. Figures~\ref{fig:VaRresults1} and~\ref{fig:TVaRresults1} show boxplots of the resulting distributions for $\hat{R}^{[m]}_K$ and $s(\hat{R}^{[m]}_K)$, where the horizontal line is the true risk measure $R$ obtained from the analytic Black-Scholes computation. \rev{For $\VaR$ we recall that our estimators are expected to converge to the Harrell-Davis estimator $R^{\VaR, HD}$ which is used as the ground truth in our discussion below. In the present example the difference relative to the true quantile was about 20, $f^{(\a N)} = -4052.02, R^{HD} = -4032.21$.}

\begin{table}[!ht]
\centering
\begin{tabular}{l|rrrr || rrrr}
& \multicolumn{4}{c||}{$\VaR_{0.005}$} &  \multicolumn{4}{c}{$\TVaR_{0.005}$} \\
& $SD(\hat{R}^{HD}_K)$ & \multicolumn{1}{c}{$\overline{s}$} & RMSE & $|\cD_K|$ & $SD(\hat{R}_K)$ & \multicolumn{1}{c}{$\overline{s}$} & RMSE &  $|\cD_K|$ \\ \hline
LB  &  44.35  &  40.42  &  46.77  &  1  & 47.29  &  53.48  &  47.42  &  1  \\
ST-GP  &  50.57  &  48.55  &  50.59  &  121.52  &  59.12  &  55.17  &  61.46  &  118.27  \\
SE-GP  &  50.48  &  50.71  &  74.03  &  116.12   &  93.36  &  87.83  &  95.70  &  111.79  \\
SV-GP  &  56.50  &  48.28  &  60.53  &  305.03 &  55.78  &  54.76  &  56.65  &  163.08  \\
SR-GP-1  &  63.27  &  61.90  &  69.74  &  112.43  &  61.48  &  55.34  &  61.86  &  165.27  \\
SR-GP-2  &  50.45  &  49.82  &  50.52  &  180.97   &  61.66  &  61.56  &  62.13  &  193.46  \\
A3-GP  &  61.07  &  54.18  &  60.83  &  292.83 &  63.57  &  59.92  &  63.18  &  297.44  \\
U2-GP  &  68.76  &  55.91  &  68.47  &  194.55 &  64.92  &  67.77  &  64.87  &  194.64  \\
U1-GP  &  695.33  &  560.52  &  2965.05  &  $10^4$ & 909.17  &  700.43  &  3003.07  &  $10^4$  \\
\end{tabular}
\caption{For the 2-D Black Scholes case study, sample standard deviation (SD) over 100 macro-replications, average GP posterior standard deviation $\overline{s}$, and RMSE for each approach for $\hat{R}_K$, as well as average final design sizes. Description of the methods is in Table \ref{table:procedures1}.  }\label{table:CS1results}
\end{table}

\begin{figure}[!ht]
  \centering
  \includegraphics[scale=0.30]{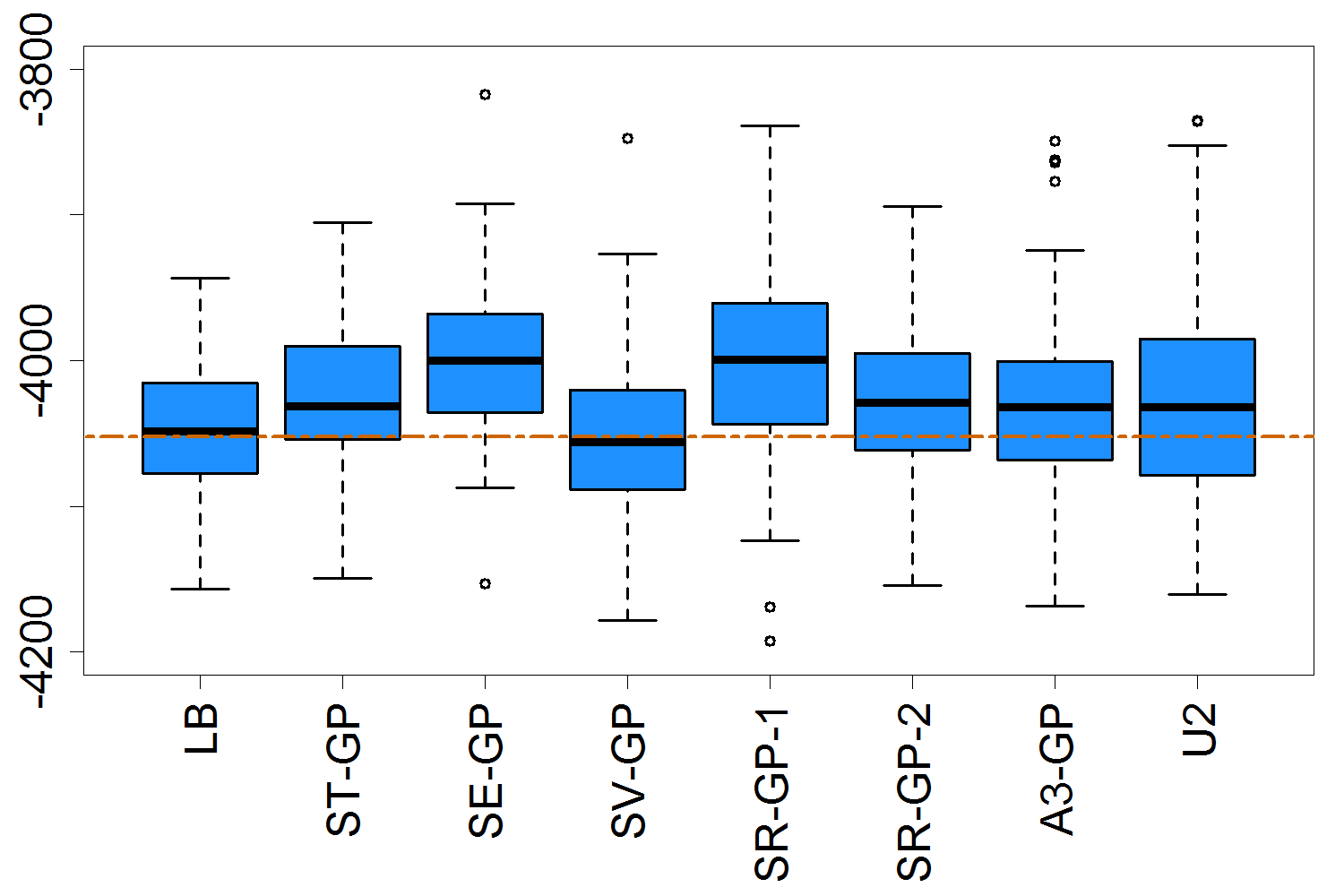}
  \includegraphics[scale=0.30]{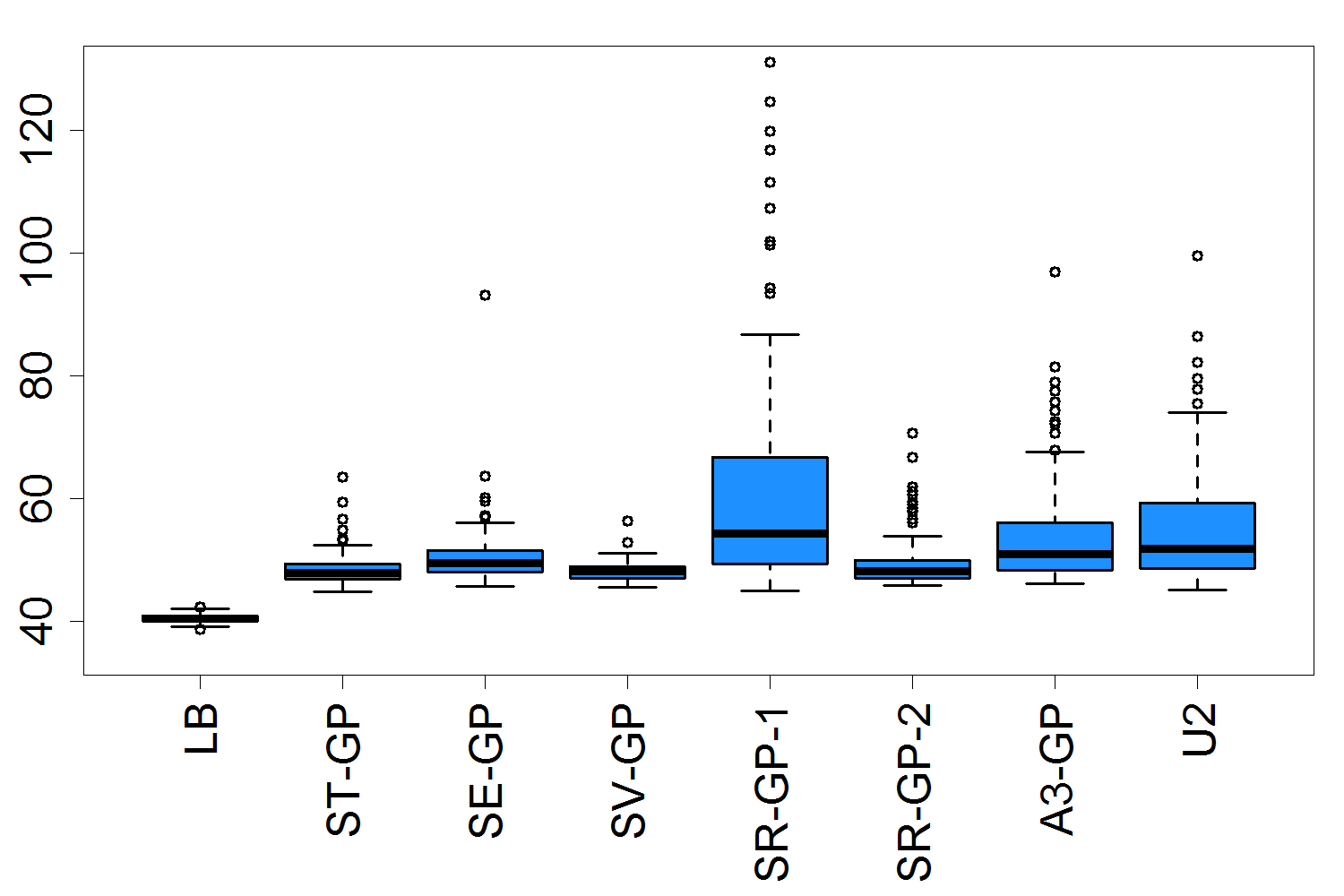}
  \caption{VaR estimation for the 2-D Black Scholes case study. Left boxplots display the distribution of the final $\hat{R}^{\VaR}_K$ estimates; on the right is distribution of the corresponding GP standard deviation $s(\hat{R}^{\VaR}_K)$. Results are based  on 100 macro-replications for each approach.}\label{fig:VaRresults1}
\end{figure}

\begin{figure}[!ht]
  \centering
  \includegraphics[scale=0.30]{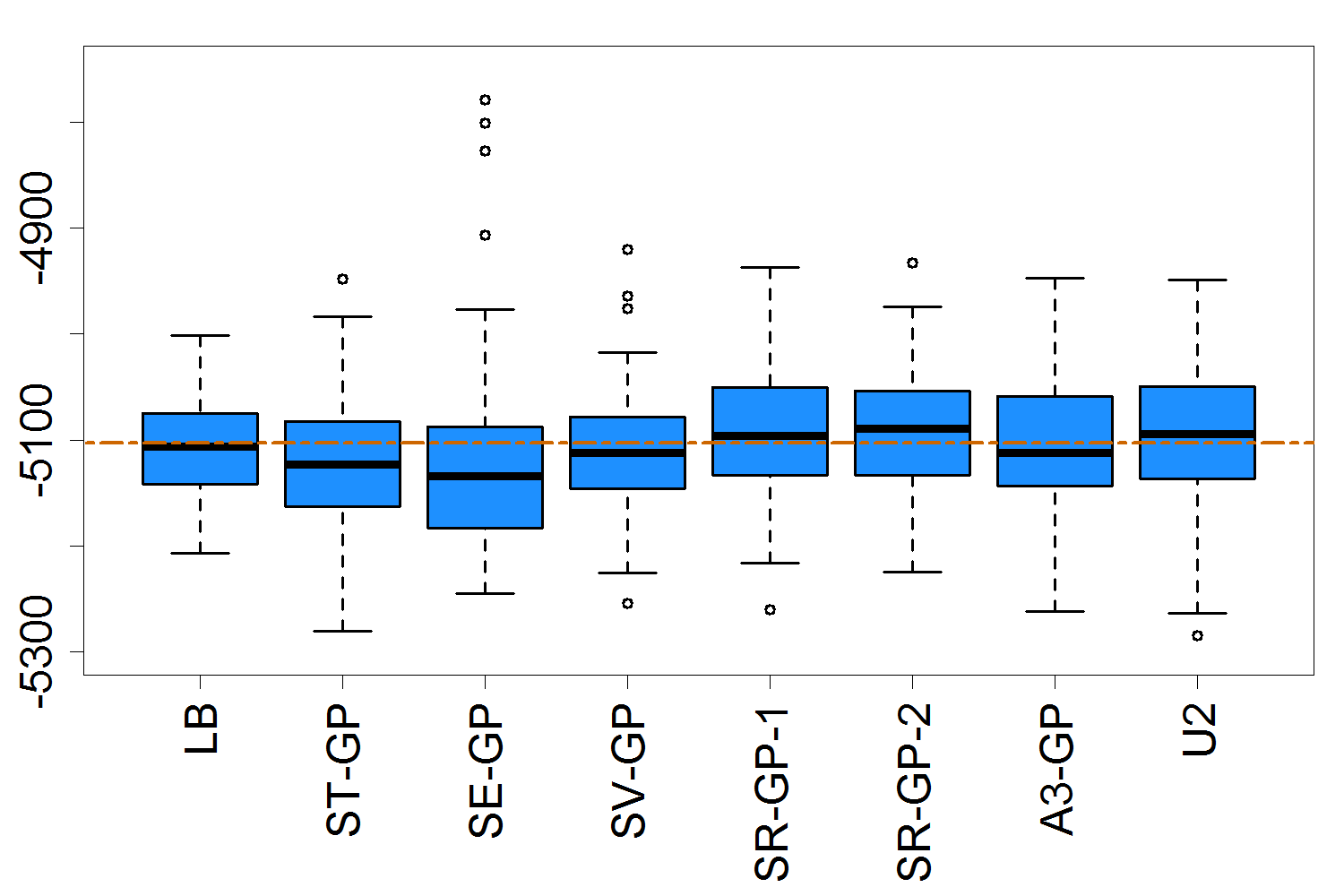}
  \includegraphics[scale=0.30]{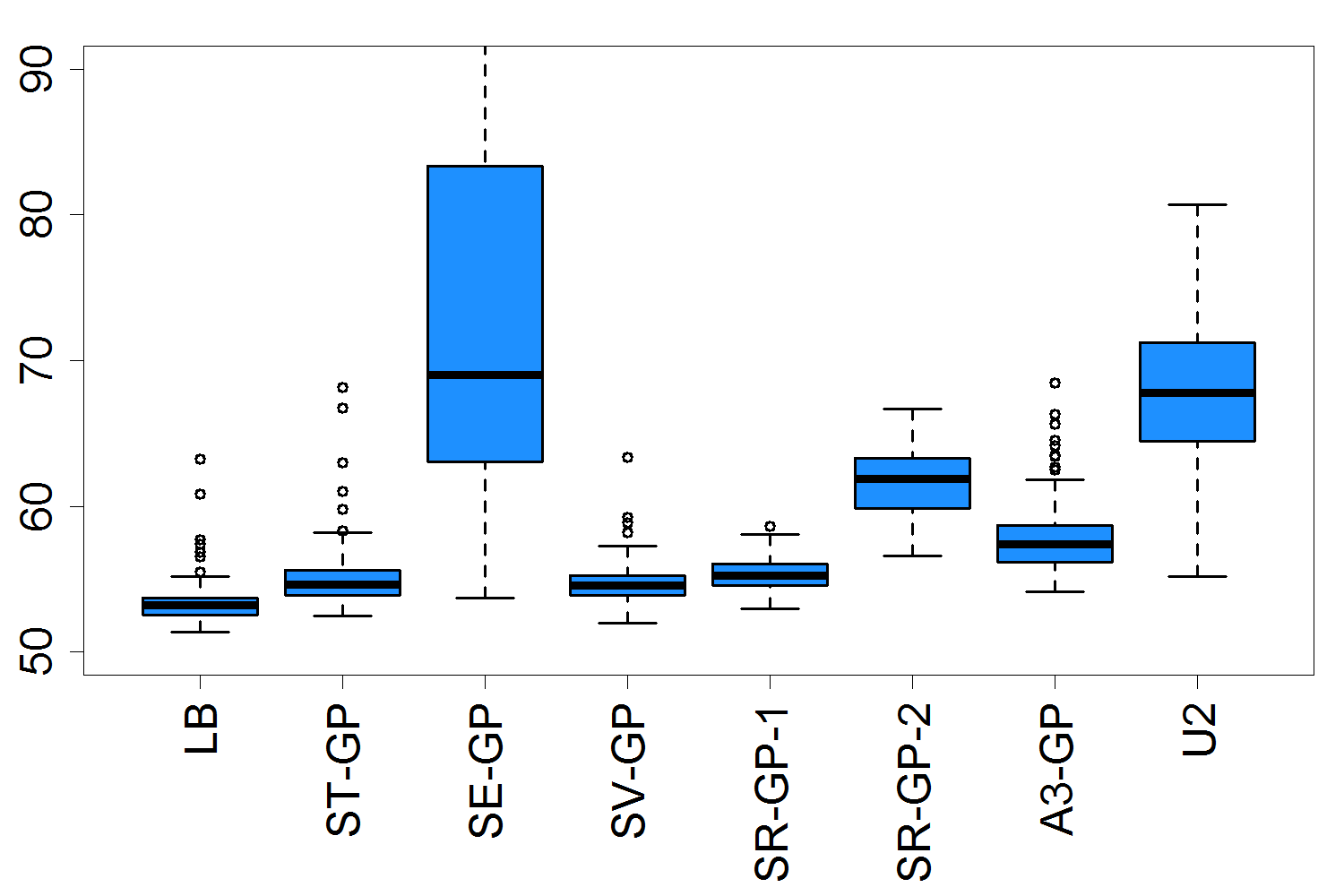}
  \caption{TVaR estimation for the 2-D Black Scholes case study. Left boxplots display the distribution of the final $\hat{R}^{\TVaR}_K$ estimates; on the right is corresponding GP standard deviation $s(\hat{R}^{\TVaR}_K)$. Results are based  on 100 macro-replications for each approach.}\label{fig:TVaRresults1}
\end{figure}

We begin discussion with Value-at-Risk.  The uniform allocation (U1-GP) results  are not plotted since they are far beyond axis limits; the tables show that its RMSE is more than 60 times larger than that of LB. All other methods offer clear improvement over traditional nested Monte Carlo. Comparing in the order of U1-GP, U2-GP and A3-GP roughly shows the improvement thanks to increasing the number of stages, with the latter gaining about 20\% in RMSE reduction relative to U2-GP. Going to a sequential scheme with 100 stages we gain another 15-20\% improvement -- compare ST-GP to A3-GP.

A few other comparisons can be made. \rev{Between the two SUR-type strategies, SE-GP does not perform well relative to ST-GP, suggesting that the latter acquisition function is more appropriate for quantile learning, likely due to its ability to account for uncertainty in the estimated $\hat{R}_K^{HD}$.   For the rank based methods, SR-GP-2 doing well (competitive with ST-GP) shows that the choices of $L=26, U=75$ work well for $\VaR_{0.005}$ in this case study, while the poor performance of SR-GP-1 is indicative of its narrow focus on the empirical quantile. By being overly aggressive, SR-GP-1 may completely miss some scenarios that constitute  $R^{HD}$, which contributes to more bias and dispersion across runs.}

Overall, in terms of RMSE, the best performing are ST-GP, SV-GP and SR-GP-2.
From the boxplots we see that all methods (except LB which as expected is unbiased) are biased high, i.e.~under-report capital requirements. Of note, SE-GP and SR-GP-1 have relatively high biases and also larger $SD(\hat{R}_K)$. Note that without spatial smoothing, quantile estimates are biased \emph{low} since empirical quantiles are more extreme relative to the ground truth. Through spatial averaging, the emulator pulls the bias the other way.

Table~\ref{table:CS1results} also compares $s(\hat{R}^{[m]}_K)$ against $SD(\hat{R}_K)$. The former is the internal emulator-based estimate of the standard error for $\hat{R}$, while the latter is the observed sampling standard deviation across the macro-replications. There is no way to calculate $SD(\hat{R}_K)$ a priori, and we hope that $s(\hat{R}_K)$ can act as a proxy. This ability to accurately report standard errors is an important part of the emulation. Of course, $s(\hat{R}^{[m]}_K)$ is itself random, so in the Table we report its average $\bar{s}$ and Figure~\ref{fig:VaRresults1} shows its own sampling distribution. Thus, the goal is to have $s(\hat{R}^{[m]}_K)$ stable across runs and with mean close to $SD(\hat{R}_K)$. We do observe that all methods have $\overline{s}$ reasonably close to $SD(\hat{R}_K^{[1:100]})$, with the biggest discrepancies occurring with SV-GP and A3-GP.  Both of these methods are designed to minimize $s(\hat{R}_k)$, which may cause these values to be biased low. From the boxplots we note that SV-GP has the most stable $s(\hat{R}_K)$ across runs, while SR-GP-1, A3-GP, and U2-GP offer least reliable standard error estimates.

For Tail Value-at-Risk, the results are nearly identical. We again witness poor performance of SE-GP (RMSE 2.01 times larger than LB) which confirms that the respective acquisition function $\cH^{ECI}$ is inadequate. Also note that in the TVaR context, the more aggressive SR-GP-1 outperforms SR-GP-2. This illustrates that SR-GP requires careful fine-tuning of the $L$ and $U$ values.

\begin{figure}[!ht]
  \centering
  \begin{tabular}{m{0.07\textwidth}m{0.18\textwidth}m{0.15\textwidth}m{0.15\textwidth}m{0.15\textwidth}m{0.15\textwidth}}
 &  $\qquad\;\;$ ST-GP & $\quad$ SE-GP & $\quad$ SV-GP & SR-GP-1 & SR-GP-2\\
$\VaR_{\a}$ &  \includegraphics[scale=0.28,trim=0.1in 0in 0.1in 0in]{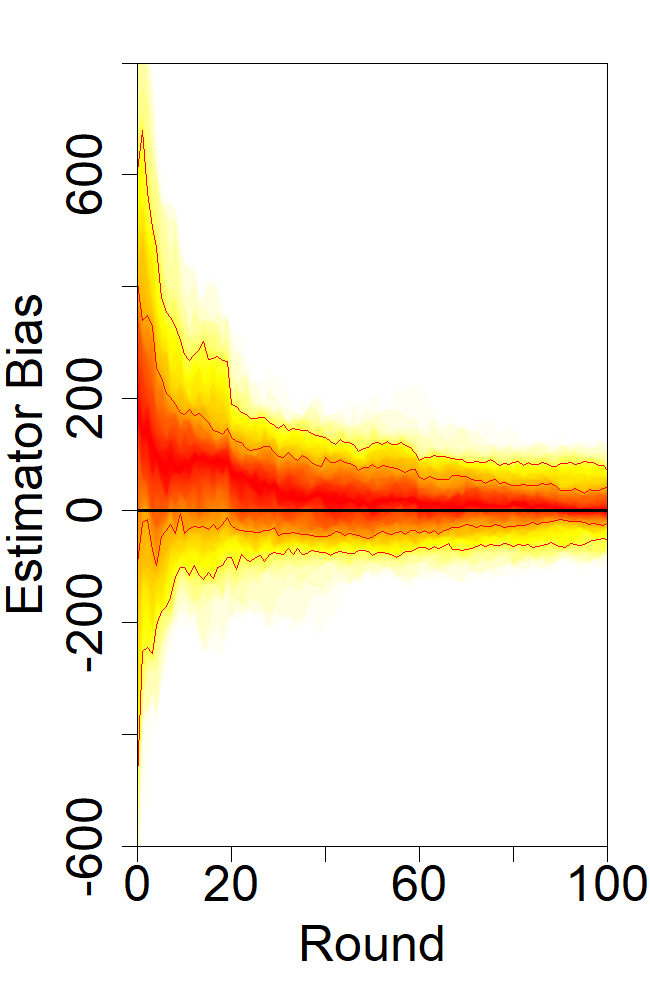}&
  \includegraphics[scale=0.28,trim=0.2in 0in 0.2in 0in]{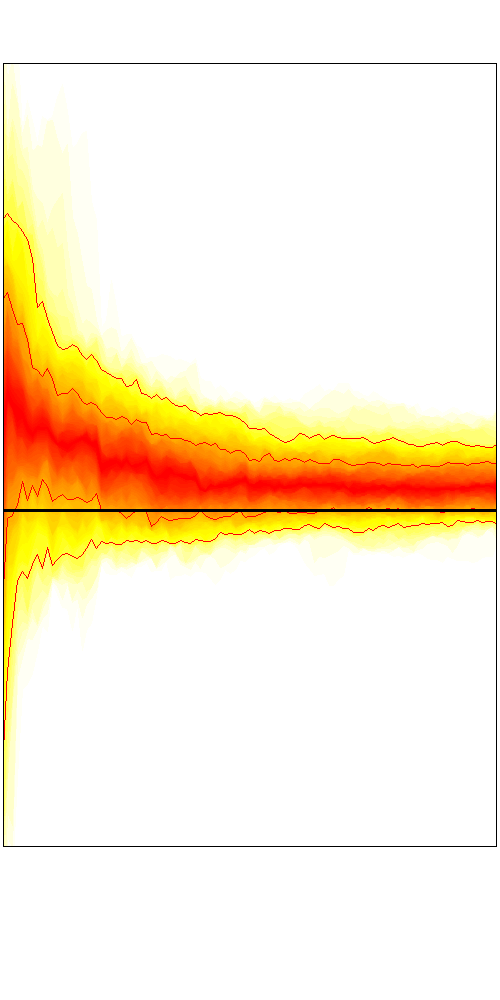}&
  \includegraphics[scale=0.28,trim=0.4in 0in 0.4in 0in]{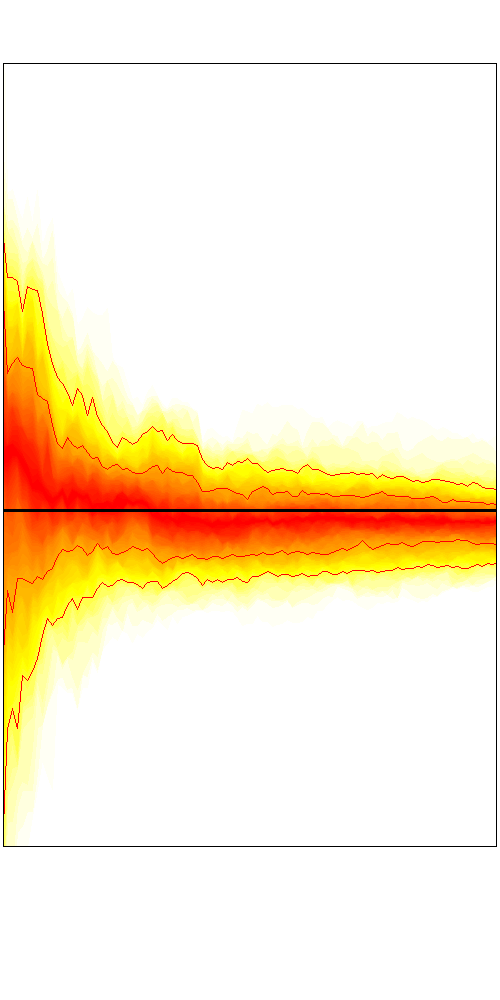}&
      \includegraphics[scale=0.28,trim=0.4in 0in 0.4in 0in]{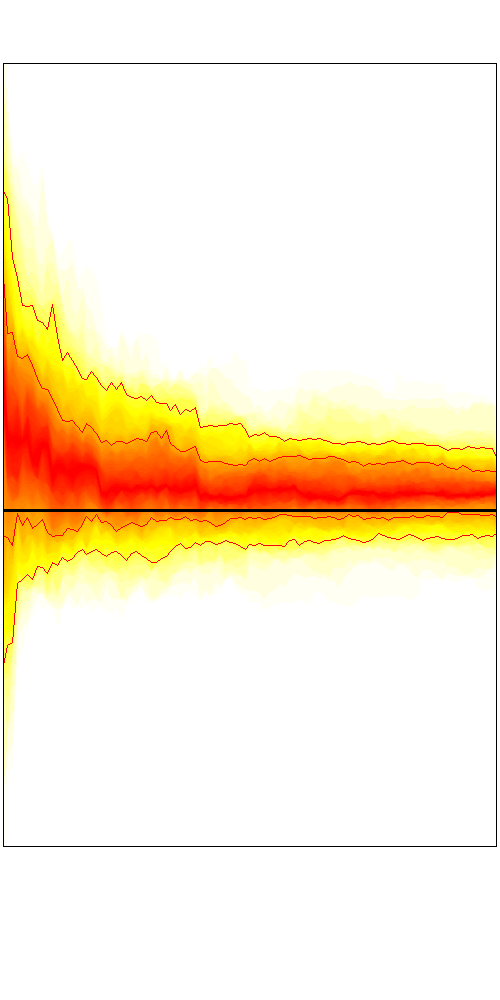}&
      \includegraphics[scale=0.28,trim=0.4in 0in 0.4in 0in]{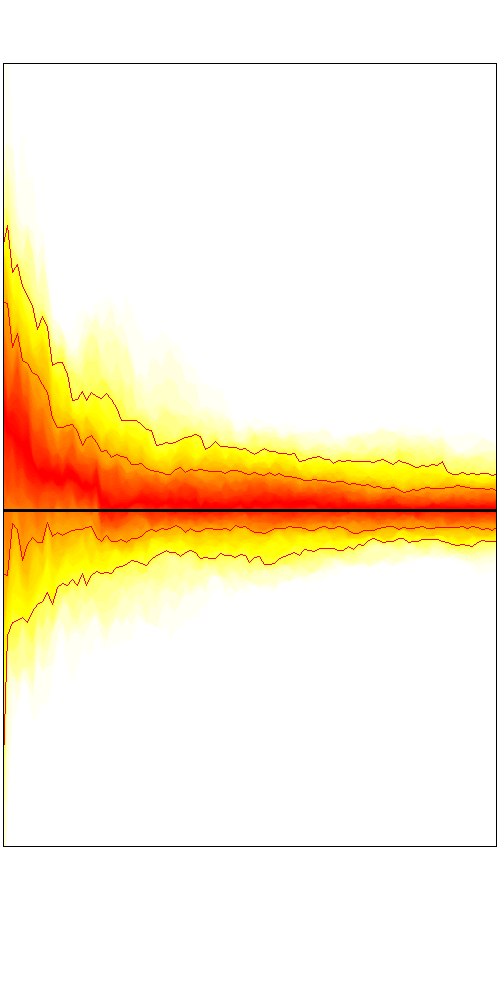} \\
$\TVaR_{\a}$  &       \includegraphics[scale=0.28,trim=0.1in 0.2in 0.1in 0in]{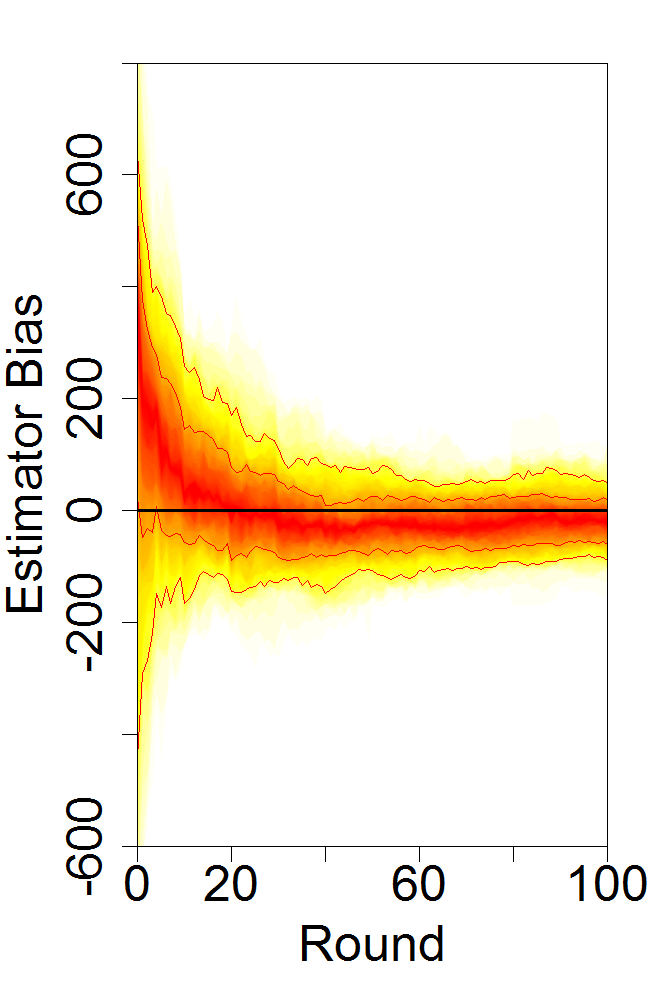}&
  \includegraphics[scale=0.28,trim=0.2in 0.2in 0.2in 0in]{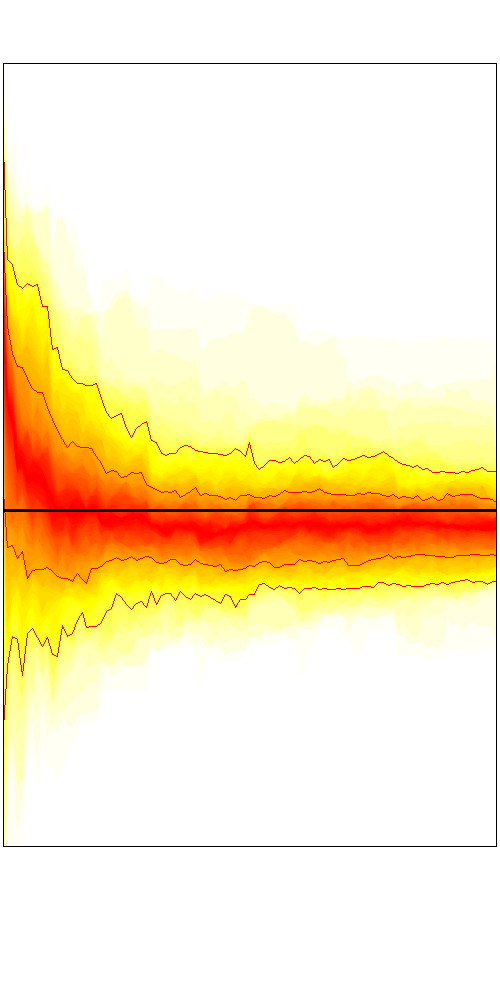}&
  \includegraphics[scale=0.28,trim=0.4in 0.2in 0.4in 0in]{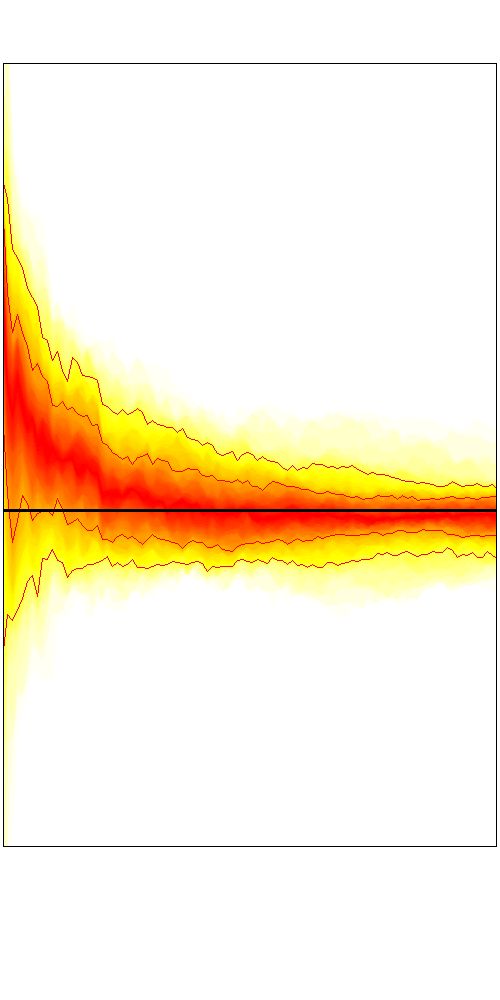}&
      \includegraphics[scale=0.28,trim=0.4in 0.2in 0.4in 0in]{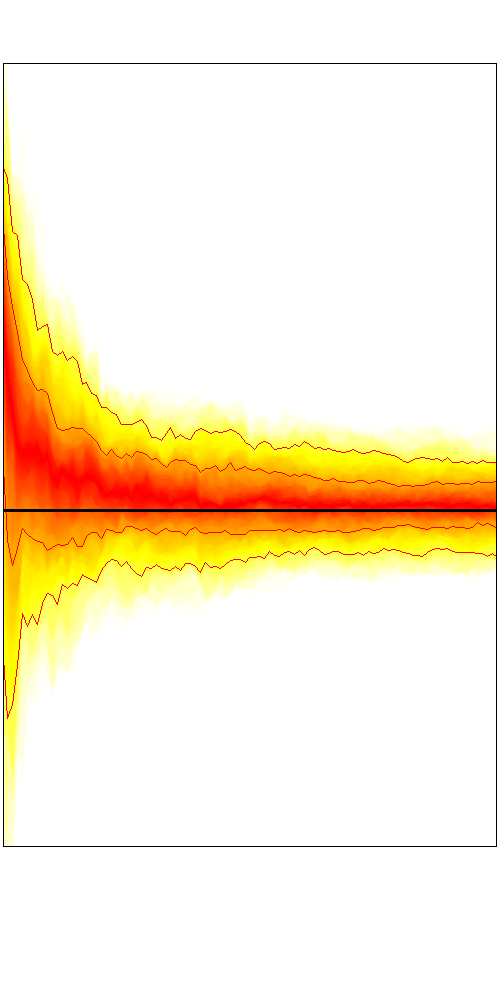}&
      \includegraphics[scale=0.28,trim=0.4in 0.2in 0.4in 0in]{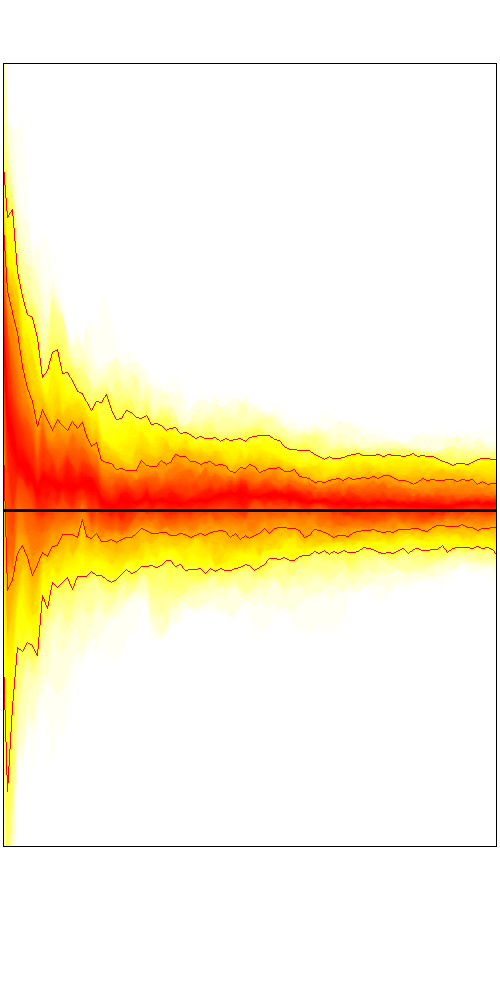}
  \end{tabular}
  \caption{For the 2-D Black Scholes case study, fan plots describing evolution of $\hat{R}^{HD,\VaR}_k$ (top row) and $\hat{R}^{\TVaR}_K$ (bottom) as budget is spent.  Red lines correspond to the 0.1, 0.25, 0.5, 0.75, and 0.9 componentwise quantiles of $\hat{R}^{[m]}_k$ over the 100 algorithm runs.}\label{fig:VaRFanPlot}
\end{figure}

{A final visual comparison of the sequential methods is provided through fan plots in Figure \ref{fig:VaRFanPlot} which illustrates evolution of the estimated risk measure $\hat{R}_k$ as the rounds $k$ progress. This is one of the major advantages of a sequential procedure which allows ``online'' use of the algorithm. Thus the user can monitor $\hat{R}_k$ for example to judge the convergence, adaptively stop the simulations, or report interim estimates. The fan plots show quantiles (in terms of $m$) over macro replications of $R^{[m]}_k$ as a function of $k$.} \rev{Most methods are initially biased high, the bias vanishes as $k$ increases and sampling variance decreases. SV-GP reduces bias the fastest. Of note, SR-GP-1 does not converge on some runs, highlighting the importance of exploration. }

\blue{For $\TVaR$ we observe that $bias(\hat{R}_k)$ is generally low already after a few rounds, indicating the easier task of estimating a level-set compared to estimating a quantile. We also see that ST-GP succeeds in learning $\TVaR$ very quickly initially but then its performance plateaus; by $k=100$ SV-GP achieves better bias and  lower $SD(R_k)$. The worse performance of SE-GP for $\TVaR$ , especially in terms of very high $SD(R_K)$ confirms the previous discussion.}

To check the effect of the batch size $\Delta r$, we re-ran ST-GP with $\Delta r = 0.001 \cN_1$ (i.e.~9 new inner simulations per round), and observed $\bar{s}$ and RMSE to be 48.37 and  59.22 respectively.  These values are statistically indistinguishable from the same ST-GP scheme with $\Delta r = 0.01 \cN_1$.  In other words, a batch size of $\Delta r = 0.001 \cN_1$ (i.e.~$K=1000$ rounds) versus $\Delta r = 0.01 \cN_1$ does not provide a significant improvement in this experiment.  We found that the smaller batches ended up yielding a procedure that picked the same scenario several times in a row, so that it behaved similarly to a larger batch size anyway.  

\subsection{Comparing Sequential Schemes}\label{sec:surcomparison}

To complement Figure~\ref{fig:acquisitionplot} that gives a snapshot of scenarios targeted by different schemes in a fixed dataset,
Figure~\ref{fig:replicationplot} shows the overall budget allocation  $r^{1:N}_K$ produced by the ST-GP, SE-GP and SV-GP at the end of a representative run. For $\VaR_{0.005}$, ST-GP and SE-GP perform similarly with concentration of effort around rank $50 = \a N$.  For reference, 93\%, 73\% and 99.1\% of the $\cN_1$ budget after initialization was allocated to $\{z : f^{(25)} \leq f(z^n) \leq f^{(75)}\}$ for ST-GP, SE-GP and SV-GP respectively, with maximum replication amounts $\max_n r^n_K$ of 3790, 1260 and 990. This highlights the high degree of adaptivity with up to 30\% of the budget spent on a single scenario (comparable to the LB benchmark which spends 100\% on the quantile scenario).

For TVaR, 95\%, 100\% and 98.1\% of the $\cN_1$ budget was spent among $\{z : f^{(1)} \leq f(z^n) \leq f^{(50)}\}$ for ST-GP, SE-GP and SV-GP respectively with maximum allocations of 1890, 4600 and 531. In particular, ST-GP and SV-GP succeed in identifying and sampling (with more or less comparable $r^n_K$'s) from nearly all scenarios in the left tail, while SE-GP leads to a hit-and-miss design as far as the true locations of $f^{(1:50)}$ are concerned. Although all of SE-GP samples were in the true tail, its design tended to be extremely concentrated (only about a dozen non-pilot scenarios added), creating small spatial clusters with a single scenario explored in each cluster.

\begin{figure}[!ht]
  \centering
  \begin{tabular}{cc}
  VaR & TVaR\\
  \includegraphics[scale=0.29]{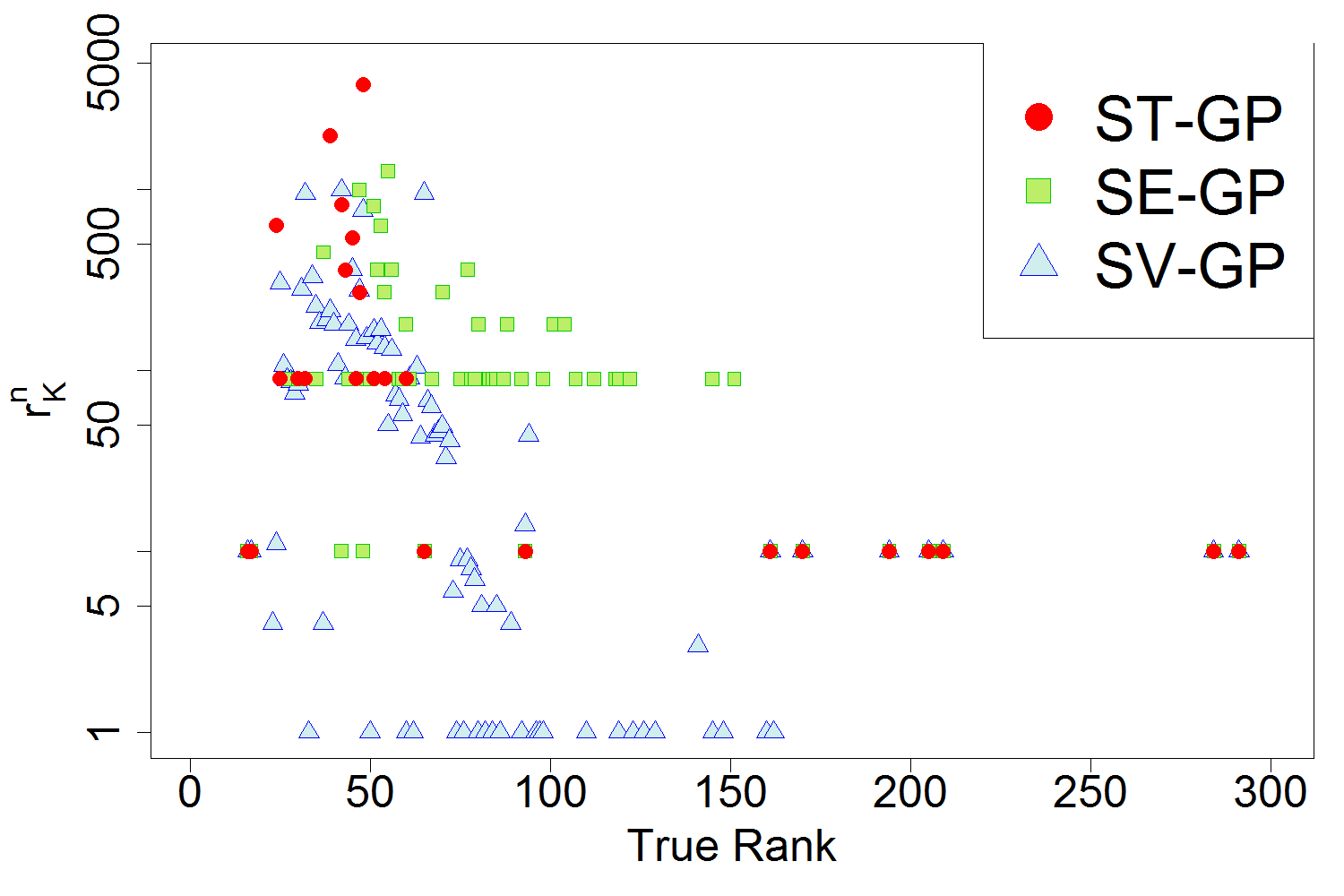} &
  \includegraphics[scale=0.29]{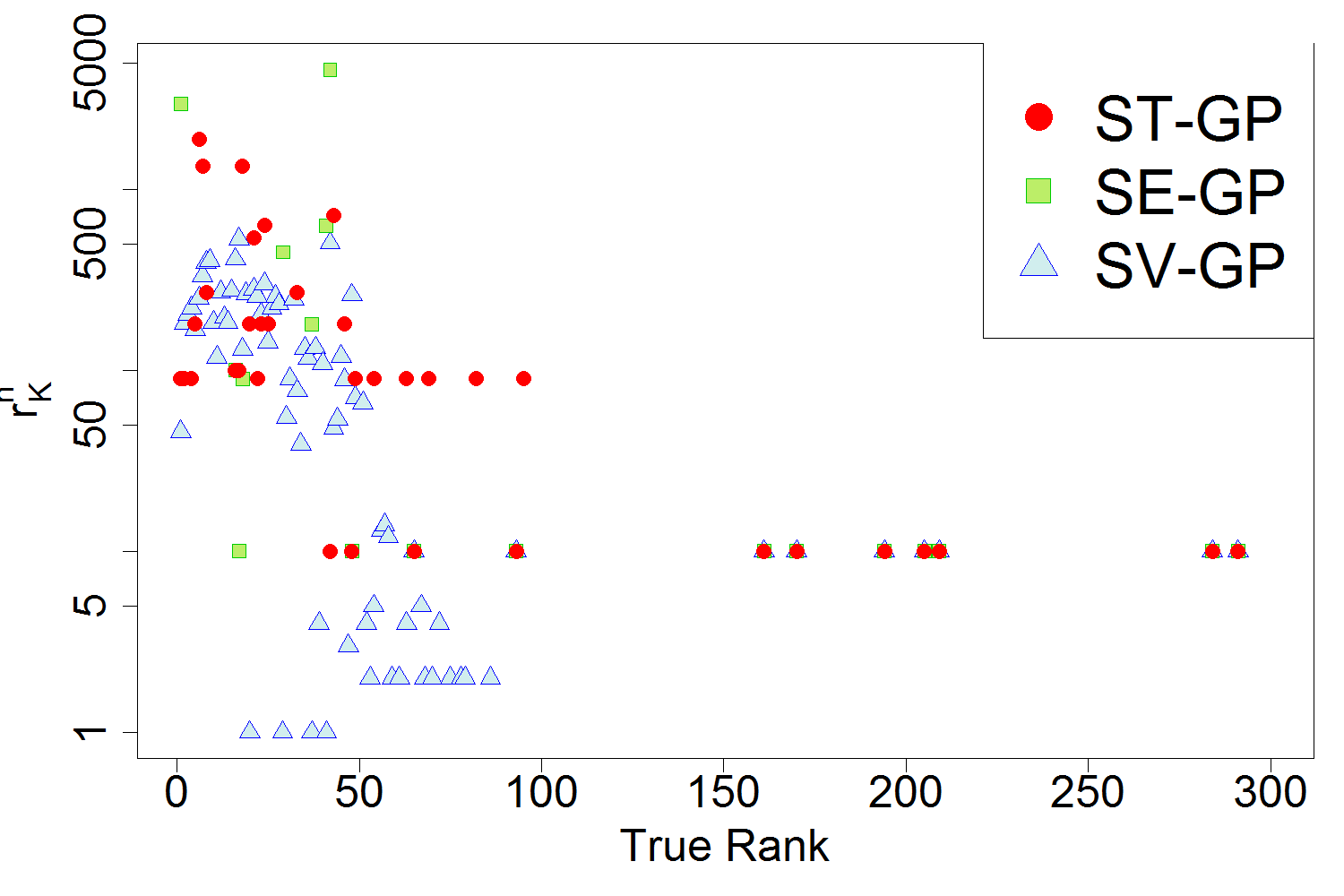}
  \end{tabular}

  \caption{Replication counts $r^n_K$ versus true rank of $f^{1:N}$ after the final stage for sequential methods for the 2-D Black-Scholes case study.  $y-$axis is on the log scale. All three schemes share the exact same initialization stage which can be seen via scenarios (especially on the right side of each plot) with $\Delta r_0 = 10$. Otherwise, for ST-GP and SE-GP all $r^n_K$'s are multiples of $\Delta r = 90$. For SV-GP there are no constraints on $r^n_K$ which can be as low as 1.  }\label{fig:replicationplot}
\end{figure}

From a different perspective, Table~\ref{table:CS1results} lists the average total design size $|\cD_K|$. We observe that ST-GP and SE-GP only use about 120 scenarios (recall that $N_{init}=100$ scenarios were already selected during initialization), while SV-GP uses about 300.
This is because SV-GP by construction allocates $r^{'n}_k$ across multiple scenarios, which can be also observed in Figure~\ref{fig:replicationplot}. In particular there are many scenarios that were just ``probed'' by SV-GP $r^n_K = 1$ (cf.~discussion on $\cZ^{SV}$ in Section~\ref{sec:dad}) but not really used. The resulting larger design size translates into larger computational overhead. Notably, SE-GP is more concentrated than ST-GP and does not appear to explore $\cZ$ sufficiently.

\begin{figure}[!ht]
  \centering
  \begin{tabular}{cc}
  $\VaR_{0.005}$ & $\TVaR_{0.005}$ \\
  \includegraphics[scale=0.29,trim=0in 0.35in 0in 0in]{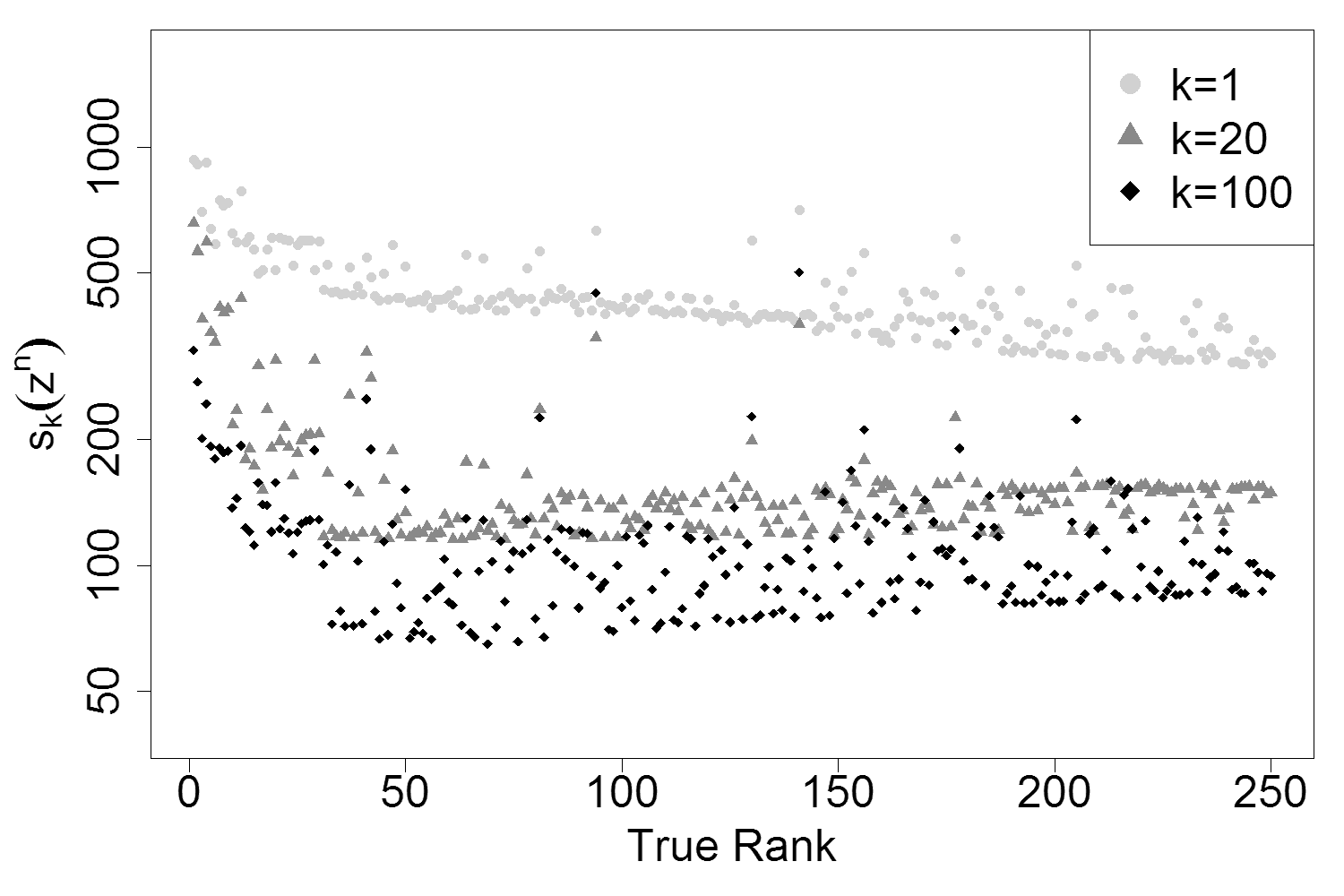}&
  \includegraphics[scale=0.29,trim=0in 0.35in 0in 0in]{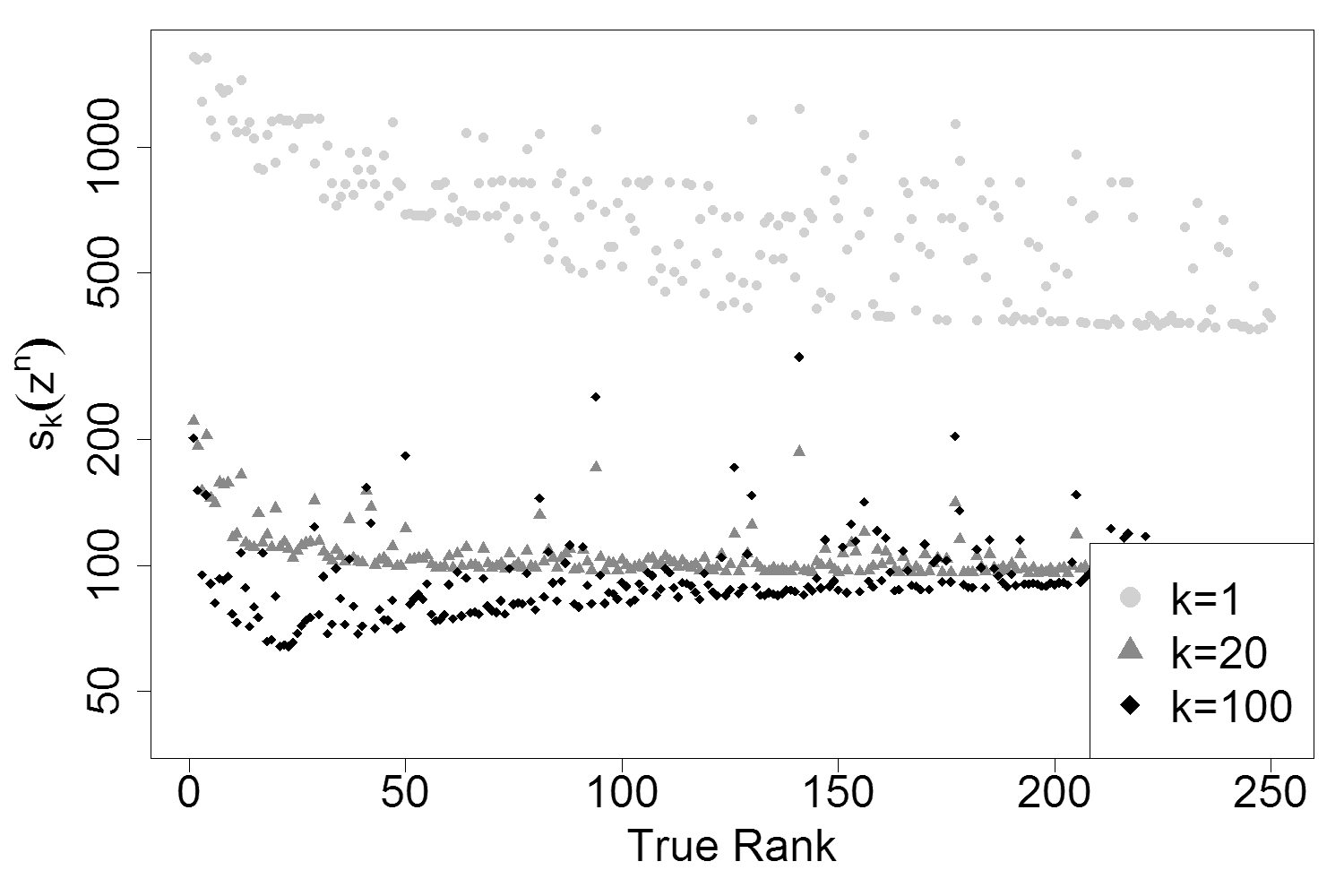}
  \end{tabular}
  \caption{Using SV-GP, posterior GP standard deviation $s_k(z^n)$ at stage $k=1,20,100$ for one run, sorted according to the true rank of $f(z^{1:N})$.  $y-$axis is on the log scale. Lower $s_k(z^n)$ indicates higher density of inner simulations in the spatial neighborhood of $z^n$.}\label{fig:s-evolution}
\end{figure}

Figure \ref{fig:s-evolution} illustrates evolution of $s_k(z)$ at $k=1, 20, 100$ using the SV-GP scheme, for both VaR and TVaR.  As $k$ increases, posterior uncertainty $s_k(z^n)$ shrinks; we furthermore see the targeted allocation of inner simulations with $s_k(z^n)$ lowest near the quantile $\cQ$ (true rank 50) for $\VaR$ and throughout the left tail (true rank $\le 50$) for $\TVaR$.  There is a clear distinction between the two panels, with $\VaR-{0.005}$ focusing more on rank 50 than rank 1-10, and vice-versa for $TVaR_{0.005}$. This effect  is also observed in Figure~\ref{fig:meanvsrank}: variance and bias are minimized at $\cQ$ for $\VaR_{\a}$ and in the entire left tail for $\TVaR_{\a}$.

\subsection{Gains from Spatial Modeling}\label{sec:gains}

\begin{figure}[!ht]
  \centering
  \begin{tabular}{cc}
  VaR & TVaR \\
  \includegraphics[scale=0.3,trim=0in 0.35in 0in 0in]{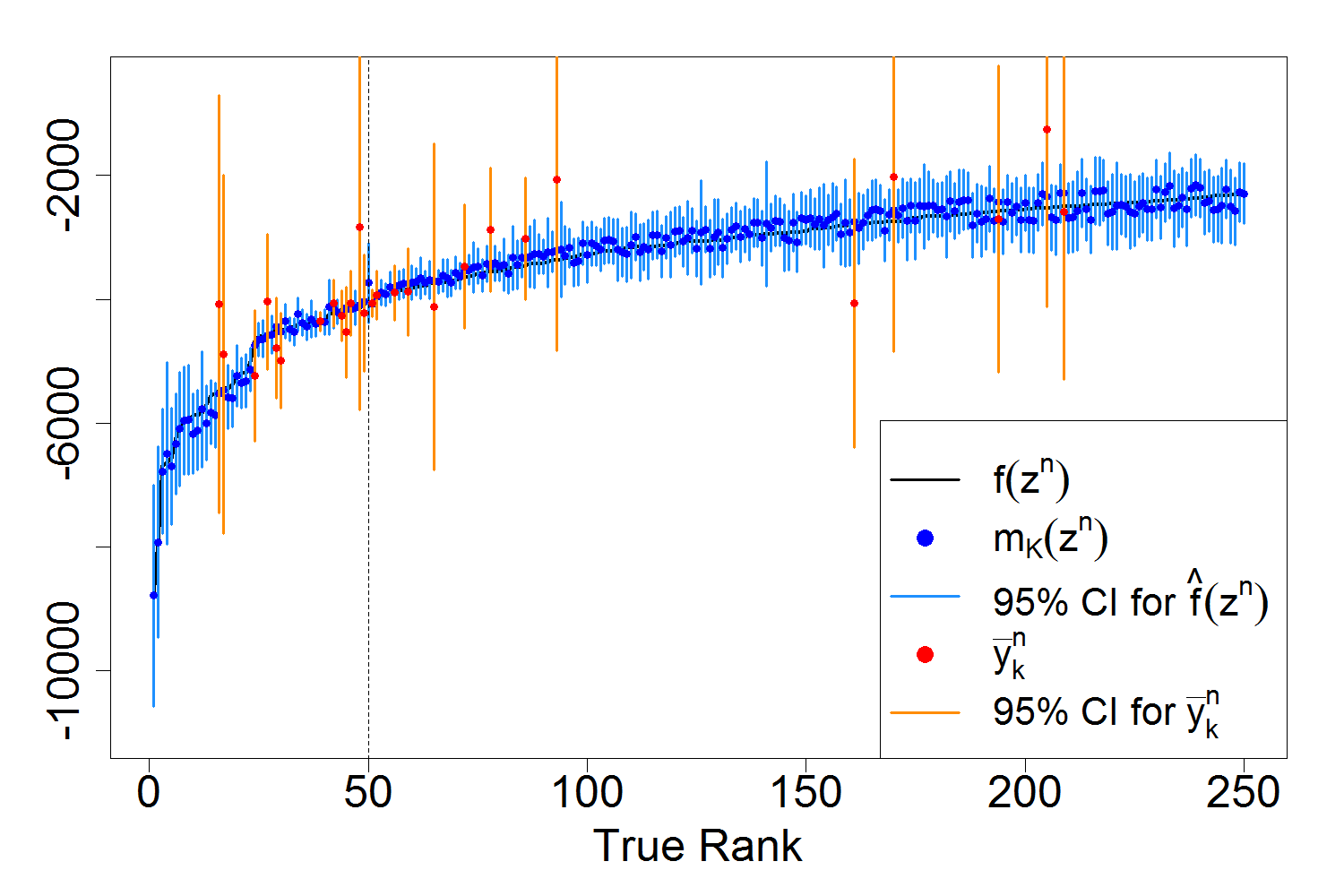} &
  \includegraphics[scale=0.3,trim=0in 0.35in 0in 0in]{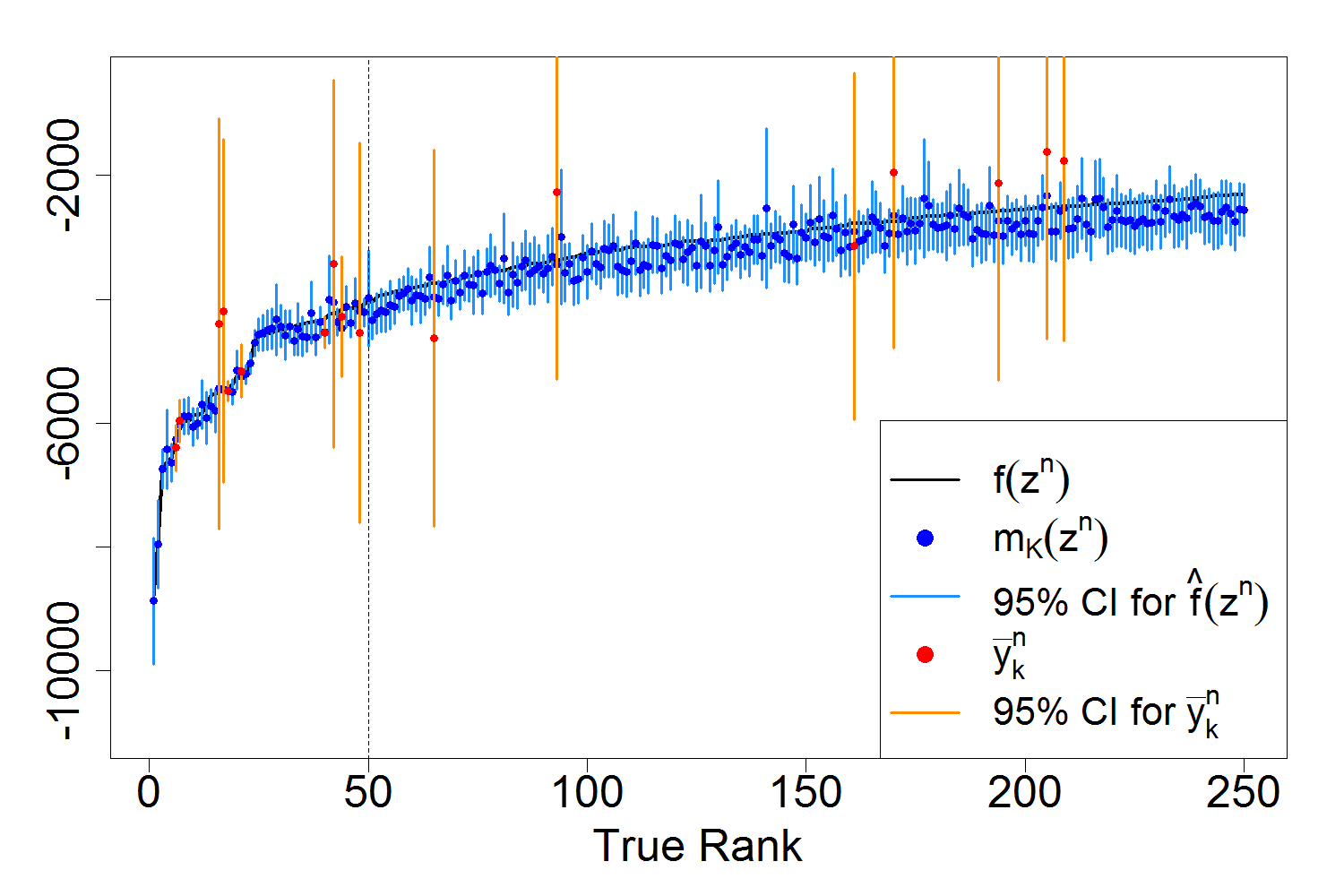}
  \end{tabular}
  \caption{Using ST-GP, output for single run after $K=100$. Both panels show the true $f(z)$, with point estimates $m_K(z)$ and $\bar{y}^n_K$, and 95\% credible/confidence intervals, all sorted according to the true rank of $f(z^{1:N})$.  Around each $\bar{y}_K^n$ is a 95\% error bar using the hetGP estimate $\tilde{\tau}(z^n)/r^n_K$ for the simulation variance while each $m_K(z^n)$ comes with error bars based on $s_K(z^n)$. }\label{fig:meanvsrank}
\end{figure}

Figure \ref{fig:meanvsrank} visualizes the spatial features of the GP surrogate at the conclusion of $K=100$ rounds for a typical run of ST-GP. Recall that our \texttt{hetGP} model smoothes both the sample averages $\bar{y}_K^n$ and the sample variances $\hat{\tau}^2_K(z^n)$. This has the double effect of improving the accuracy of learning $f(z^n)$ and lowering the uncertainty of that estimate. Indeed, in the Figure we may see that $m_K(z^n)$ is much closer to $f(z^n)$ relative to the empirical $\bar{y}_K^n$'s and moreover, the GP posterior variance $s_K(z^n)$ is much lower than the empirical $\hat{\tau}_K(z^n)$. We stress that the Figure flattens the underlying 2-D space into a one-dimensional realization, so $z$'s of neighboring rank might be far apart spatially and hence carry quite different $(m_K(z), s_K(z))$ values. In particular, scenarios with large error bars are generally spatially distant from the true quantile contour (left tail, in the TVaR case), whereby the algorithm discovered it was not worth sampling and kept $r^n_K$ low. For scenarios with high $r^n_k$ the two error bars are closer. Large uncertainty and bias are apparent on the \emph{right} side of both panels since the algorithms are explicitly targeting the left tail and hence sacrifice accuracy elsewhere. This feature is also observed on the extreme left for $\VaR$.

To better quantify  the improvement in learning $R$ thanks to the spatial borrowing of information across scenarios hinted in Figure~\ref{fig:meanvsrank}, we first compare two non-adaptive algorithms: standard nested sampling dubbed U1-SA which works with sample averages $\bar{y}$, versus U1-GP which smoothes them via a GP surrogate. This comparison crystallizes the pure ``spatial'' effect without any confounding due to the adaptive design. Table~\ref{table:MonteCarlo} reports the results for various simulation budgets $\cN$.  We see that the use of a spatial model allows a speed-up of about $5\mathrm{x}$ in terms of improved bias and RMSE. For example, U1-SA with $\cN=500,000$ has RMSE of 152.49 which is comparable to the RMSE of U1-GP with just a fifth as many simulations (149.09 for $\cN = 100,000$). For really large budgets, the gap narrows since sample averages start to yield decent estimates of $f(z)$. Table~\ref{table:MonteCarlo} also compares the average local empirical variance at the sample quantile $\hat{\tau}( \hat{R}^{HD})$ (used by U1-SA to compute standard errors) against $\bar{s}$, the average GP-based standard error of $\hat{R}^{HD}_K$. We observe an uncertainty reduction of more than 50\% for same budgets, again thanks to spatial borrowing. 

\begin{table}[!ht]
\centering
\begin{tabular}{r||rr||rr||rr}
& \multicolumn{2}{c||}{U1-SA} & \multicolumn{2}{c||}{U1-GP} & \multicolumn{2}{c}{BR-SA} \\
$\cN$ & $\overline{\hat{\tau}_1(\hat{R}_1^{HD})}$ & RMSE &  $\overline{s}$ & RMSE & $\overline{\hat{\tau}_1(\hat{R}_1^{HD})}$ & RMSE \\ \hline
$1 \cdot 10^4$ &  \textemdash &  6578.12  &  585.01  &  2965.05 & \textemdash & \textemdash  \\
$2 \cdot 10^4$ &  1263.18  &  3660.92  &  359.05  &  1385.44   &  \textemdash  &  \textemdash   \\
$5 \cdot 10^4$ &  650.62  &  1569.38  &  202.97  &  344.55 &  135.60  &  186.82 \\
$1 \cdot 10^5$ &  410.72  &  759.87  &  146.49  &  148.62  &  53.50  &  61.58  \\
$2\cdot 10^5$ &  272.09  &  384.27  &  107.35  &  112.63  &  34.72  &  43.16  \\
$5 \cdot 10^5$ &  162.42  &  163.86  &  81.26  &  77.22 &  19.45  &  21.58 \\
$1 \cdot 10^6$ &  112.89  &  92.44  &  54.28  &  66.69 &  12.98  &  13.87  \\
\end{tabular}
\caption{For the 2-D Black Scholes portfolio case study, average values over 100 macro replications of posterior standard deviation and RMSE for $\hat{R}_K$ using nested MC (U1-SA), the BR-SA algorithm in Broadie et al.~\cite{broadie2011efficient} and uniform GP surrogate, U1-GP.  U1-SA is the standard nested Monte Carlo based on sample averages; U1-GP uses a GP fitted to the simulation output based on single-stage uniform allocation, and BR-SA uses MC sample averages along with a sequential R\&S procedure, see Section \ref{sec:benchmarks}. For $\cN=10^4$ the inner simulation budget is $r^n=1$, not allowing an estimate of $\tau(\cQ)$ for -SA methods.  }\label{table:MonteCarlo}
\end{table}

From a different angle we can compare the BR-SA method, which does non-uniform allocation of inner simulations, against, say, ST-GP. While BR-SA dramatically outperforms U1-SA
(for instance BR-SA with $\cN = 5 \cdot 10^4$ is comparable to U1-SA with a budget 10 times larger), it is still far behind the surrogate-based sequential approaches. It takes a 10-20 times larger budget for BR-SA to begin matching the RMSE's of the methods from the previous section that only employed $\cN = 10^4$.
So while  BR-SA carries essentially no regression/sequential design overhead, the cost in terms of increasing simulation effort is enormous. To conclude, spatial modeling yields an \emph{order-of-magnitude} simulation savings in this example. (We also remark that by its construction, BR-SA first allocates two  inner simulations at each scenario and only then enters its adaptive phase. So for $\cN \le 2N$ it reduces to U1-SA; our schemes will produce adaptive allocations even under small simulation budgets.)

\subsection{Choice of GP Emulator}\label{sec:het-vs-sk}

To use a GP emulator, one must pick a covariance kernel and also a specific method for fitting the model to the simulation outputs. To investigate the respective impact,
we compare use of two different kernel families---Mat\'ern-5/2 versus Gaussian, cf.~\eqref{eq:maternkernel2} and \eqref{eq:Gsn-kernel}---and two different approaches for modeling simulation variance $\tau^2(z)$. Specifically we compare the stochastic kriging approach of using point estimates $\hat{\tau}^2(z^n)$ to estimate $\tau^2(z^n)$ (implemented in \texttt{DiceKriging} \texttt{R} package) and the joint response-noise surface approach of \texttt{hetGP}. The latter is more complex and hence brings more computational overhead. Table \ref{table:hetGPcomparison} provides final RMSE, as well as bias at several stages when using the ST-GP acquisition function with the above three choices of the GP emulator.

\begin{table}[!ht]
\centering
\begin{tabular}{cl|r|rrrrr}
&&& \multicolumn{5}{c}{Bias} \\
Approach & Kernel & RMSE  & $k=1$ & $k=10$ & $k=20$ & $k=50$ &  $k=100$ \\ \hline
hetGP & Mat\'ern-5/2 &   57.52  &     166.065 &  113.925 &  97.751 &  37.158 &  28.066\\
hetGP & Gaussian &     68.06  &     91.475 &  104.427 &  74.874 &  52.716 &  48.249\\
SK & Mat\'ern-5/2   &  69.31    &   914.728 &  206.267 &  113.716 &  69.821 &  48.670 \\ \hline
\end{tabular}
\caption{For the 2-D Black Scholes portfolio case study, average RMSE of $\hat{R}_K$ across different GP models/kernel families. We also report average bias $\text{bias}(\hat{R}_k)$ across a selection of intermediate stages $k=1,10,20,50,100$. All methods use the ST-GP rule and are based on 100 macro-replications.  }\label{table:hetGPcomparison}
\end{table}

We observe that using the empirical variance estimates $\hat{\tau}^2(z^n)$ (SK) generates glaringly high biases at initial stages. In turn, the \texttt{hetGP} model is able to de-bias efficiently. Although bias falls rapidly as the algorithm progresses, the SK emulator still has a 50\% higher bias (and 20\% higher RMSE) at \rev{the} $k=100$ stage. Because in practice one  does not know when the bias has been effectively alleviated, this is a significant shortcoming of working with an SK emulator and highlights the importance of properly modeling the noise surface $\tau(\cdot)$. In comparing kernels, the Gaussian kernel shows slightly higher RMSE and bias compared to Mat\'ern-$5/2$, though not by a significant margin. This matches the folklore view that Mat\'ern is the default choice for the kernel family. 

	\section{Case Study: Life Annuities under Stochastic Interest Rate and Mortality}\label{sec:casestudy2}

 To check if the relative performance of each method remains the same in a more elaborate setting, we move to a case study in higher dimensions with a more complicated payoff function. Despite increased complexity we maintain the same budget $\cN = 10^4$ which is expected to magnify the relative differences across proposed methods.

The setup in this section considers an annuitant who enters into a contract to begin collecting payments in $S$ years, whence the payments continue annually until death of the individual.  (In practice some cutoff age $x_{u}$ is set for the final payment.)  Regulations require analysis of quantiles of the $T=1$-year value of this contract to the insurer. The major drivers of portfolio loss are interest rate risk (low interest rates increasing the present value of annuity payments) and mortality risk (increased longevity raising the value of the annuity). These factors are captured by 6-D stochastic state: a three-factor model for interest rates and a three-factor stochastic mortality model. 

We wish to find the present value of this annuity at horizon $T \leq S$, before payments begin at $S$.
To this end, let $\mathcal{T}(x)$ be the remaining random lifetime of an individual aged $x$ today. The annuity payment at $t$
is predicated on $\mathcal{T}(x) > t$ whose likelihood $P(t,x)$ we represent as
\begin{align}\notag
P(t,x) &= Pr(\mathcal{T}(x) > t) = 1 - \sum_{u=0}^{t-1} Pr(u < \mathcal{T}(x) \leq u+1)\\
	& \doteq 1-\sum_{u=0}^{t-1} q(u,x+u) , \label{eq:PfromQ}
\end{align}
where $q(u,x+u)$ is the mortality rate, i.e.~the probability of an individual aged $x+u$ dying in calendar year $u$ (alternatively interpreted as individual aged $x$ today dying  between ages $x+u$ and $x+u+1$). Integrating out idiosyncratic mortality experience (denoted by $Pr(\cdot)$ in \eqref{eq:PfromQ}), we focus on systematic factors that affect future evolution of mortality. Thus, we view $q(u,x+u)$ as a random variable that is driven by stochastic mortality factors $(Z_t)$, i.e.~$q(Z_t; t, x)$. This implies that the survival probability $P(Z_{[T,t]}; t,x)$ in \eqref{eq:PfromQ} depends on the whole \emph{path} $Z_{[T,t]}$ between the horizon $T$ and the summed years $t$.

Let $\beta_t$ be the instantaneous interest rate at time $t$, which also depends on (other) components of $(Z_t)$. The net present value of the annuity at $T$ (conditioning on $\mathcal{T}(x) > T$ and assuming that the longevity and interest rate risks are independent) is
\begin{align}
f(z) &=  -\E\left[\left. \sum_{t=S}^{x_u-x} e^{-\int_T^{t} \beta_u du}  P(Z_{[T,t]}; t,x) \right|Z_T=z, \mathcal{T}(x) > T \right]. 
\label{eq:PV(t)2}
\end{align}

For survival probabilities we choose the following model, known as (M7) in
Cairns et al.~\cite{cairns2011mortality}:
\begin{equation}
\logit q(Z_t;t,x) = \kappa^1_t + \kappa^2_t(x-\bar{x}) + \kappa^3_t \left((x-\bar{x})^2-\hat{\sigma}^2_x\right) + \gamma_{t-x}. \label{eq:q-CBD}
\end{equation}
In \eqref{eq:q-CBD}, the stochastic drivers are $\kappa^j_t,j=1,2,3$ as well as $\gamma_{t-x}$; the rest are fixed parameters. Specifically, $\bar{x}$ is the average age the model is fit to, and $\hat{\sigma}^2_x$ is the mean value over fitted ages $x$ of $(x-\bar{x})^2$, which are interpreted as \emph{age} effects. The model allows a plugin for age $x$, while the stochastic factors are the \emph{period} effects $\kappa^i_t$ which capture mortality evolution over calendar year, and $\gamma_{t-x}$ which is the \emph{cohort} effect. These constitute discrete-time ARIMA models, following the common choice that each $\kappa^i$ is a random walk with drift.  Typically, $\gamma_{\cdot}$ is fitted as an ARMA model.

 For reproducibility, we use the \texttt{R} package \texttt{StMoMo} \cite{villegas2015stmomo} which contains tools for fitting and simulating \eqref{eq:q-CBD}. Namely, we utilize the  England \& Wales (E\&W) mortality contained in the data and previously studied in Cairns et al.~\cite{cairns2011mortality}. Because the expression  \eqref{eq:PfromQ} for $P(Z_{[T,t]}; t,x)$ has Equation \eqref{eq:q-CBD} evaluated at $(u,x+u)$ for age/year, it results in a fixed cohort effect $\gamma_x$ fitted from historic values of the ARMA process.

The interest rate dynamics for $(\beta_t)$ are from \cite{chen1996stochastic}, defined through a three-factor Cox-Ingersoll-Ross model with stochastic volatility $\zeta_t$ and stochastic drift $\alpha_t$
\begin{align}\label{eq:chenmodel}
d\beta_t &= (\bar{\beta} - \alpha_t)dt + \sqrt{\beta_t} \zeta_t \, dW^\beta_t,\\
d\alpha_t &= (\bar{\alpha} - \alpha_t)dt + \sqrt{\alpha_t} \zeta_t \, dW^\alpha_t,\notag\\
d\zeta_t &= (\bar{\zeta} - \zeta_t)dt + \sqrt{\zeta_t} \varphi \, dW^\zeta_t,\notag
\end{align}
where $W^\beta, W^\alpha$ and $W^\zeta$ are independent standard Brownian motions.

 Overall, this implies a 6-D Markov state process  $Z_t = (\beta_t, \alpha_t, \zeta_t, \kappa^1_t, \kappa^2_t, \kappa^3_t)$. Thus, the procedure to obtain payoff realizations $Y^i(z_\cdot^{n,i})$ is to first simulate paths $(z_t^{n,i})_{t \geq T}$ from the distribution $(Z_t)_{t \geq T} | Z_T = z^n$, plug-in the realizations for the $(\kappa^j_t)_{t \geq T}$, $j=1, 2, 3$ into \eqref{eq:q-CBD} to evaluate \eqref{eq:PfromQ}, and finally use these survival probabilities as well as the simulated $(\beta_t)_{t \geq T}$ to compute
 \begin{align}
Y^i(z_\cdot^{n,i}) = \sum_{t=S}^{x_u-x}  e^{-\int_T^{t} \beta^{n,i}_u du}P(z_{[T,t]}^{n,i}; t, x).
\end{align}

 Note that the $\kappa$'s are discrete-time, while the interest rate model uses continuous $t$. For the latter, we use a simple forward Euler method with discretization $\Delta t = 0.1$. In all, one cashflow $Y(\cdot)$ takes approximately $0.01115$ seconds to evaluate (i.e.~about 2 minutes for the $10^4$ overall inner simulations), while it takes $0.000513$ seconds to evaluate for the first case study. Such a computationally expensive simulator is more representative of real-life simulation engines and reduces the impact of overhead costs in emulator fitting, selection, and prediction.

\subsection{Results}\label{sec:VA-results2}
For the remainder of this section, we consider the horizon $T=1$ and analyze $f(Z_1)$ in Equation \eqref{eq:PV(t)2}, the net present value of the life annuity one year into the future. We take $x=55$ and  annuity start date of $S=10$, when the individual retires at age $x=65$.   As before, we simulate $Z_1$ via Algorithm \ref{alg:initialFit} in Appendix~\ref{App:initial} to determine a scenario set $\cZ$ with $N=10^4$.
  The interest rate parameters in Equation \eqref{eq:chenmodel} are $\bar{\beta}=0.04, \bar{\alpha}=0.04, \bar{\zeta}=0.02, \phi=0.05$, with the E\&W mortality model fitted over the age range $x \in [55,89]$ using the \texttt{StMoMo} package \cite{villegas2015stmomo}.  As in the previous case study, the GP hyperparameters are refitted after stages $ 10, 20, \ldots, 100$, and after each round for A3-GP and U2-GP. For the mean function, we fit the constant $\mu(z) \equiv \beta_0$.

Under the setup \eqref{eq:PV(t)2}, there is no closed form evaluation for $f(z)$, so we obtain a benchmark $R_B$ through simulation---this value is determined by running the SR-GP approach with $\cN = 2\cdot 10^7$ using $K=200$ rounds (so that $\Delta r_k = 1 \cdot 10^5$; we also take a very conservative $L=1, U=200$).
The result is $R^{\VaR}_B = -16.0498$ and $R^{\TVaR}_{B}=-16.3739$, with estimator standard deviations of $s(R^{\VaR}_B)=0.00202$ and $s(R^{\TVaR}_B)=0.00188$.

We repeat the methods of Section~\ref{sec:casestudy-BS}, performing 100 macro replications with fixed $\cZ$ for both $\VaR_{0.005}$ and $\TVaR_{0.005}$.  The boxplots for bias and posterior uncertainty are reported in Figure~\ref{fig:VaRresults2}, and the numeric values for estimator standard deviation, RMSE, and final design sizes in Table~\ref{table:CS2results}.  

\begin{table}[!ht]
\centering
\begin{tabular}{l|rrrr || rrrr|r}
& \multicolumn{4}{c||}{$\VaR_{0.005}$} &  \multicolumn{4}{c}{$\TVaR_{0.005}$} \\
& $SD(\hat{R}^{HD}_K)$ & \multicolumn{1}{c}{$\overline{s}$} & RMSE & $|\cD_K|$ & $SD(\hat{R}_K)$ & \multicolumn{1}{c}{$\overline{s}$} & RMSE &  $|\cD_K|$ & \rev{Time} \\ \hline
ST-GP  &  0.0394  &  0.0455  &  0.0403  &  151.83 &  0.0461  &  0.0404  &  0.0472  &  147.10  & 330\\
SE-GP  &  0.0427  &  0.0493  &  0.0459  &  143.27  &  0.2853  &  0.2717  &  0.2850  &  101.62 & 295 \\
SV-GP  &  0.0382  &  0.0406  &  0.0380  &  497.81  &  0.0408  &  0.0393  &  0.0407  &  254.35 & 403  \\
SR-GP-1  &  0.0467  &  0.0470  &  0.0485  &  135.03 &  0.0430  &  0.0402  &  0.0430  &  184.73 & 219 \\
SR-GP-2  &  0.0391  &  0.0437  &  0.0434  &  217.36    &  0.0447  &  0.0431  &  0.0450  &  224.96 & 198 \\
A3-GP  &  0.0434  &  0.0497  &  0.0436  &  298.15   &  0.0464  &  0.0490  &  0.0461  &  298.55 & 115 \\
U2-GP  &  0.0598  &  0.0598  &  0.0596  &  194.03  &  0.0684  &  0.0601  &  0.0689  &  195.81 & 112 \\
U1-GP  &  0.5020  &  0.4156  &  0.5853  &  $10^4$ & 0.4940  &  0.4705  &  0.6709  &  $10^4$ & 177  \\
\end{tabular}
\caption{Results for the 6-D life annuity case study based on 100 macro-replications. We report sample standard deviation of $\hat{R}^{[1:100]}_K$, average GP posterior standard deviation $\overline{s}$, and RMSE of $\hat{R}_K$, as well as average final design size for each approach. \rev{The last column reports the running time in seconds.} Description of the methods is in Table~\ref{table:procedures1}. }\label{table:CS2results}
\end{table}

  \begin{figure}[!ht]
  \centering
  \begin{tabular}{m{0.09\textwidth}m{0.4\textwidth}m{0.4\textwidth}}
   & $\qquad\qquad$ Estimated $\hat{R}_K$  & $\quad$ GP posterior uncertainty $s(\hat{R}_K)$ \\
    $\;\VaR_{0.005}$   &  \includegraphics[scale=0.25]{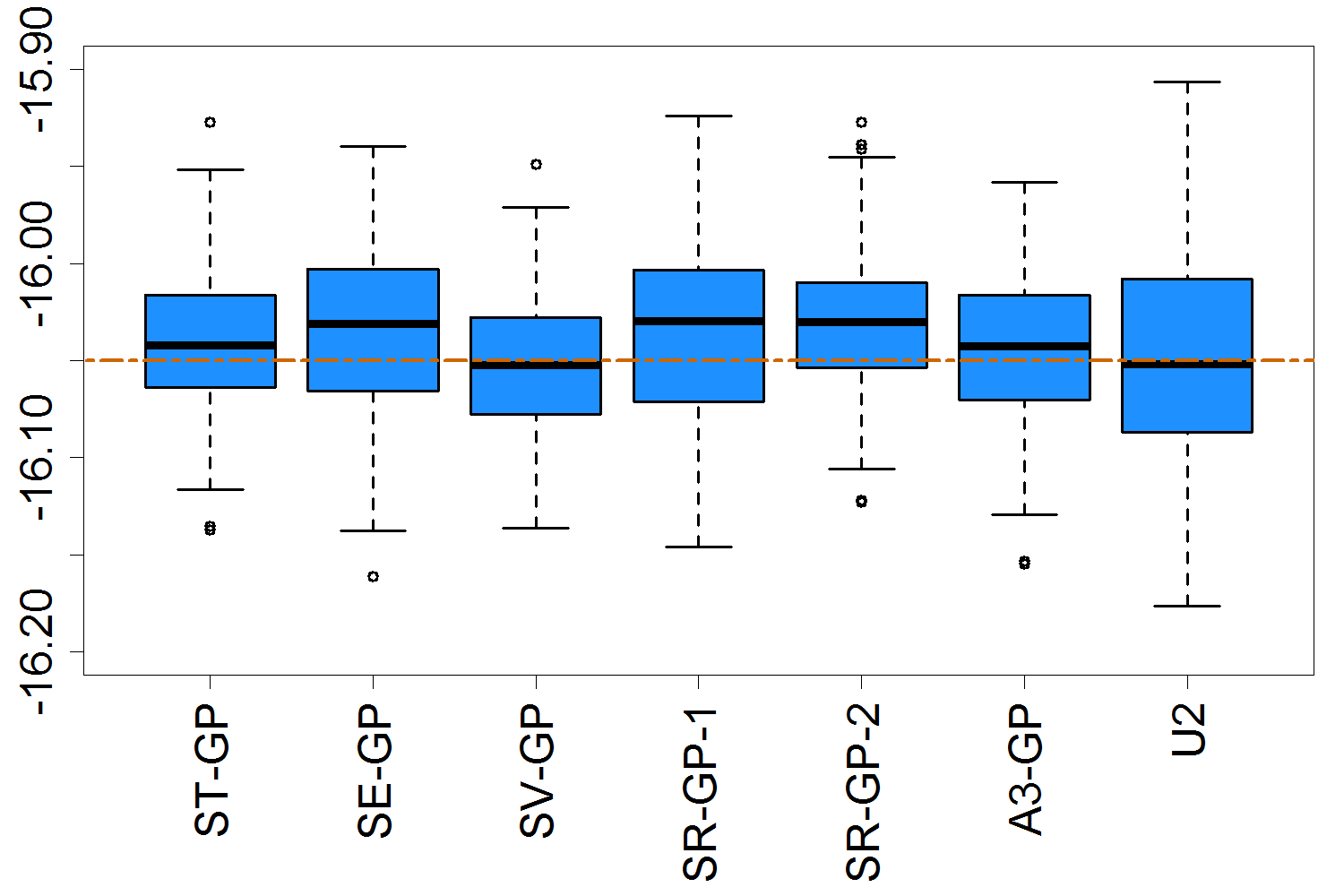} &
  \includegraphics[scale=0.25]{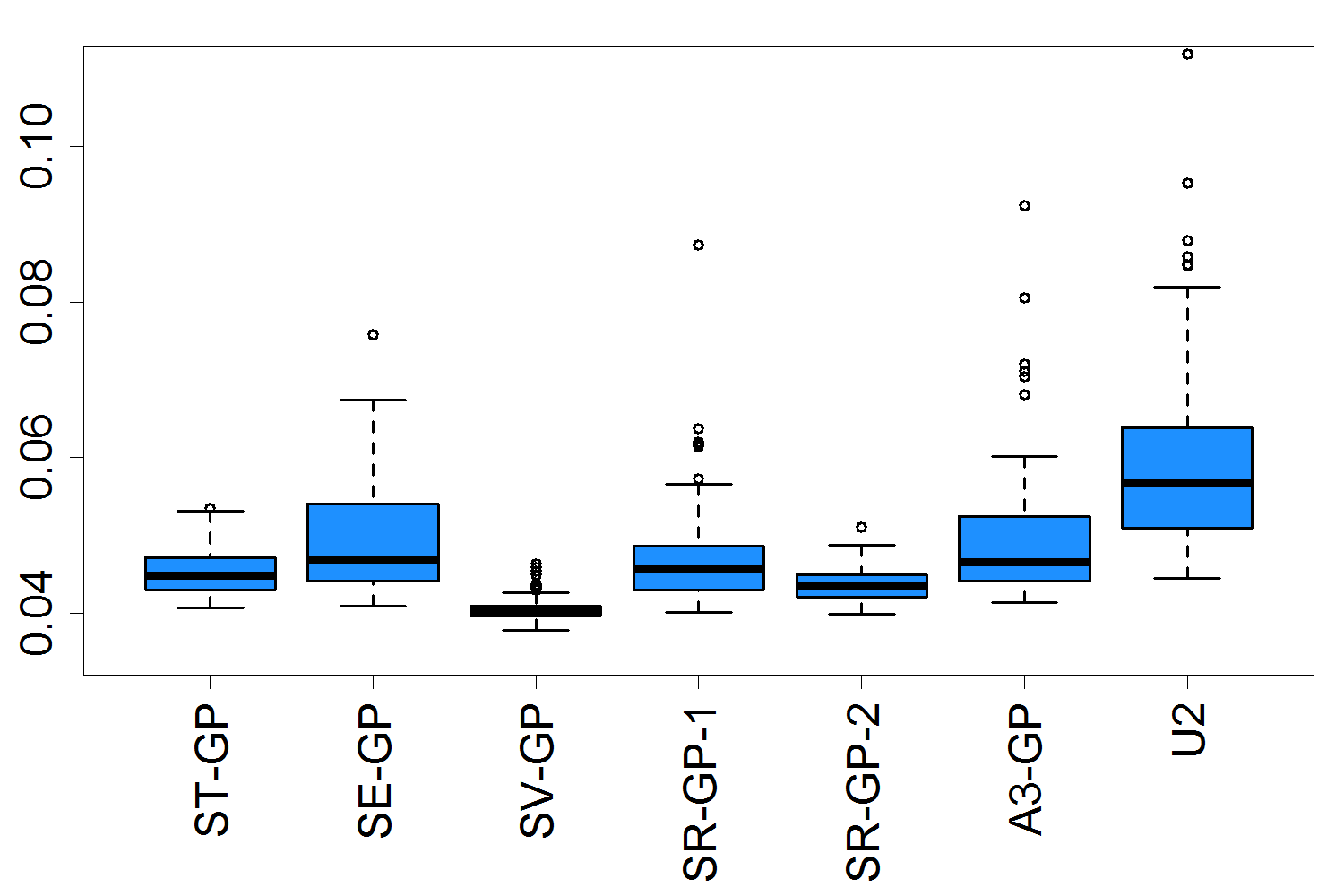} \\
   $\TVaR_{0.005}$   &  \includegraphics[scale=0.25]{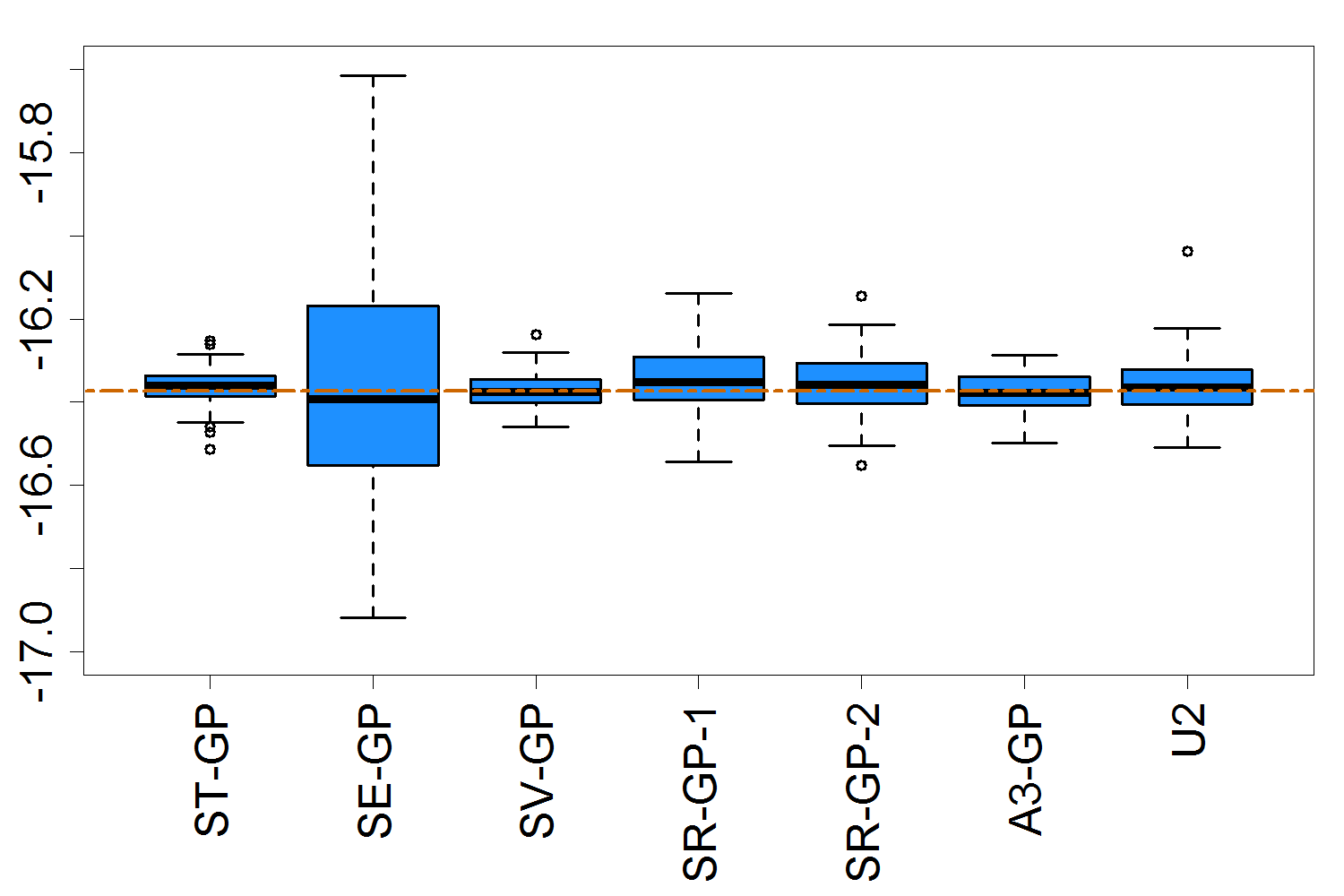} &
  \includegraphics[scale=0.25]{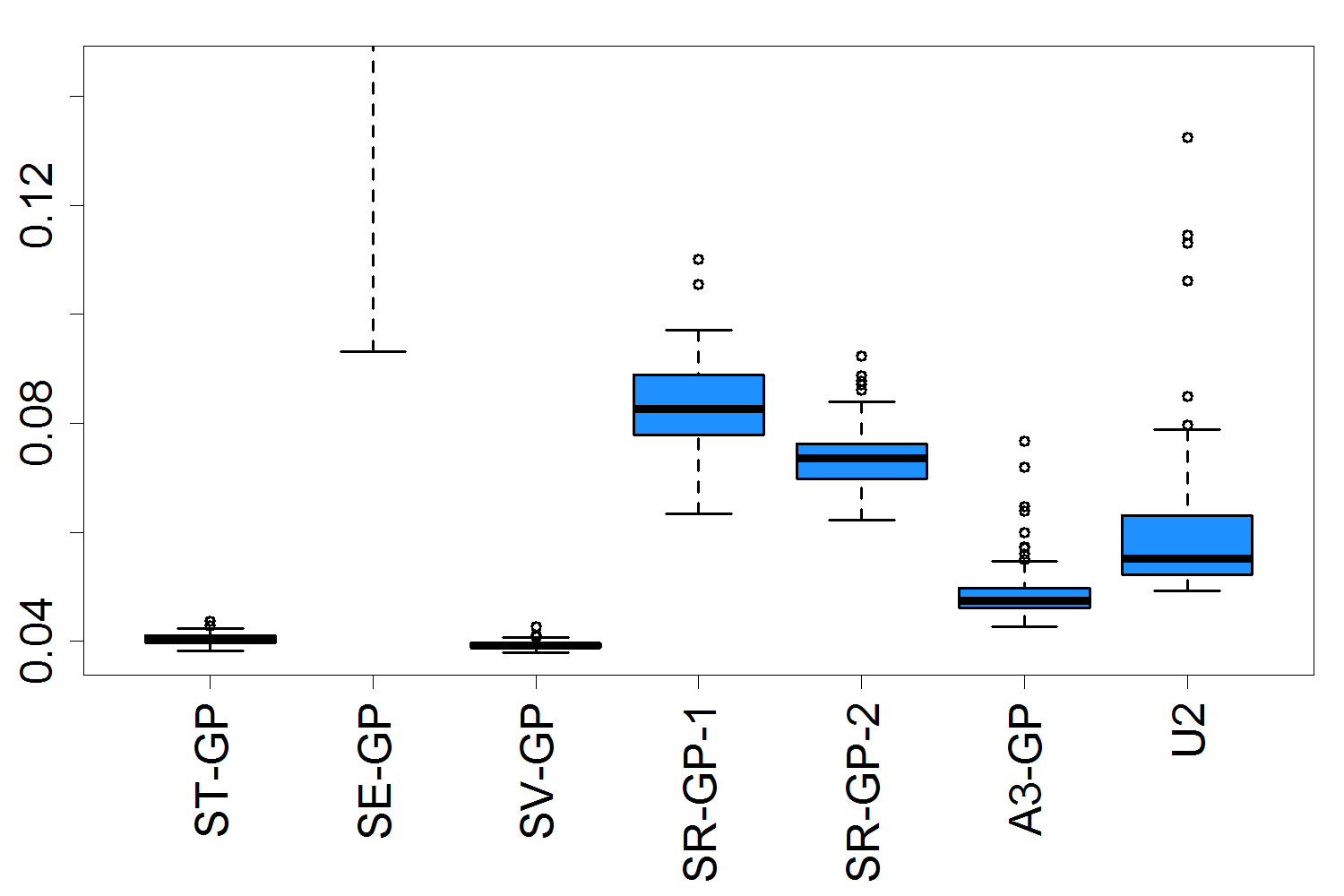}
  \end{tabular}
  \caption{For the life annuity case study, boxplots of final $\hat{R}_K$ estimates (left) and $s(\hat{R}_K)$  (right) over 100 macro replications for each approach. Top row: Value-at-Risk; bottom row: TVaR.
  The orange dotted lines on the left are the reference values of the risk measure.}\label{fig:VaRresults2}
\end{figure}

The results only slightly differ  from the first case study in terms of relative performance. For $\VaR_{0.005}$, SV-GP again has the lowest RMSE value, although ST-GP is only 6\% worse. Interestingly, A3-GP performs as well as SR-GP-2.
%
Since A3-GP allocates $\Delta r_1$ among 200 locations, this suggests that the bandwidth $L=26, U=75$ for SR-GP-2 may be too tight for this more complex setup, especially since this time it only slightly outperforms the  aggressive SR-GP-1.  As before, U1-GP fails (RMSE is more than 15x higher than for SV-GP), and while U2-GP offers a large improvement, it still has a much higher RMSE compared to the other methods.  The boxplots in Figure~\ref{fig:VaRresults2} show that the most stable uncertainty quantification (tight distribution of $s(\hat{R}^{HD}_K)$ and close to the actual $SD(\hat{R}_K))$ is provided by ST-GP, SV-GP and SR-GP-2. Like in the first case study, SE-GP leads to more across-run dispersion in both $\hat{R}_K$ and its reported standard error.

The $\TVaR$ results are similar; we comment on a few noticeable differences.  SV-GP still has the lowest RMSE and the lowest $SD(\hat{R}_K)$. At the other end, SE-GP shows the same issues with its acquisition function and unstable $s(\hat{R}_K)$.  The SR-GP methods do second best (although their uncertainty quantification is  significantly less reliable with unstable $s(\hat{R}_K)$), with A3-GP and ST-GP close behind. We attribute the relative ``win'' of SV-GP to its more diffuse allocation, i.e.~larger $|\cD_K|$, which indicates that estimating $\TVaR$ in a more complex example requires wider casting of the net, rather than narrowly  targeting the tail.
Overall, this case study confirms the performance gains of the GP-based sequential methods  even in a higher-dimensional setting.

\section{Conclusion}
We provided the first comprehensive examination of sequential methods for estimating $\VaR_\a$ and $\TVaR_\a$ using Gaussian process emulators and nested simulations over a fixed set of outer scenarios.  Existing approaches in the simulation literature fall short for this application,  whether it be assuming the level ($\VaR_\a$) is known \cite{picheny2010adaptive,bect2012sequential,chevalier2014fast} and/or non-noisy outputs \cite{labopin2016sequential}, using a non sequential algorithm \cite{liu2010stochastic}, or not using a spatial model for smoothing \cite{broadie2011efficient}.

We document a \emph{two orders-of-magnitude} savings relative to plain nested Monte Carlo, with about 10x savings thanks to adaptive allocation and 5x savings due to spatial modeling. These gains are strongest for small simulation budgets $\cN \approx N$ where existing approaches essentially fail or are not applicable. As a further advantage, employing a GP emulator also leads to essentially unbiased risk estimators (in contrast to biased unsmoothed versions) and provides a reliable estimate of standard error.  We also took advantage of the newly introduced heteroskedastic GP framework in \texttt{hetGP}
\cite{binois2016practical}. The improved estimation of the noise surface $\tau^2(\cdot)$ slashes the bias in $\hat{f}$ when  replicate counts $r_k$'s are low, and yields a more robust estimator for $\hat{R}_k$ throughout.

All fully sequential methods showed significant improvement over both the simple two-stage approach and three-stage method of \cite{liu2010stochastic}. The major downside to our combination of GP emulation and sequential design is the regression/prediction overhead, however this is reduced significantly by batching and carefully choosing candidate sets $\cZ^{cand}_k$.  In addition, the overall time spent computing acquisition functions, fitting and predicting becomes negligible in more realistic examples where calls to $Y(\cdot)$ are expensive.

\blue{
We have identified two top-performing approaches. SV-GP, the fully sequential adaptation of \cite{liu2010stochastic}, did best in terms of RMSE in both case studies. SV-GP also led to the smallest variation in the outputted posterior uncertainty $s_K(\hat{R}_K)$, allowing this value to be a reliable proxy for the true sampling standard deviation. ST-GP performed nearly as well, and in fact appeared to do best for the first few rounds, cf.~the fan plots of Figure~\ref{fig:VaRFanPlot}. The implementation of SV-GP involves a separate auxiliary optimization problem and brings more overhead and coding complexity. In contrast the criterion of ST-GP is straightforward and also keeps the emulator model size to a minimum (in terms of $|\cD_k|$ which is the key determinant for GP overhead). An even simpler/faster scheme is the rank-based SR-GP that also performed well, but its results depended heavily on the choices of $L$ and $U$, which in turn were varying among case studies and $\VaR_\a$ versus $\TVaR_\a$.  This is a further advantage of SV-GP and ST-GP that require no user-specified parameters. The mixed performance of SE-GP offers a cautionary tale that designing a good acquisition function is important for overall performance.  Our take-away is that ST-GP is an excellent choice when budget $\cN$ is really limited, while for larger budgets a scheme like SV-GP is best.}

\blue{There are multiple ways to move further.
The best-performing SV-GP could be further embellished through fine-tuning the weights $\mathbf{w}$ that are used during the respective look-ahead variance optimization. For example, the related three-stage approach in \cite{liu2010stochastic} uses an empirical weight estimate obtained through simulation in order to determine weights for the TVaR estimator. Modifying the optimization criterion to avoid the requirement that $r^n_k \ge 1$ would remove the scenario probing and drastically shrink SV-GP's design size $|\cD_k|$.  Alternatively, adaptive selection of $\Delta r_k$ would reduce emulation overhead.

 Conversely, one could extend the ST-GP or SE-GP approaches to adaptively optimize over both $z^{k+1}$ and $\Delta r_k$, i.e.~to internalize the batching into the design construction (whereas currently the respective acquisition functions are technically 1-step lookahead). Theoretical analysis regarding various convergence \emph{rates} of $\hat{R}_K$-bias would also be welcome, although it is already extremely challenging for the simpler problem of contour finding with a known threshold $L$.}

The initialization stage also warrants further analysis. Recall that we used a fixed budget of $\Delta r_0 = 0.1 \cN$, spent through $N_{init}$ representative scenarios determined via a space-filling scheme. In principle, $\Delta r_0$ should be taken as small as possible, letting the algorithm itself allocate the rest of the simulations, but in practice a decent starting point for the GP hyperparameters is crucial. Better rules/heuristics for the choice of $\Delta r_0$ remain to be considered. Also, rather than space-filling the given $\cZ$, it could be desirable to focus further on the respective extreme regions, in the sense of geometrically extremal points of $\cZ$, see e.g.~the use of statistical depth functions in \cite{chauvigny2011fast}.  

\rev{Last but not least, a variety of adjustments could be made to the GP emulators underlying our framework. Recall that the surrogate makes a number of assumptions, such as Gaussian observation noise and spatial stationarity in the covariance structure. These
 are likely to be violated in practice; if model mis-specification is a  concern, one could entertain numerous alternatives that constitute an active area of research within the GP ecosystem---GP-GLM (e.g.~with $t$-distributed noise), treed GPs and local GPs.}

\bibliographystyle{siam} 
\bibliography{VaRKriging}

\appendix

\section{Variance Minimization Calculations}\label{sec:varianceMinimization}

\begin{lemma}\label{lemma:inverseApprox}
We have the following approximation that improves as each $r_{k+1}^n$ increases:
\begin{align}\label{eq:inverseApprox}
(\bC + \bm{\Delta}^{cand}_{k+1})^{-1} \approx & (\bC + \bm{\Delta}_k)^{-1} + (\bC + \bm{\Delta}_k)^{-1} (\bm{\Delta}_k - \bm{\Delta}^{cand}_{k+1})    (\bC + \bm{\Delta}_k)^{-1}.
\end{align}
\end{lemma}

\begin{proof}
Recall that $\bm{\Delta}^{cand}_{k+1}$ is diagonal with entries $\tau^2(z^n)/(r_k^n + r_k'^n)$ and $\bm{\Delta}_{k}$ is diagonal with entries $\tau^2(z^n)/r_k^n$. We re-write
\begin{equation*}
\bm{\Delta}^{cand}_{k+1}= \bm{\Delta}_k + \bB_{k+1}(-\mathbf{I})\bB_{k+1},
\end{equation*}
where  $\bB_{k+1}$ is a diagonal matrix with elements $b_{nn} \doteq \tau(z^n) \sqrt{\left(\frac{1}{r^n_k} - \frac{1}{r^n_k+r_k'^n}\right)}$ and $\mathbf{I}$ is the identity matrix. By the Woodbury matrix inversion formula \cite{golub2012matrix} we then obtain
\begin{align}
(\bC + \bm{\Delta}^{cand}_{k+1})^{-1} &= (\bC + \bm{\Delta}_k + \bB_{k+1}(-\mathbf{I})\bB_{k+1})^{-1} \notag\\
	&= (\bC + \bm{\Delta}_k)^{-1}  \notag\\
	 -(\bC + \bm{\Delta}_k)^{-1} &  \bB_{k+1}\left(\bB_{k+1}(\bC + \bm{\Delta}_k)^{-1}\bB_{k+1} - \mathbf{I}   \right)^{-1}\bB_{k+1}   (\bC + \bm{\Delta}_k)^{-1}. \label{eq:woodbury1}
\end{align}
When $r^n_{k+1}$ is large, both $\bB_{k+1}$ and $\bm{\Delta}_k$ have relatively small entries, so that 
$\bB_{k+1}(\bC + \bm{\Delta}_k)^{-1}\bB_{k+1} \approx \mathbf{0}$.  Therefore, the middle matrix inverse in \eqref{eq:woodbury1} is $ \approx (-\mathbf{I})$ and
\begin{equation}
(\bC + \bm{\Delta}^{cand}_{k+1})^{-1} \approx  (\bC + \bm{\Delta}_k)^{-1} + (\bC + \bm{\Delta}_k)^{-1} \bB_{k+1}\bB_{k+1}   (\bC + \bm{\Delta}_k)^{-1}.
\end{equation}
Plugging in $\bB_{k+1}^2  = \bm{\Delta}_k - \bm{\Delta}^{cand}_{k+1},$ we have \eqref{eq:inverseApprox}.

Returning to \eqref{eq:minimizing-var2}, the optimization of the RHS is over $r^{'n}_k$'s which appear only inside $\bm{\Delta}^{cand}_{k+1}$. Plugging in \eqref{eq:inverseApprox} and dropping all terms that do not depend on $r^{'n}$, we obtain that minimizing  \eqref{eq:minimizing-var2} is equivalent to
\emph{maximizing}
\begin{align}
\hat{\mathbf{w}}_k \bC (\bC + \bm{\Delta}^{cand}_{k+1})^{-1} \bC \hat{\mathbf{w}}_k^T &= \hat{\mathbf{w}}_k  \bC (\bC + \bm{\Delta}_k)^{-1} (\bm{\Delta}_k - \bm{\Delta}^{cand}_{k+1})    (\bC + \bm{\Delta}_k)^{-1} \bC   \hat{\mathbf{w}}_k^T. \notag
\end{align}
After dropping the term with $\bm{\Delta}_k$ (which is still independent of $r^{'n}$) this reduces to \emph{minimizing} \eqref{eq:sv-gp}.
\end{proof}

\section{Generating Pilot Scenarios}\label{App:initial}

To construct $N_{init}$ pilot scenarios for the stage-0 design $\cD_0$ we first standardize the scenario set $\mathbf{z} = (z^1, \ldots, z^N),$ via
\begin{equation}
z_{std}^n \doteq \frac{z^n-\bar{\mu}_z^n}{\sigma_z^n}, \label{eq:zstd}
\end{equation}
where $\bar{\mu}_z^n$ and $\sigma_z^n$ are the coordinate-wise sample mean and standard deviation of $\mathbf{z}$. We next consider a parameter $d_0$ that determines the desired density of pilot scenarios. To pick the latter, we randomly permute the order of $\mathbf{z}$ and sequentially add new design locations whenever they  are at least $d_0$ (Euclidean) distance away from the current sites, continuing until all $z^n$'s are tried. While the resulting design size $ \le N_{init}$ depends on the random permutation, we found that practically it is only slightly affected by the randomization in the order of $\mathbf{z}$. If a specific size $N_{init}$ of $\cD_0$ is desired, one approach is to start with relatively large $d_0$ ($d_0 = 10\sqrt{d} N^{-1}_{init}$ is a good start), and sequentially reduce its value if the algorithm cannot find enough locations satisfying the distance requirement.

\begin{algorithm}
\begin{algorithmic}
\REQUIRE {$\mathbf{z}$, $N_{init}$, $d_0$}
\STATE Randomize the order of $\mathbf{z}$;
\STATE Compute $\mathbf{z}_{std}$ according to \eqref{eq:zstd};
\STATE Initialize $I \leftarrow \{1\}$; $j \leftarrow 2$;
\WHILE{$\text{card}(I) < N_{init}$}
  \IF{$\min_{i \in I} \| z_{std}^j,z_{std}^i\|_2 \geq d_0$ }
    \STATE $I \leftarrow I \cup \{j\} \qquad $    // Add $z^j$
  \ENDIF
  \STATE $j \leftarrow j+1$;
\ENDWHILE
\STATE Return pilot set $\cD_0 = \{ z^n : n \in I\}$ and $N_{init} = \text{card}(I)$\\
\end{algorithmic}
\caption{A space-filling design for initializing the emulator}\label{alg:initialFit}
\end{algorithm}

\end{document}